%% file: DpIPtoKE.tex
\documentclass[11pt]{article}

\usepackage[
pagebackref,backref=true,
hyperfootnotes=false,
colorlinks=true,
urlcolor=externallinkcolor,
linkcolor=internallinkcolor,
filecolor=externallinkcolor,
citecolor=internallinkcolor,
breaklinks=true,
pdfstartview=FitH,
pdfpagelayout=OneColumn]{hyperref}

\usepackage[numbers]{natbib} 

\usepackage{titlesec}
\usepackage{mathtools}
\usepackage{bbm}
\usepackage{amsthm,amsfonts,amsmath,amssymb,color,amsthm,boxedminipage,url,xparse,fullpage}
\definecolor{internallinkcolor}{rgb}{0,.5,0}
\usepackage{amsthm} 

\usepackage{xspace}
\usepackage{nicefrac}

\usepackage{cleveref}
\usepackage{nicefrac}

\input{macros}
\newcommand{\Dist}{\MathAlgX{Dist}}
\newcommand{\Rec}{\MathAlgX{Rec}}
\newcommand{\EveDP}{\MathAlgX{\Dist}}
\newcommand{\RecBit}{\MathAlgX{\Rec}}
\newcommand{\AlgEstimateBit}{\MathAlgX{EstBit}}

\newcommand{\Enote}[1]{\authnote{Eliad}{#1}}
\newcommand{\Inote}[1]{\authnote{Iftach}{#1}}

\newcommand{\Nnote}[1]{\authnote{Noam}{#1}}

\title{On the Complexity of Two-Party Differential Privacy}
\author{Iftach Haitner\thanks{The Blavatnik School of Computer Science at Tel-Aviv University. E-mail:\{\texttt{iftachh@tauex.tau.ac.il}, \texttt{noammaz@gmail.com}, \texttt{jadsilbak@gmail.com}, \texttt{eliadtsfadia@gmail.com}. Research supported by Israel Science Foundation grant   666/19 and  the Blavatnik Interdisciplinary Cyber Research Center at Tel-Aviv University.} \thanks{Member of the  Check Point Institute for Information Security. }
\and Noam Mazor$^\ast$
\and Jad Silbak$^\ast$
\and Eliad Tsfadia$^\ast$}

\begin{document}
\sloppy
\maketitle

\begin{abstract}
In \textit{distributed}  differential privacy, the parties perform analysis
over their joint data while preserving the privacy for \emph{both} datasets.  Interestingly, for a few fundamental two-party functions such as \textit{inner product} and \textit{Hamming distance}, the accuracy of the distributed solution lags way behind what is achievable in the \textit{client-server} setting. \citeauthor*{McGregorMPRTV10} [FOCS '10] proved that this gap is inherent, showing upper bounds on the accuracy of (any) distributed solution for these functions. These limitations can be bypassed when settling for \emph{computational} differential privacy, where the data is differentially private only in the eyes of a computationally bounded observer, using oblivious transfer.  

We prove that the use of public-key cryptography is \emph{necessary} for bypassing the limitation of \citeauthor{McGregorMPRTV10}, showing that a non-trivial solution for the inner-product, or the Hamming distance, \emph{implies} the existence of a key-agreement protocol. Our bound implies a combinatorial proof for the fact that non-Boolean inner product of independent (strong) \textit{Santha-Vazirani} sources is a good condenser. We obtain our main result by showing that the inner-product of a (single, strong) SV source with a uniformly random seed is a good condenser, even when the seed and source are \emph{dependent}.
\end{abstract}


\Tableofcontents

\input{Introduction}

\input{Technique}

\input{Preliminaries}

\input{The_KE_Protocol}

\input{CondensingSantaVazirani}

\input{HadamardRec}
\input{Key_agreement_ampl}


\subsection*{Acknowledgment}
We are grateful to Kobbi Nissim, Eran Omri and Ronen Shaltiel for very useful  discussions.

\bibliographystyle{abbrvnat}
\bibliography{crypto}

\appendix
\input{MissingProofs}

\end{document}

%% file: macros.tex
\providecommand{\remove}[1]{}

\ifdefined\IsDraft
    \newcommand{\authnote}[2]{{\bf [{\color{red} #1's Note:} {\color{blue} #2}]}}
\else
    \newcommand{\authnote}[2]{}
\fi

\titleclass{\subsubsubsection}{straight}[\subsection]

\newcounter{subsubsubsection}[subsubsection]
\renewcommand\thesubsubsubsection{\thesubsubsection.\arabic{subsubsubsection}}

\titleformat{\subsubsubsection}
{\normalfont\normalsize\bfseries}{\thesubsubsubsection}{1em}{}
\titlespacing*{\subsubsubsection}
{0pt}{3.25ex plus 1ex minus .2ex}{1.5ex plus .2ex}

\makeatletter
\renewcommand\paragraph{\@startsection{paragraph}{5}{\z@}%
	{3.25ex \@plus1ex \@minus.2ex}%
	{-1em}%
	{\normalfont\normalsize\bfseries}}
\renewcommand\subparagraph{\@startsection{subparagraph}{6}{\parindent}%
	{3.25ex \@plus1ex \@minus .2ex}%
	{-1em}%
	{\normalfont\normalsize\bfseries}}
\def\toclevel@subsubsubsection{4}
\def\toclevel@paragraph{5}
\def\toclevel@paragraph{6}
\def\l@subsubsubsection{\@dottedtocline{4}{7em}{4em}}
\def\l@paragraph{\@dottedtocline{5}{10em}{5em}}
\def\l@subparagraph{\@dottedtocline{6}{14em}{6em}}
\makeatother

\newenvironment{protocol}{\begin{mybox} \vspace{-.1in}\begin{proto}}{ \vspace{-.1in} \end{proto}\end{mybox}}

\newenvironment{algorithm}{\begin{mybox} \vspace{-.1in}\begin{algo}}{ \vspace{-.1in} \end{algo}\end{mybox}}

\newenvironment{mybox}{\begin{center}\begin{tabular}{|p{0.97\linewidth}|c|}   \hline} {  \\ \hline \end{tabular} \end{center}}

\let\originalleft\left
\let\originalright\right
\renewcommand{\left}{\mathopen{}\mathclose\bgroup\originalleft}
\renewcommand{\right}{\aftergroup\egroup\originalright}

\newcommand{\aka} {also known as,\xspace}
\newcommand{\resp}{resp.,\xspace}
\newcommand{\ie}  {i.e.,\xspace}
\newcommand{\eg}  {e.g.,\xspace}

\newcommand{\wrt} {with respect to\xspace}
\newcommand{\wlg} {without loss of generality\xspace}
\newcommand{\cf}{cf.,\xspace}

\newcommand{\ceil}[1]{\left\lceil #1 \right\rceil}
\newcommand{\ip}[1]{\iprod{#1}}
\newcommand{\iprod}[1]{\langle #1 \rangle}

\newcommand{\sset}[1]{\ens{#1}}
\newcommand{\set}[1]{\left\{#1\right\}}

\newcommand{\paren}[1]{\left(#1\right)}

\newcommand{\floor}[1]{\left \lfloor#1 \right \rfloor}

\newcommand{\half}{\tfrac{1}{2}}

\newcommand{\eqdef}{:=}

\newcommand{\N}{{\mathbb{N}}}
\newcommand{\Z}{{\mathbb Z}}

\newcommand{\zo}{\set{0,1}}
\newcommand{\oo}{\mo}

\newcommand{\mo}{\set{\!\shortminus 1,1\!}}
\newcommand{\mon}{\mo^n}

\newcommand{\zn}{{\zo^n}}
\newcommand{\zs}{{\zo^\ast}}

\newcommand{\condition}{\;\ifnum\currentgrouptype=16 \middle\fi|\;}

\newcommand{\xor}{\oplus}

\newcommand{\eps}{\varepsilon}

\newcommand{\from}{\leftarrow}
\newcommand{\la}{\gets}

\makeatletter
\@ifpackageloaded{complexity}{}{%
\usepackage[classfont=roman,langfont=roman,funcfont=roman]{complexity}}
\makeatother

\newcommand{\loglog}{\operatorname{loglog}}

\newcommand{\Var}{\operatorname{Var}}

\newcommand{\Hmin}{\operatorname{H_{\infty}}}

\newcommand{\Dec}{\operatorname{Dec}}

\newcommand{\negl}{\operatorname{neg}}

\newcommand{\Supp}{\operatorname{Supp}}

\newcommand{\argmax}{\operatorname*{argmax}}

\newcommand{\MathAlg}[1]{\mathsf{#1}}

\newcommand{\sign}{\MathAlg{sign}}

\newcommand{\Ensuremath}[1]{\ensuremath{#1}\xspace}

\newcommand{\MathAlgX}[1]{\Ensuremath{\MathAlg{#1}}}

\newcommand{\tth}[1]{\Ensuremath{#1^{\rm th}}}

\newcommand{\ith}{\tth{i}}
\newcommand{\jth}{\tth{j}}

\renewcommand{\cref}{\Cref}
\usepackage{aliascnt}

\newtheorem{theorem}{Theorem}[section]

\newaliascnt{lemma}{theorem}
\newtheorem{lemma}[lemma]{Lemma}
\aliascntresetthe{lemma}
\crefname{lemma}{Lemma}{Lemmas}

\newaliascnt{claim}{theorem}
\newtheorem{claim}[claim]{Claim}
\aliascntresetthe{claim}
\crefname{claim}{Claim}{Claims}

\newaliascnt{corollary}{theorem}
\newtheorem{corollary}[corollary]{Corollary}
\aliascntresetthe{corollary}
\crefname{corollary}{Corollary}{Corollaries}

\newaliascnt{construction}{theorem}

\aliascntresetthe{construction}
\crefname{construction}{Construction}{Constructions}

\newaliascnt{fact}{theorem}
\newtheorem{fact}[fact]{Fact}
\aliascntresetthe{fact}
\crefname{fact}{Fact}{Facts}

\newaliascnt{proposition}{theorem}
\newtheorem{proposition}[proposition]{Proposition}
\aliascntresetthe{proposition}
\crefname{proposition}{Proposition}{Propositions}

\newaliascnt{conjecture}{theorem}

\aliascntresetthe{conjecture}
\crefname{conjecture}{Conjecture}{Conjectures}

\newaliascnt{definition}{theorem}
\newtheorem{definition}[definition]{Definition}
\aliascntresetthe{definition}
\crefname{definition}{Definition}{Definitions}

\newaliascnt{notation}{theorem}

\aliascntresetthe{notation}
\crefname{notation}{Notation}{Notation}

\newaliascnt{assertion}{theorem}

\aliascntresetthe{assertion}
\crefname{assertion}{Assertion}{Assertion}

\newaliascnt{assumption}{theorem}

\aliascntresetthe{assumption}
\crefname{assumption}{Assumption}{Assumption}

\newaliascnt{remark}{theorem}
\newtheorem{remark}[remark]{Remark}
\aliascntresetthe{remark}
\crefname{remark}{Remark}{Remarks}

\newaliascnt{question}{theorem}
\newtheorem{question}[question]{Question}
\aliascntresetthe{question}
\crefname{question}{Question}{Questions}

\newaliascnt{example}{theorem}
\ifdefined\excludeexample
\else
   
\fi

\aliascntresetthe{example}
\crefname{exmaple}{Example}{Examples}

\crefname{equation}{Equation}{Equations}

\newaliascnt{proto}{theorem}

\newtheorem{proto}[proto]{Protocol}

\aliascntresetthe{proto}
\crefname{proto}{protocol}{protocols}

\newaliascnt{algo}{theorem}
\newtheorem{algo}[algo]{Algorithm}
\aliascntresetthe{algo}
\crefname{algo}{algorithm}{algorithms}

\newaliascnt{expr}{theorem}
\newtheorem{expr}[expr]{Experiment}
\aliascntresetthe{expr}
\crefname{experiment}{experiment}{experiments}

\newcommand{\reqref}[1]{Requirment~\ref{#1}}

\def\FullBox{$\Box$}
\def\qed{\ifmmode\qquad\FullBox\else{\unskip\nobreak\hfil
\penalty50\hskip1em\null\nobreak\hfil\FullBox
\parfillskip=0pt\finalhyphendemerits=0\endgraf}\fi}

\def\qedsketch{\ifmmode\Box\else{\unskip\nobreak\hfil
\penalty50\hskip1em\null\nobreak\hfil$\Box$
\parfillskip=0pt\finalhyphendemerits=0\endgraf}\fi}

\newenvironment{proofsketch}{\begin{trivlist} \item {\it
Proof sketch.}} {\qed\end{trivlist}}

{\proof[#1]\list{}{\leftmargin=1cm\rightmargin=1cm}}
{\endproof%
  \vspace{10pt}
}

\newcommand{\eex}[2]{\Ex_{#1}\left[#2\right]}
\newcommand{\ex}[1]{\Ex\left[#1\right]}
\newcommand{\Ex}{{\mathrm E}}
\renewcommand{\Pr}{{\mathrm {Pr}}}
\newcommand{\pr}[1]{\Pr\left[#1\right]}
\newcommand{\ppr}[2]{\Pr_{#1}\left[#2\right]}

\newcommand{\Ac}{\MathAlgX{A}}
\newcommand{\Mc}{\MathAlgX{M}}
\newcommand{\Ec}{\MathAlgX{E}}
\newcommand{\Gc}{\MathAlgX{G}}

\newcommand{\Bc}{\MathAlgX{B}}
\newcommand{\Cc}{\MathAlgX{C}}

\newcommand{\Fc}{\MathAlgX{F}}
\newcommand{\Dc}{\MathAlgX{D}}

\newcommand{\tT}{\widetilde{T}}

\newcommand{\cC}{{\mathcal{C}}}

\newcommand{\ens}[1]{\{#1\}}
\newcommand{\size}[1]{\left|#1\right|}

\newcommand{\prob}[1]{\mathsf{\textsc{#1}}}

\newcommand{\SD}{\prob{SD}}

\newcommand{\cM}{{\cal{M}}}

\newcommand{\I}{\mathcal{I}}

\newcommand{\pptm}{{\sc pptm}\xspace}
\newcommand{\ppt}{{\sc ppt}\xspace}

\def\cA{{\cal A}}
\def\cB{{\cal B}}
\def\cC{{\cal C}}
\def\cD{{\cal D}}

\def\cG{{\cal G}}
\def\cH{{\cal H}}
\def\cI{{\cal I}}

\def\cK{{\cal K}}

\def\cM{{\cal M}}
\def\cN{{\cal N}}

\def\cP{{\cal P}}

\def\cR{{\cal R}}
\def\cS{{\cal S}}
\def\cT{{\cal T}}

\def\cX{{\cal X}}
\def\cY{{\cal Y}}

\def\bbN{\N}

\def\bbR{{\mathbb R}}

\def\bbZ{\Z}

\newcommand{\isize}{n}

\newcommand{\Tableofcontents}{
\thispagestyle{empty}
\pagenumbering{gobble}
\clearpage
\tableofcontents
\thispagestyle{empty}
\clearpage
\pagenumbering{arabic}
}

\DeclareMathOperator{\V}{V}

\newcommand{\Pc}{\MathAlgX{P}}

\newcommand{\px} {x}
\newcommand{\py} {y}

\newcommand{\pw} {w}
\newcommand{\ptr} {r}
\newcommand{\rr} {r}

\newcommand{\zz} {z}

\newcommand{\hT} {\widehat{T}}

\newcommand{\OA} {O_\Ac}
\newcommand{\OB} {O_\Bc}

\newcommand{\htt} {\widehat{t}}
\newcommand{\hE} {\widehat{E}}
\newcommand{\hEc} {\widehat{\Ec}}

\newcommand{\NoAbort} {\mathsf{NoAbort}}

\newcommand{\out}{{\rm out}}
\newcommand{\tout}{\widetilde{\rm ou}{\rm t}}
\newcommand{\Out}{\mathsf {Out}}
\newcommand{\tOut}{\widetilde{\rm Ou}{\rm t}}

\newcommand{\Lap}{{\rm Lap}}

\newcommand{\est}{g}

\newcommand{\iseg}[1]{[[{#1}]]}

\makeatletter
\let\xx@thm\@thm
\AtBeginDocument{\let\@thm\xx@thm}
\makeatother

\DeclareMathSymbol{\shortminus}{\mathbin}{AMSa}{"39}

\newcommand{\hY}{\widehat{Y}}

\newcommand{\BuckA}{O_\Ac}
\newcommand{\BuckB}{O_\Bc}

\newcommand{\Ss}{\Sigma^\ast}

\newcommand{\Eve}{\Ec}

\newcommand{\hC}{\widehat{C}}

\newcommand{\CXYT}{C_{XYT}}

\newcommand{\CXYTk}{C_{X_{\kappa}Y_{\kappa}T_{\kappa}}}

\newcommand{\Ham}{\MathAlgX{Ham}}
\newcommand{\CDP}{{\sf CDP}\xspace}
\newcommand{\Dp}{{\sf DP}\xspace}
\renewcommand{\DP}{\Dp}

\newcommand{\SV}{\mathsf{SV}}

\newcommand{\CSV}{\mathsf{CSV}}

\newcommand{\sdist}[2]{\SD\paren{#1, #2}}

\newcommand{\hx}{\widehat{x}}
\newcommand{\hy}{\widehat{y}}
\newcommand{\hX}{\widehat{X}}
\newcommand{\tD}{\widetilde{\cD}}

\newcommand{\EveF}{f}
\newcommand{\hl}{{\widehat{\ell}}}

\newcommand{\hv}{{\widehat{v}}}

\newcommand{\transF}{r,x_{r^+},y_{r^-},t}
\newcommand{\transFp}{r,x'_{r^+},y'_{r^-},t}

\newcommand{\ceps}{c_\eps}
\newcommand{\flip}[1]{_{<#1>}}
\newcommand{\flipi}{\flip{i}}
\newcommand{\flipj}{\flip{j}}

\newcommand{\EveDPI}{\MathAlgX{\widetilde{\Dist}}}

\newcommand{\lland}{\text{ }\land\text{ }}
\newcommand{\tE}{\MathAlgX{\widetilde{\Ec}}}
\newcommand{\tPi}{\widetilde{\Pi}}
\newcommand{\tC}{\widetilde{C}}

\newcommand{\hAc}{\widehat{\Ac}}
\newcommand{\hBc}{\widehat{\Bc}}

\newcommand{\nfrac}{\nicefrac}
\newcommand{\hM}{\widehat{M}}

%% file: Introduction.tex
\section{Introduction}\label{sec:intro}
\textit{Differential privacy} aims to enable statistical analyses of databases while protecting individual-level information. A common model for database access is the \textit{client-server model}: a single server holds the entire database, performs a computation over it, and reveals the result. When the database contains sensitive information of individuals, the server should be restricted to only reveal the result of a differentially private function of the database. That is, a function that leaks very little information on any particular (single) individual from the database. 

\begin{definition}[Differential Privacy \cite{DMNS06}]\label{def:dpIntro}
	A randomized function (``mechanism'') $f$ is $(\eps,\delta)$-{\em differentially private}, denote $(\eps,\delta)$-\Dp, if for any two databases $\px,\px'$ that differ in one entry, and {\sf any} event $\cT$: 
	$$\Pr[f(\px)\in \cT]\leq e^{\eps}\cdot \Pr[f(\px')\in \cT]+\delta.$$ 
\end{definition}
For the sake of simplicity, in this section, we only focus on the case $\delta=0$, called {\em pure} differential privacy. 

In this work, we consider \emph{distributed}, two-party, database access: each party holds a private database, and they interact to perform data analysis over the \emph{joint} data.  Such interaction is differentially private, for short, \textit{two-party differential privacy} (\citet{DN04,BNO08}), if the parties perform the analysis while protecting the differential privacy of \emph{both} parts of the data. That is, each party's view of the protocol execution is a differentially private function of the other party's database (input).\footnote{More specifically, for a two-party protocol $\Pi=(\Ac,\Bc)$, let $V^\Pc(x,y)$ denote the view of a party $\Pc\in\set{\Ac,\Bc}$ in random execution of $\Pi(x,y)$. Then for every algorithm (distinguisher) $\Dc$, input $x\in \oo^n$ and  pair of inputs $y,y'\in\oo^n$ that differ in one entry, it should holds that
$\pr{\Dc(V^\Ac(x,y))=1}\le \pr{\Dc(V^\Ac (x,y'))=1}\cdot e^{\eps}+\delta.$ A similar constraint applies when considering the leakage from $V^\Bc$. A formal definition appears in \cref{sec:DP}.} 
 Motivated by the works of \citet{DN04} and \citet*{McGregorMPRTV10}, we focus on performing natural statistical analysis of the joint database. Specifically, the databases $\px = (x_1,\ldots,x_n)$ and $\py=(y_1,\ldots,y_n)$ are vectors in $\mon$ (\eg  each row $x_i$ is one if \ith individual smokes or not, and each row $y_i$ is one if it suffers from high blood pressure), and the desired functionality is to estimate their correlation (\eg to estimate the correlation between smoking and high blood pressure). The parties do that by estimating the \textit{inner product (\aka scalar product)} of the two (private) databases, \ie $\ip{\px,\py} = \sum_{i=1}^n x_i y_i$.\footnote{\citet{DN04} reduced a central data-mining problem (detecting correlations between two binary attributes) to approximating the inner product between two binary vectors. (\cite{DN04} considered databases over $\zo^n$, but there is a simple reduction between the $\oo$ case we consider here and the $\zo$ case.)}, or equivalently their \textit{Hamming distance}, \ie $\Ham(x,y) = \size{\set{i\colon x_i = y_i}}$. (Indeed,  $\ip{x,y}=n-2\cdot \Ham(x,y)$ for every $x,y\in\mon$).

The simplest $\eps$-\Dp protocol for estimating the inner product  is based on ``randomized response'': roughly, the party that holds $\px$, sends a randomized version $\hat{x}_i$ of each entry $x_i$ (where $\hat{x}_i$ is set to $x_i$ w.p. $(1+\eps)/2$ and to $-x_i$ otherwise), and the other party estimates the inner product based on $(\hat{x}_1,\ldots,\hat{x}_n)$ and $\py = (y_1,\ldots,y_n)$. This  protocol, however, induces an (expected) additive error of $\Omega\paren{\sqrt{n}/\eps}$ (\wrt the true value of $\ip{\px,\py}$). 
For comparison, in the standard client-server model where the server holds the entire database $\pw = (\px,\py)$, it is easy to achieve an accuracy of only $O(1/\eps)$.\footnote{The inner product over $\mon$ is a {\em sensitivity-$2$} function (\ie changing a single entry may only change the result by at most $2$). Therefore, a server that holds both $\px$ and $\py$ can simply compute $\ip{\px,\py}$, and output a (privacy-preserving) noisy estimation of it by adding a Laplace noise with standard deviation $2/\eps$.}  \citet{McGregorMPRTV10} proved that the large gap between the randomized response protocol and what is achievable in the client-server model this gap is unavoidable. Specifically, they showed that any two-party $\eps$-differentially private protocol for the inner product, must incur an additive error of $\Omega({\sqrt{n}}/{(e^\eps \cdot \log n}))$.\footnote{\citet{McGregorMPRTV10} proved it using a deterministic extraction approach, and showed that it can be extended to $(\eps,\delta)$-\Dp for $\delta = o(1/n)$. Using a different approach that explore connection between differentially private protocols and communication complexity,  \cite{McGregorMPRTV10} also proved a slightly stronger lower bound of $\Omega(\sqrt{n})$ for $\eps$-\Dp protocols for small enough constant $\eps$. The latter, however, does not extend to the approximate \Dp case (i.e., when $\delta > 0$).\label{fn:intro:McGregor}}

\paragraph{Computational Differential Privacy (\CDP).}
Motivated by the above limitations on multi-party differential privacy, \citet*{BNO08} and \citet*{MPRV09} considered protocols that only guarantee a computational analog of differential privacy.
Roughly, instead of requiring that each party's view is a differentially private function of the other party's input, it is only required that the output of any \emph{efficient} Boolean function over a party's view, is  differentially private (see \cref{sec:protocol} for a formal definition, and see \cite{MPRV09} for a broader discussion on computational differential privacy). With this relaxation, it is well known that assuming the existence of \textit{oblivious transfer} protocol, any efficient single party (i.e., client-server) \Dp mechanism can be emulated by a multi-party \CDP protocol (\eg \cite{BNO08,DKMMN06}). Specifically, the parties just need to perform a \textit{secure multi-party computation} for emulating the single-party mechanism. In particular, by emulating a (single-party) inner-product mechanism, we can obtain a multi-party \CDP protocol that is very accurate.

The above separation between computational and information-theoretic differential privacy has spawned an interesting research direction for understanding the complexity of computational differential privacy. In particular, \citet{Vadhan17} raised the following question:

\begin{question}[\cite{Vadhan17}]\label{ques:salil}
	What is the minimal complexity assumption needed to construct a computational task that can be solved by a computationally differentially private protocol, but is impossible
	to solve by an information-theoretically differentially private protocol?
\end{question}

Recent works have made progress on understanding this question for computing \emph{Boolean} functions, for example, showing that differential private protocol for computing the XOR function with non-trivial accuracy requires the existence of oblivious transfer \cite{HMSS19}. However, boolean functionalities, and in particular XOR, are less interesting in the context of $\Dp$ since even in the centralized model, the error of a $\Dp$ algorithm for estimating XOR must be close to half. 
In contrast, the inner-product, which is a much more natural functionality, has a much larger gap between the possible accuracy that is achievable with two-party $\Dp$ and $\CDP$. Much less progress has been made towards understanding the complexity of estimating such natural statistical tasks over large databases, and in this work, we make the first step towards filling this gap.
%
%
%
%
%
%


\subsection{Our Results}

\subsubsection{Differentially Private Two-Party Inner Product}
	
Our main result is that any (common output) computational differentially private protocol that estimates the inner product non-trivially, can be used to construct a key-agreement protocol.
\begin{theorem}[Main result, informal]\label{thm:intro:main}
	An $\eps$-\CDP two-party protocol that, for some $\ell \geq \log n$, estimates the inner product over $\oo^n \times \oo^n$ up to an additive error $\ell$ with probability \\$c\cdot e^{c\cdot \eps} \cdot \ell / \sqrt{n}$ (for some universal constant $c > 0$), can be used to construct a key-agreement protocol.
\end{theorem} 
\cref{thm:intro:main} extends to $(\eps,\delta)$-\CDP two-party protocols, for $\delta \leq 1/n^2$.
\cref{thm:intro:main} also extends to protocols whose accuracy guarantee only holds on average: over uniform inputs chosen by the parties, and it is tight (up to a constant) for this case: the trivial protocol that always outputs zero (which clearly cannot imply key-agreement) is with probability $\Theta(\ell/\sqrt{n})$ at distance at most $\ell$ from the inner product of two uniform vectors over $\oo^n$. A high-level proof of \cref{thm:intro:main} appears at \cref{sec:Technique}. 

Furthermore, \cref{thm:intro:main} also extends to the information theoretic settings: an (information theoretic) $\Dp$ protocol that accurately estimates the inner-product functionality, implies an information theoretically secure key agreement. Since the latter does not exist, it implies that such protocols do not exist either. Applying this result for $\eps=O(1)$ and $\ell =  \Theta(\sqrt{n})$, reproves (with slightly better parameters) the result of \cite{McGregorMPRTV10} regarding the in-existence of such protocols.\footnote{More specifically, \citet{McGregorMPRTV10} proved that for any $\beta > 0$, there exists no $\eps$-\Dp protocol that with probability $\beta$, estimates the inner product with additive error $O\paren{{\beta \sqrt{n}}/({e^{\eps}\cdot \log n})}$. For $\eps\in O(1)$ and $\beta \in \Omega\paren{{\log n}/{\sqrt{n}}}$, \cref{thm:intro:main} improves the result of \cite{McGregorMPRTV10} by a $\log n$ factor.}

Finally, \cref{thm:intro:main} also holds for a weaker notion of \CDP protocols, known as $\CDP$ {\em against external observer}: the (computational) privacy is guaranteed to hold only \wrt the transcript of the execution (and not necessarily \wrt the parties' view). Since the existence of a key-agreement protocol trivially implies a highly accurate \CDP against external observer protocol for estimating the inner product,\footnote{The parties can jointly emulate a single server functionality over an encrypted channel that they established.}
\cref{thm:intro:main} yields that the existence of such a non-trivial \CDP protocol is \emph{equivalent} to the existence of  key-agreement protocols.

\subsubsection{Condensers for Strong Santha-Vazirani Sources}
An additional contribution of our work is a new result about condensing strong Santha-Vazirani (SV) sources \cite{McGregorMPRTV10}. A random variable $X = (X_1,\ldots,X_n)$ over $\oo^n$ is called an {\em $\alpha$-strong SV source} if, for every $i$ and every fixing $\px_{-i}$ of $X_{-i} = (X_1,\ldots,X_{i-1},X_{i+1},\ldots,X_n)$, it holds that
\begin{align*}
\frac{\pr{X_i = 1 \mid X_{-i} = x_{-i}}}{\pr{X_i = -1 \mid X_{-i} = x_{-i}}} \in [\alpha, 1/\alpha]
\end{align*}
\citet{McGregorMPRTV10} proved their main result mentioned above, by showing that the inner product is a good \textit{two-source extractor} for (standard) SV sources.\footnote{In a (standard) SV source \cite{SV87}, each bit is somewhat unpredictable given only the {\em previous} bits (but not necessary given {\em all} other bits, as in strong SV).} Specifically, they proved that for any two independent SV sources $X$ and $Y$, the inner product $\ip{X,Y}$ modulo $m = \Theta(\sqrt{n}/\log n)$ is statistically close to the uniform distribution over $\bbZ_m$. We observe that, to some extent, the converse direction also holds: for every two independent strong SV sources $X,Y$: the \emph{nonexistence} of a \Dp-protocol for accurately estimating their inner product, implies that their inner product is a good two-source \emph{condenser}. Assume otherwise, then there exists $z\in \N$ such that $\pr{\ip{X,Y}=z}$ is large. Consider the two-party protocol $(\Ac,\Bc)$ in which $\Ac$ draws a random sample from $X$, party $\Bc$ draws a random sample from $Y$, and both parties output $z$ (regardless of their samples). By definition, this protocol is an accurate (information theoretic) \Dp-protocol for the inner-product functionality of the parties' samples.  

Equipped with the above observation, we use \cref{thm:intro:main} to deduce the following corollary:
\begin{corollary}[Inner product is a good condenser for strong SV sources, informal]\label{cor:intro:condenser}
	For any size $n$, independent, $e^{-\eps}$-{\em strong} SV sources $X$ and $Y$, it holds that $\Hmin(\iprod{X,Y}) \ge \log\left({\sqrt{n}}/{(c\cdot e^{c\cdot \eps} \cdot \log n)}\right)$ (for some universal constant $c > 0$).\footnote{A similar result holds for $\ip{X,Y} \mod c\cdot \sqrt{n}$, see \cref{cor:KAProtocol}}
\end{corollary}
In most aspects \cref{cor:intro:condenser} is weaker than the result of \citet{McGregorMPRTV10}: it only states that the inner product is a good condenser (and not extractor), does not hold for (standard) SV sources, and only holds when both sources remain hidden (\ie we did not prove  ``strong'' condenser). On the upside, our condenser has an {\em efficient black-box reconstruction algorithm}: given an oracle-access to an algorithm that predicts the value of $\iprod{X , Y}$ too well, the reconstruction algorithm violates the unpredictability  guarantee of the sources. (The result of \cite{McGregorMPRTV10},  proven  via Fourier analysis, does not yield a  reconstruction algorithm.)

In addition to \cref{cor:intro:condenser}, a key part for proving \cref{thm:intro:main} is showing that the inner product of a (single) strong SV source with a uniformly random seed is a good condenser, even when the seed and the source are \emph{dependent}.
\begin{theorem}[informal]\label{thm:intro:condenser}
	Let $W=(X,Y)$ be an $e^{-\eps}$-strong SV source, and let $R$ be a uniformly random seed over $\mon$. Then conditioned on the values of $R$, $X_{R^+} \eqdef \set{X_i}_{R_i = 1}$ and $Y_{R^-} \eqdef \set{Y_i}_{R_i = -1}$, it holds that $\Hmin(\iprod{X \cdot Y, R})\ge \log\paren{\frac{\sqrt{n}}{c\cdot e^{c\cdot \eps} \cdot \log n}}$ for some universal constant $c > 0$, letting $\cdot$ stand for coordinate/element-wise product.
\end{theorem}

We remark that when only conditioning on $R$ and $X_{R^+}$ (but not $Y_{R^-}$), the result of  \cref{thm:intro:condenser} is easy to prove:  $Y$ is a SV source conditioned on $X$, and thus by \cite{McGregorMPRTV10},  $\iprod{X \cdot Y, R}$ is a good extractor, conditioned on $R$ and $X$ (and thus a good condenser). The surprising part of \cref{thm:intro:condenser} is that the result holds also when conditioning also on the \emph{seed related} information $(X_{R^+}, Y_{R^-})$. \cref{thm:intro:condenser}  plays a critical role in the proof of our main result: in the key-agreement protocol we construct for proving \cref{thm:intro:main}, it is critical to expose these seed related values. We hope that such seed-related condensers will find further applications.

\paragraph{Computational Santha-Vazirani sources.}
Some of the  above results  extend to \emph{computational} Santha-Vazirani sources:  an ensemble of  random variables $\set{X^\kappa = (X^\kappa_1,\ldots,X^\kappa_{n(\kappa)})}_{\kappa \in \N}$ over $\oo^{n(\kappa)}$ is called a {\em computational $\alpha$-strong SV source}, if for every \ppt \Pc, every $\kappa \in \bbN$ and every $i \in  [n(\kappa)]$, it holds that
\begin{align*}
	\frac{\pr{\Pc(X^\kappa_{-i})  = X^\kappa_{i}}}{\pr{\Pc(X^\kappa_{-i})  = -X^\kappa_{i}}} \in  [\alpha(\kappa),1/\alpha(\kappa)] \pm \negl(\kappa)
\end{align*}
Namely, each entry $X_i$ of $X$ is somewhat unpredictable by a computationally bounded algorithm, even when all the other entries $X_{-i}$ are known.

Computationally unpredictable sources have an important role in the study of cryptography, most notably in constructions of pseudorandom generators. For instance,  \textit{next-block pseudo-entropy} \cite{HaitnerReVa13} quantifies the (average)   hardness  of efficiently predicting $X_i$ from $X_{<i} = X_1,\ldots,X_{i-1}$. Next-block pseudo-entropy is a key ingredient  in  modern   constructions of pseudorandom generators from one-way functions \cite{HaitnerReVa13,VadhanZh12}, but the lack  of efficient extraction tools for  such sources prevents pushing the efficiency of these constructions even further.\footnote{Current  technique   apply a seeded extractor  entry-by-entry, on the direct product of the  source.} In contrast,     \cref{thm:intro:condenser}, which is proven via an efficient reconstruction algorithm, yields that the inner product is a good condenser  for computational SV sources.\footnote{The reconstruction algorithm of \citet{dinur2003revealing} implies that it is hard to approximate the inner-product of computational SV source and an uniformly chosen vector. Their result, however,  fails short (in terms of the approximation needed) to imply that the inner product is a good condenser for such sources.}  \footnote{Like in the information theoretic case, the  inner-product remains a good condenser  also when conditioning on some seed related information.}


Interestingly, we do not know whether  \cref{cor:intro:condenser} extends to  computational  SV sources (even if we require the sources to be efficiently samplable). In particular, trying to adjust the proof of \cref{cor:intro:condenser} for the computational  settings requires proving that  there exists \emph{no}  pair of efficiently samplable computational  SV sources such the the following protocol is a (weak) key agreement: each party samples from one of these sources,  and then the parties interact, with these samples as private inputs, in the protocol we introduce  for proving \cref{thm:intro:main} (see \cref{prot:Technique:main}).

\subsubsection{Reconstruction Attacks}
Another contribution of our work regards  revealing   linear statistics of a databases under differential privacy. Given a database $z = (z_1,\ldots,z_n) \in \oo^n$, you would like to reveal an estimation $\Fc_z(r)$ of $\ip{z,r} \in \bbZ$, for all $r \in \oo^n$, while preserving differential privacy. Such an estimation is $(\ell,\beta)$-accurate if $\ppr{r \la \oo^n}{\size{\Fc_z(r)-\ip{z,r}} \leq \ell} \geq \beta$. (I.e., $\Fc_z(r)$ is with additive distance at most $\ell$ for at least $\beta$ fraction of the $r$'s, and otherwise is unrestricted.)  For utility, we would like to decrease $\ell$ and increase $\beta$ as possible. The question is, in what regimes of $\ell$ and $\beta$, an  $(\ell,\beta)$-accurate estimation violates the differential privacy of $z$?

\citet{dinur2003revealing,DY08,DMT07} have shown that, for certain regimes, if $\Fc_z$ is $(\ell,\beta)$-accurate, then revealing it is  \emph{blatantly non-private} \cite{dinur2003revealing}: there exists an efficient attack that given oracle access to $\Fc_z$, compute (with high probability) a database $z'\in \oo^n$ that differ from $z$ by at most $0.1n$ coordinates, which clearly violates the $(1,0.1)$-differential privacy of $z$.
However, since the above attacks aim to show blatantly non-privacy, they inherently  fail on the \textit{low-confidence regime}: $\beta = 0.01$ or even a sub-constant, and this holds even  when the additive error $\ell$ is very small.\footnote{Even inefficient attacks cannot reconstruct a close database $z'$ with high probability when $\beta \leq 1/2$. For instance, this cannot be done in the case that $\Fc_z$ output $\ip{z,r}$ for half of the $r$'s, and $\ip{-z,r}$ for the other half of the $r$'s.}

We overcome this barrier  by showing that ``non trivial'' statistics in the low confidence regime suffice for efficiently violating  differential privacy. 


\begin{theorem}[Tight reconstruction attacks, informal]\label{thm:reconstruction:intro}
	For every $\ell \in \bbN$, an $\paren{\ell, \beta=300\ell/\sqrt{n}}$-accurate estimator $\Fc_z$ is {\sf not} $(1.0.1)$-differentially private. 
	The proof is constructive: there exists a \ppt algorithm $\Rec$ that for every database $z \in \oo^n$ and every oracle access to an $\paren{\ell, 300\ell/\sqrt{n}}$-accurate estimator $\Fc_z$, for at least $0.9$ fraction of the $i \in [n]$ it holds that $\Rec^{\Fc_z}(i,z_{-i}) = z_i$ with high probability. $\Rec$ uses $\widetilde{O}(n^3)$ queries to $\Fc_z$.
\end{theorem}

In particular, if we start with a uniformly random database $Z=(Z_1,\ldots,Z_n)$, then \cref{thm:reconstruction:intro} implies that there exists $i \in [n]$ such that $\pr{\Rec^{\Fc_Z}(i,Z_{-i})=Z_i} \geq 0.9$. This yields that given an access to $\Fc_Z$, $Z$ is not a strong SV source, and therefore $\Fc_Z$ is not differentially private.


Note that the trivial estimation $\Fc_z(r) = 0$ for all $r \in \oo^n$ is $\paren{\ell, \beta=\Omega(\ell/\sqrt{n})}$-accurate for every $\ell \geq 0$. By \cref{thm:reconstruction:intro} we deduce that up to a constant factor in the confidence, one cannot do anything better while preserving $\Dp$, or even $\CDP$ (since the proof is constructive).





\subsection{Perspective: Hardness Hierarchy}
Understanding the inter-connection between the different primitives and hardness assumptions is a fundamental task in the study of computational complexity, and in particular of complexity-based cryptography. Such understanding can be achieved by \textit{oracle separations/black-box impossibilities}: prove that a primitive cannot be constructed from a second one in a certain way, \eg key agreement cannot be constructed in a black-box way from one-way functions, \citet{ImpagliazzoRu89}. In other cases, we enrich our knowledge via \emph{reductions}: use one primitive to construct a second one, \eg one-way functions imply pseudorandom generators, \citet{HastadImLeLu99}. Such reductions enable us to base a complex primitive on a more basic and trustworthy one, but they also serve as lower bounds: they imply that the primitive you started with is at least as complicated as the constructed one, \eg coin flipping imply one-way functions \cite{HaitnerOmri14,BermanHT18}. Finding reductions gets rather challenging when the primitive you start with is less structured than the one you are trying to build. Nevertheless, a sequence of celebrated works showed that the very unstructured form of hardness guaranteed by one-way functions, suffices to construct rather complex and structured primitives such  as pseudorandom generators \cite{HastadImLeLu99}, pseudorandom functions \cite{GoldreichGoMi86} and permutations \cite{LubyR88}, commitment schemes \cite{Naor1991,HaitnerNgOnReVa09}, universal one-way hash functions \cite{Rompel90}, zero-knowledge proofs \cite{GoldreichMiWi87}, and more. Such reductions, however, are much less common outside the one-way functions regime, most notably, when the primitive to construct is a \emph{public-key} one.

Public-key cryptography, in its broader sense, is all about creating correlation, \ie mutual
information, between the parties’ outputs, which  is hidden from an external or internal observer. 

So when trying to use  a less  structured two-party functionality $f$  to construct a key agreement, for instance,    the challenge is to purify the correlation   induced  by the call(s) to $f$, into the one required by a   key agreement.   If the less  structured $f$ is a single-bit  input functionality,  the typically constant   amount of correlation a call to $f$ induces, is  distributed between the two input bits. This  makes, at least in some settings, purifying/extracting the correlation a feasible task, see examples in
\cref{sec:intro:Related Work}. But handling longer input functionalities is much more challenging. First, the “per bit” correlation is much smaller, \eg  the per-bit correlation induced by an accurate \Dp inner-product functionality is only $O(\log n/n)$, and  most bits might have no correlation at all. Moreover, \emph{efficiently} extracting correlation from {super-polynomial} domain size variables might get extremely challenging. For example, any non-trivial \textit{channel} implies oblivious transfer \cite{nascimento2008oblivious}, but the running time of the induced  oblivious transfer   is  proportional to the channel domain size.

\subsection{Additional Related Work on Computational Differential Privacy}\label{sec:intro:Related Work}
There are two natural approaches for defining computational differential privacy. 
The more relaxed and common one is  the \emph{indistinguishably-based} definition, which restricts the distinguishing event, $\cT$ in \cref{def:dpIntro}, to \emph{computationally identified} events. The second approach is the \emph{simulation}-based definition, which asserts that the output of the mechanism $f$ is computationally close to that of an (information-theoretic) differentially private mechanism. Relations between these (and other) notions are given in \cite{MPRV09}. We remind that our reduction from key-agreement to \CDP holds even when assuming \CDP against external observer, which is weaker than the notions consider in \cite{MPRV09}. See \cref{sec:DP} for the formal definition and comparison to the standard notions. 

For the single-party case (\ie  the client-server model), computational and information-theoretic differential privacy seem closer in power. Indeed, \citet*{GKY11} showed that a wide range of \CDP mechanisms can be converted into an (information-theoretic) $\Dp$ mechanism.  \citet*{BunCV16} showed that under (unnatural) cryptographic assumptions, there exists a (single-party) task that can be efficiently solved using \CDP, but is infeasible (not impossible) for information-theoretic $\Dp$. Yet, the existence of a stronger separation (\ie one that implies the impossibility for information-theoretic $\Dp$) remains open (in particular, under more standard cryptographic assumptions).

Another extreme (and very applicable) scenario is the \emph{local model}, in which each of the, typically many, parties holds a single element. Usually, information-theoretic $\Dp$ protocols for this model are based on randomized response. Indeed, \citet*{chan2012optimal}  proved that randomized-response is optimal for any counting functionality (and in particular, inner product). In contrast, {local} \CDP protocols can emulate \emph{any} efficient (single party) mechanism using secure multiparty computation (MPC), yielding a separation between the $\CDP$ and $\Dp$ notions. 

So the main challenge is understanding the complexity of \CDP protocols in the two-party (or ``few'' party) case. Most works made progress on the Boolean case, where each party holds one (sensitive) bit, and the goal is to privately estimate a boolean function over the bits (\eg the XOR).  \citet*{GMPS13} demonstrated a constant gap between the maximal achievable accuracy in the client-server and distributed settings for any non-trivial boolean functionality, and showed that any \CDP protocol that breaks this gap implies the existence of one-way functions. \citet*{GKMPS16} showed that the existence of an accurate enough \CDP protocol for the XOR function implies the existence of an oblivious transfer protocol. \citet*{HNOSS20} showed that any non-trivial $\eps$-\CDP two-party protocol for the XOR functionality, implies an (infinitely-often) key agreement protocol. Recently,  \citet*{HMSS19} improved the results of \cite{GKMPS16,HNOSS20}, showing that any non-trivial \CDP two-party protocol for XOR  implies oblivious transfer.

In contrast to the study of Boolean functionalities, understanding the complexity of \CDP two-party protocols for more natural tasks (\ie low-sensitivity many-bits functionalities, such as the inner product) remains (almost) completely open. The only exception is the result of \citet*{haitner2016limits}, who applied their generic reduction on the impossibility result of \citet{MPRV09},  to deduce that accurate \CDP protocol for the inner product does not exist in the \textit{random oracle model} (and thus such protocol cannot be constructed in a fully black-box way from a symmetric-key primitive).  

\subsection{Open Questions}
In this work, we make progress towards understanding the complexity of \CDP protocols for estimating the inner-product functionality. The main challenge is to extend this understanding to other \CDP distributed computations. For some functionalities, \eg Hamming distance, we have a simple reduction to the inner-product functionality. But finding a more general characterization that captures more (or even all) functionalities, remains open. 

Another important question is to determine the {\em minimal} complexity assumption required for constructing a non-trivial \CDP for the inner-product functionality. In this work, we answer this question with respect to the weaker notion of \CDP against external observer (showing that Key-agreement is necessary and sufficient). It is still open, however, whether oblivious transfer is the right answer for \CDP protocols for the inner product, achieving the standard (stronger) notion of differential privacy (and doing the same for other functions as well).

\subsection{Paper Organization}
In \cref{sec:Technique}, we give a high-level proof of \cref{thm:intro:main}. Notations, definitions and general statements used throughout the paper are given in \cref{sec:Preliminaries}. 
Our key-agreement protocol and its security proof (i.e., the proof of \cref{thm:intro:main}), and also the proof of \cref{cor:intro:condenser}, are given in \cref{sec:KAProtocol}. The proof given in \cref{sec:KAProtocol} relies on technical tools that are proven in \cref{sec:CondensingSV,sec:reconstruction,sec:KAApmlification}. \cref{thm:reconstruction:intro} is proven in \cref{sec:reconstruction}.
\cref{thm:intro:condenser} is restated in \cref{sec:CondensingSV} and proven in \cref{sec:appendix}, which also contains the other missing proofs.

%% file: Technique.tex
\newcommand{\hf}{\widehat{f}}
\newcommand{\tuple}{s}
\newcommand{\Tuple}{S}
\newcommand{\es}{e}

\section{Our Technique}\label{sec:Technique}
In this section, we provide a rather elaborate description of our proof technique. In \cref{sec:Technique:Accurate} we consider an easy variant of \cref{thm:intro:main} where the protocol computes the inner-product {very} accurately.  In \cref{sec:Technique:hardCase}, we discuss the much more challenging case of slightly accurate protocols.

\subsection{Highly Accurate Protocols}\label{sec:Technique:Accurate}
We show how to construct a key-agreement protocol from an (external observer) $\eps$-$\CDP$ protocol $\Gamma$ (\ie $\eps$-differentially private against computationally bounded adversaries) that  almost always computes the inner-product functionality with an additive error smaller than $\sqrt{n}$. That is, 
\begin{align}\label{eq:Technique:1}
\pr{\size{\Out-\ip{X,Y}} \le \sqrt{n}/c} \ge  1 - 1/n^4
\end{align}
for large enough constant $c>0$, where $(X,Y)\gets (\mon)^2$, and $\Out$ is the common output of $\Gamma(X,Y)$ (part of the transcript). As noted by \citet{McGregorMPRTV10}, if  $\Gamma$ would have been  $\eps$-$\DP$ (\ie against computationally \emph{unbounded} adversaries), then conditioned on the (common) transcript $T$, it holds that $X$ and $Y$ are (independent) $e^{-\eps}$-strong SV sources. \citet{McGregorMPRTV10}  proved 
that $\ip{X,Y}$, the (non boolean) inner product of $X$ and $Y$, has min-entropy $\approx \log \paren{\sqrt{n}}$\remove{, and this holds even when revealing $X$ or $Y$}. By that, they concluded that the expected distance between $\Out$ (which is a function of $T$)  and   $\ip{X,Y}$, is $\sqrt{n}$,  in contradiction to the accuracy of $\Gamma$.\footnote{Actually, the argument of \cite{McGregorMPRTV10} fails short of contradicting the accuracy stated in \cref{eq:Technique:1}, and only contradicts $\pr{\size{\Out-\ip{X,Y}} \le \sqrt{n}/\polylog(n)} \approx 1$.}

However, since we only assume that $\Gamma$ is $\eps$-$\CDP$, it is no longer true that $X$ and $Y$ are $e^{-\eps}$-Santha-Vazirani sources. Indeed,  assuming the existence of oblivious transfer, there exists an accurate protocol $\Gamma$ for which the inner product of $X$ and $Y$ has \emph{tiny} min-entropy given $T$ (\ie $\log (1/\eps)$).   Yet, we prove, and this is our main technical contribution, that a randomized inner product of  $X$ and $Y$, \ie $\ip{X\cdot Y,R}$, where $\cdot$ stands for coordinate-wise product and  $R$ is a random seed in $\zn$,  does have high-min-entropy in the \emph{eyes of a computationally bounded observer}, which only sees the transcript $T$ and the seed $R$. Not only that, the inner product remains hidden (\ie have large min-entropy), even when some seed-related information about  $X$ and $Y$ leaks to the observer.  We exploit this observation to construct the following  ``weak'' key-agreement protocol.

\begin{protocol}[$\Pi = (\Ac,\Bc)$]\label{prot:Technique:main}
	\item Parameter: $1^n$.
	\item Operation:
	\begin{enumerate}
		\item  $\Ac$ samples $x\gets \mon$, and   $\Bc$ samples $y\gets \mon$.
		
		\item The parties interact in $\Gamma(x,y)$.  Let $\out$ be the common output.

		\item  $\Ac$ samples  $\rr \la \zo^n$, and sends $(\rr,\px_{\rr} = \set{x_i \colon \rr_i =1})$ to $\Bc$.
		
		
		\item  $\Bc$ sends $\py_{-\rr}=  \set{y_i \colon \rr_i =0}$ to $\Ac$.

		\item  $\Ac$ (locally) outputs  $(\out- \ip{\px_{-\rr}, \py_{-\rr}})$,
		and \Bc (locally) outputs $ \ip{\px_{\rr}, \py_{\rr}}$.

	\end{enumerate}
\end{protocol}

That is, $\Ac$ uses its knowledge of $\px_{-\rr}$, and the estimation of  $\ip{x,y}$ given by the execution of $\Gamma$, to estimate $\Bc$'s output ($ \ip{\px_{\rr}, \py_{\rr}}$). 

Let $X^n,Y^n,\Out^n$ and $R^n$, be the values of $x,y,\out$ and $r$ in a random execution of $\Pi(1^n)$. Let $T^n$ be the transcript of $\Gamma$ in this execution, and let  $\Out_\Ac^n,\Out_\Bc^n$ be the parties local output.  \cref{eq:Technique:1} immediately yields that 
\begin{align}\label{eq:Technique:2}
&\text{Agreement:}  &\pr{\size{\Out^n_\Ac - \Out^n_\Bc} < \sqrt{n}/c}\ge  1 - 1/n^4
\end{align}
The crux of the proof, and its most technical part, is showing that the computational  differential privacy of $\Gamma$ yields that    no  \ppt  $\Ec$ can estimate $\Out^n_\Bc$ ``too well'':
\begin{align}\label{eq:Technique:3}
&\text{Secrecy:}  	&\pr{\size{\Ec(1^n,X_R^n,Y_{-R}^n,T^n,R^n)-\ip{X_R^n,Y^n_R}}< \sqrt{n}/c} < 1 - 3/n^4
\end{align}
Combining \cref{eq:Technique:2,eq:Technique:3},  yields that $\Pi$ enjoys a gap between the ``agreement'' and ``secrecy'', of the parties' local output. With some technical work, such a gap can be amplified to get a full-fledged key-agreement protocol. Parts of this amplification part are described as an independent result in \cref{sec:KAApmlification}.  For the sake of this section, however, we focus only on the proof of \cref{eq:Technique:3}.\footnote{It is instructive to note that if $\Ec$ has access only to $(X_R,T,R)$ (and even to all of $X$, and not just $X_R$),  then \cref{eq:Technique:3} would have easily followed by the fact that the inner product, with a random seed,  is a strong extractor for  SV sources. Actually, the above  argument requires that $\Gamma$ is \textit{simulation-based computational differential private}:  $X,Y|T$ is computationally indistinguishable from $X',Y'|T$ for $(X',Y',T)$ that is (information theoretic) differentially private. (A stronger notion of privacy that is not known to be implied by the notion we consider here.) What makes proving \cref{eq:Technique:3} challenging, is that \Ec has also access to $Y_{-R}$, an information that is \emph{dependent} on the seed $R$. Arguing about the entropy of an extractor's output in the face of such ``seed dependent'' leakage is typically a non-trivial task.} Hereafter, we omit $n$ when clear from the context.

Assume towards a contradiction that there exists a \ppt  $\Ec$ that violates \cref{eq:Technique:3}. That is
\begin{align}\label{eq:Technique:4}
&\pr{\size{\Ec(X_R,Y_{-R},T,R)-\ip{X_R,Y_R}}\leq \sqrt{n}/ c} \geq 1 - 3/n^4
\end{align}
We will show that $\Ec$ violates the (external observer)  computational differential privacy of $\Gamma$. In the following we assume for simplicity that $\Ec$ is deterministic,  and let 
\begin{align}\label{eq:Technique:G}
	\cG = \set{(x,y,t) \colon \pr{\size{\Ec(x_R,y_{-R},t,R)-\ip{x_R,y_R}}\leq \sqrt{n}/ c} \geq 1 - 3/n^2}
\end{align}
I.e., the triplets for which $\Ec$ does well.  \cref{eq:Technique:4} yields that 
\begin{align}\label{eq:Technique:InG}
&\pr{(X,Y,T) \in \cG} \geq 1 - 1/n^2
\end{align}
As an easy  warm-up, assume that  for every good $(x,y,t)\in \cG$  it  holds that $\Ec(x_r,y_{-r},t,r) = \ip{x_r,y_r}$  (for every $r\in \zn$).  Then,  for every $i\in [n]$ and $r\in \zn $ with $r_i =1$, it holds that  
$$\Ec(x_r,y_{-r},t,r) - \ip{x_{r\xor e^i},y_{r\xor e^i}} =  \ip{x_r,y_r} - \ip{x_{r \xor e^i},y_{r \xor e^i}}  = x_i \cdot y_i$$ 
for $e^i \eqdef 0^{i-1}1 0^{n-i}$.  That is, knowing $x$ and $y_{-i} (\eqdef y_1,\ldots,y_{i-1},y_{i+1},\ldots,y_n)$, but not $y_i$, suffices for  learning $y_i$, which blatantly violates the  the differential privacy of $\Gamma$.  Doing such a reconstruction using the much weaker guarantee we have about $\Ec$,   is more challenging.    Details below.

For a triplet $\tuple= (x,y,t)$ and $i\in [n]$, let  
\begin{align}\label{eq:Technique:5}
\alpha_i^\tuple\eqdef 	\underbrace{ \eex{R|_{R_i=1}}{\Ec(x_R,y_{-R},t,R)}}_{\alpha_{i,\cY}^\tuple }- \underbrace{\eex{R|_{R_i=0}}{\Ec(x_R,y_{-R},t,R)}}_{ \alpha^\tuple_{i,\cX}}
\end{align}
Note that $\alpha_{i,\cY}^\tuple$  can be computed \emph{without} knowing  $y_i$, and  similarly $\alpha^\tuple_{i,\cX}$  can be computed \emph{without} knowing  $x_i$.   Below we exploit this  property for learning  $y_i$, or learning $x_i$.  We make the following key observation,\footnote{For $w\in\mon$, consider the \textit{Non-Boolean Hadamard  encoding}  defined by $C(w) \eqdef \set{ \ip{w,r}}_{r\in \zn}$. Since $\ip{x_r,y_r} = \ip{x\cdot y,r}$,  \cref{clm:Technique:ReconstrructionHadm} implies that  given access to an approximation of  $C(z)$ (as the one induced by $\Ec$), it is possible to  reconstruct most  bits of  $w$. While such reconstruction algorithms  are known (\cf \citet{dinur2003revealing}), for our purposes we critically exploit the very specific structure of the reconstruction value $\alpha^{\tuple}_i$. In particular, that it combines two estimations: one does not require knowing   $y_i$, and the second does not require knowing   $x_i$.}  see proof sketch in \cref{sec:Technique:ReconstrructionHadm}. Let $\sign(v) =1$ if $v>0$,  and $-1$ otherwise. 
\begin{claim}[Reconstruction from non-boolean Hadamard  encoding]\label{clm:Technique:ReconstrructionHadm}
For any $\tuple\in \cG$	it holds that $\ppr{i\gets [n]}{\sign(\alpha^{\tuple}_i) =x_i \cdot y_i} \geq 0.9$.
\end{claim}
That  is,  by computing both $\alpha^{\tuple}_{i,\cY}$ and $\alpha^{\tuple}_{i,\cX}$, one can reconstruct $x_i \cdot y_i$ for most $i$'s.  While for computing both of these values one has to know both $x_i$ and $y_i$,    we bootstrap  the above for learning  either  $x_i$  or $y_i$.   Let $\Tuple= (X,Y,T)$. By \cref{clm:Technique:ReconstrructionHadm} and the assumption about the size of $\cG$ (\cref{eq:Technique:InG}),   for most $i\in [n]$ it holds that 
 \begin{align}\label{eq:Technique:7}
 	\pr{\sign(\alpha^{\Tuple}_i) =X_i \cdot Y_i} \geq 0.85
 \end{align}

For ease of notation, we assume that \cref{eq:Technique:7} holds for $i=1$, fix $i$ to this value and omit it from the notation. For $w\in \mon$,  let $\widehat{w}\eqdef (- w_1, w_2,\dots,w_n)$ (\ie first bit is flipped). As  mentioned above, one cannot directly use \cref{eq:Technique:7} for computing $Y_1$ from $(X,Y_{-1},T)$, since computing  $\alpha^{\Tuple}$ requires knowing  $X_1$. So rather, we use the fact that
\begin{align*}
	\alpha_{\cY}^{x,y,t} &=  \alpha_{\cY}^{x,\hy,t} \; \,\text{and}  \; \,\ \alpha_{\cX}^{x,y,t} =  \alpha_{\cX}^{\hx,y,t}
\end{align*}
for all $(x,y,t)$, to  make the following observation (proof sketch in \cref{sec:Technique:Diamond}).
 
\begin{claim}[Inconsistent variant]\label{clm:Technique:Diamond} 
	$\min_{ X' \in \set{X,\hX},Y' \in \set{Y,\hY}}\set{\pr{\sign(\alpha^{X',Y',T}) = X'_1\cdot Y'_1} }\le 0.75 $.
\end{claim}
That is, not  all variants of the first bit of $X$ and   $Y$ are highly consistent with the prediction induced by $\alpha$.  Assume for concreteness  that $\pr{\sign(\alpha^{\hX,Y,T}) = \hX_1\cdot Y_1}\le 0.75$ (other cases are analogous), and consider the algorithm \Dc that on input $(x_{-1},y,t)$ outputs one if  $\sign(\alpha^{(1,x_{-1}),y,t}) =  y_1$. \cref{eq:Technique:7} yields that 
\begin{align*}
	\pr{\Dc(X_{-1},Y,T) = 1 \mid X_1 =1} \ge \pr{\Dc(X_{-1},Y,T) = 1 \mid X_1 =-1} + 0.1.
\end{align*}
Since $\alpha^{(1,x_{-1}),y,t}$ can be efficiently approximated from  $(X_{-1},Y,T)$,  given access to $\Ec$,  the above violates the assumed computational differential privacy of $\Gamma$ (for small enough constant $\eps$).\footnote{We remark that our results hold for any $\eps >0$.}

\remove{

We use the following easy-to-prove fact about differential private mechanisms:

\begin{proposition}\label{prop:DPMechanism}
	Let $\Mc\colon \zn\to \Ss$be an $\eps$-DP mechanism, let $X$ be uniform over $\zn$ and let $\cD$ be an algorithm such that $\pr{\cD(X,\Mc(X))=1}=\alpha$. Then, it holds that 
	$$\pr{\cD(\hX,\Mc(X))=1}\leq e^\epsilon\cdot \alpha,$$
	for $\hX \eqdef (-X_1,X_2,\cdot,X_n )$.
\end{proposition}\Nnote{The proof is in the end of the section}


To show the above, we use the following claim:

\begin{algorithm}[The distinguisher \cD]\label{alg:CondensingSV:EveDP}
	
	\item Oracle:  $\Ec$.

	\item Input: $(xy)\in\mo^{2n}$ and $t\in\Ss$.
	
	\item Operation:~
	\begin{enumerate} 	
		
		\item Compute $a = \RecBit^{f_{x,y,t}}$ for $f_{x,y,t}(r) \eqdef \Ec(x_R,y_{-R},t,r)$ 
		\item If $a \neq x_1\cdot y_1$ output $1$.
		\item Otherwise output $0$. 
	\end{enumerate}
\end{algorithm} 
It is easy to see that, by \cref{eq:overview:succ_of_A}, $\pr{\cD(X,Y,T)=1} \leq 0.01$.
\paragraph{The case that $\pr{\RecBit^{f_{\hX,Y,T}} = \hX_1\cdot Y_1} \leq 0.9$.}Assume the first option holds. It follows that $\pr{\cD(\hX,Y,T)=1} \geq 0.1$. That is, $\pr{\cD(\hX,Y,T)=1} > e^\eps \cdot\pr{\cD(X,Y,T)=1}$, with contradiction to \cref{lemma:overview:dist_to_dp}. 

\paragraph{The case that $\pr{\RecBit^{f_{X,\hY,T}} = X_1\cdot \hY_1} \leq 0.9$.} The second case is symmetric to the first. By the same proof outlines, it can be shown that exists an distinguisher that distinguish between $x,y$ to $x,\hy$.

\paragraph{The case that $\pr{\RecBit^{f_{\hX,\hY}} = \hX_1\cdot \hY_1} \leq 0.7$.} For the third case, assume that the two first cases do not hold. Thus $\pr{\RecBit^{f_{X,\hY,T}} = X_1\cdot \hY_1} > 0.9$. Let $\tD(x,y,t) \eqdef \tD(x,\hy,t)$. By assumption, $\pr{\tD(X,Y,T)}\leq 0.1$ and $\pr{\tD(\hX,Y,T)}\geq 0.3$. Again, this is a contradiction to \cref{lemma:overview:dist_to_dp}.
}

\subsubsection{Reconstruction from   Non-Boolean Hadamard Code}\label{sec:Technique:ReconstrructionHadm}
We sketch the proof of \cref{clm:Technique:ReconstrructionHadm}. 

\begin{proofsketch}
 Assume for simplicity that for any $\tuple =(x,y,t) \in \cG$:
\begin{align}\label{eq:Technique:6}
	\size{\Ec(x_r,y_{-r},t,r)-\ip{x_r,y_r}}\leq \sqrt{n}/ c
\end{align}
for \emph{all} $r\in \zn$ (and not for $1-1/n^2$ fraction of the $r$'s, as in the definition of $\cG$).\footnote{Note that the $1/n^2$ fraction of ``bad'' $r$'s (for which \cref{eq:Technique:6} does not hold)  can only affect the $\alpha_i$'s by at most $\frac{2}{n^2}\cdot \paren{\max_{r \in \zo^n}\set{\Ec(x_r,y_{-r},t,r)}-\min_{r \in \zo^n}\set{\Ec(x_r,y_{-r},t,r)}}$. Therefore, since \wlg $\Ec$ always outputs an estimation in $[-n,n]$,  the ``bad'' $r$'s might only affect the following calculation by the insignificant additive term of $4/n$.}

Fix $\tuple= (x,y,t) \in \cG$ and  omit it  when clear from the context, and let $\delta(r)\eqdef \Ec(x_r,y_{-r},t,r)- \ip{x_r,y_r}$.  A simple calculation yields that
	\begin{align}\label{eq:Technique:ReconstrructionHadm:1}
		\alpha_i &= \eex{r \gets {\zn}|_{r_i=1}}{\Ec(x_r,y_{-r},t,r)}-\eex{r \gets {\zn}|_{r_i=0}}{\Ec(x_r,y_{-r},t,r)}\\
		&= \ldots = x_i\cdot y_i +\underbrace{\eex{R|_{R_i=1}}{\delta(R)}-\eex{R|_{R_i=0}}{\delta(R)}}_{\xi_i}. \nonumber
	\end{align}
	It follows that if $\size{\xi_i} <1$, then  $\sign(\alpha_i) =x_i \cdot y_i$. Thus, for proving the claim it suffices to argue that $\xi_i$ is smaller than $1$ for $.9$ fraction of the  $i$'s. Let $\cI \eqdef \set{i \in [n] \colon \xi_i \geq 1}$ and $\cI' \eqdef \set{i \in [n] \colon \xi_i \leq -1}$. We conclude the proof showing that $\max\set{\size{\cI},\size{\cI'}} \leq 0.05 n$. Assume towards a contradiction that this is not the case, and specifically that $\size{\cI} > 0.05 n$ (the case $\size{\cI'} > 0.05 n$ is analogous).
	Let $I$ be uniform over $\cI$, and compute

	\begin{align}\label{eq:Technique:ReconstrructionHadm:3}
		\size{\ex{\xi_I}}&=\size{\eex{I,R|_{R_I=1}}{\delta(R)} - \eex{I, R|_{R_I=0}}{\delta(R)}}\\
		&\leq (\max_r \set{\delta(r)} -\min_r \set{\delta(r)} ) \cdot\SD(\delta(R|_{R_I=1}), \delta(R|_{R_I=0})) \nonumber\\
		&\le (\sqrt{n}/c) \cdot   \SD(\delta(R|_{R_I=1}), \delta(R|_{R_I=0})). \nonumber
	\end{align}
	The second inequality is by \cref{eq:Technique:6}. A rather straightforward bound, see  \cref{prop:Raz}, yields that    $\SD(R|_{R_I=0}, R|_{R_{I}=1}) \leq 1/\sqrt{\size{\cI}}$, 
	and thus, by the data-processing  property of statistical distance:
	\begin{align}\label{eq:Technique:ReconstrructionHadm:2}
		\SD(\delta(R|_{R_I=0}), \delta(R|_{R_I=1})) \leq 1/\sqrt{\size{\cI}}
	\end{align}
	Combining  \cref{eq:Technique:ReconstrructionHadm:2,eq:Technique:ReconstrructionHadm:3}, yields that  
	\begin{align*}
		\size{\ex{\xi_I}} \leq (\sqrt{n}/c) \cdot (1/\sqrt{\size{\cI}}) \leq \sqrt{20}/c.
	\end{align*}
	Thus,  for large enough $c$, we obtain that $\size{\ex{\xi_I}} < 1$, in contradiction to the fact that, by definition of $\cI$, it holds that $\size{\ex{\xi_I}} \geq 1$.
\end{proofsketch}

\subsubsection{Proving \cref{clm:Technique:Diamond}}\label{sec:Technique:Diamond}
We sketch the proof of \cref{clm:Technique:Diamond}.

\begin{proofsketch}
 By definition, for every  $(x,y,t)$ it holds that 
\begin{align*}
	\alpha_{\cY}^{x,y,t} &=  \alpha_{\cY}^{x,\hy,t} \; \,\text{and}  \; \,\ \alpha_{\cX}^{x,y,t} =  \alpha_{\cX}^{\hx,y,t}
\end{align*}
Recalling that $\alpha^{\tuple}\eqdef \alpha_{\cY}^{\tuple}  - \alpha_{\cX}^{\tuple}$, we conclude that  
\begin{align}\label{eq:Technique:Diamond:1}
	\alpha^{x,y,t} + 	\alpha^{\hx,\hy,t}
	&= 	\alpha_{\cY}^{x,y,t}  - \alpha_{\cX}^{x,y,t} + \alpha_{\cY}^{\hx,\hy,t}  - \alpha_{\cX}^{\hx,\hy,t}\\
	&= \alpha_{\cY}^{x,\hy,t}  - \alpha_{\cX}^{\hx,y,t} + \alpha_{\cY}^{\hx,y,t}  - \alpha_{\cX}^{x,\hy,t}\nonumber\\
	&= \alpha^{\hx,y,t} + 	\alpha^{x,\hy,t}.\nonumber
\end{align}

Assume towards contradiction that \cref{clm:Technique:Diamond} does not hold. That is, 
\begin{align}
	\forall X' \in \set{X,\hX}, Y' \in \set{Y,\hY}:\quad \pr{\sign(\alpha^{X',Y',T}) = \sign(X_1' \cdot Y_1')} > 0.75  
\end{align}
We conclude that 
\begin{align}
		\lefteqn{\pr{\sign(\alpha^{X,Y,T} +\alpha^{\hX,\hY,T}) = \sign(X_1 \cdot Y_1) }}\\
 &\ge \pr{\sign(\alpha^{X,Y,T}) = \sign(X_1 \cdot Y_1) \land \sign(\alpha^{\hX,\hY,T}) = \sign(\hX_1 \cdot \hY_1)} > 0.5, \nonumber 
\end{align}
and
\begin{align}
	\lefteqn{\pr{\sign(\alpha^{\hX,Y,T} +\alpha^{X,\hY,T}) = \sign(\hX_1 \cdot Y_1) }}\\
	& \ge \pr{\sign(\alpha^{\hX,Y,T}) = \sign(\hX_1 \cdot Y_1) \land \sign(\alpha^{X,\hY,T}) = \sign(X_1 \cdot \hY_1)} > 0.5. \nonumber 
\end{align}
Since $\sign(X_1 \cdot Y_1)$ and $\sign(\hX_1 \cdot Y_1)$ have opposite values, the above  is in contradiction to \cref{eq:Technique:Diamond:1}.
\end{proofsketch}

\subsection{Slightly Accurate Protocols}\label{sec:Technique:hardCase}
Our result holds for differentially private protocols for computing the inner product,  of  much weaker accuracy than what we considered above. In particular,  we can only  assume that for some $\ell \in \N$ it holds that
\begin{align}\label{eq:Technique:hardCase:1}
	\pr{\size{\Out-\ip{X,Y}} < \ell} \ge  c\cdot \ell /\sqrt{n}
\end{align}
for large enough constant $c >0$.  Namely, an accuracy which is only a constant factor away from the trivial bound. This weaker starting point translates into a few  additional challenges comparing to the highly accurate protocols case  discussed above. The first challenge (more details in \cref{sec:tech:IdentifyGoodTupples}) is that for such a weak accuracy, it is much harder to identify a noticeable fraction of \emph{non-trivial triplets}:    triplets $(x,y,t)$ on which  $\Ec$ (the estimator that violates the secrecy of \cref{prot:Technique:main})  has non-trivial accuracy in computing   $\ip{x_r,y_r}$.  Furthermore, for violating differential privacy using similar means  to those used in \cref{sec:Technique:Accurate}, it is not enough to prove that many such  non-trivial triplets exist. Rather,  it should be possible to identify  them, while missing  one of the entries of either  $x$ or of $y$.

A second challenge (more details in \cref{sec:tech:Reconstruction})  is that  the accuracy guarantee of  such non-trivial triplets is $c\cdot \ell/\sqrt{n}$, and not close to $1$ as  assumed in \cref{sec:Technique:Accurate}. This  requires us to use a much more sophisticated  reconstruction algorithm  than  the one we use in \cref{sec:Technique:Accurate} (\ie $\sign(\alpha^{\tuple}_i)$).

\subsubsection{Identifying  Good Triplets}\label{sec:tech:IdentifyGoodTupples}
 We need to argue that  even \wrt the weak accuracy of the inner-product protocol   $\Gamma$ stated in \cref{eq:Technique:hardCase:1}, an estimator $\Ec$ that violates the secrecy of the key-agreement protocol $\Pi$ (\cref{prot:Technique:main}), has many non-trivial triplets.  Our first step is to use a more sophisticated amplification reduction  for  $\Pi$, such that \Ec has the following guarantee:
\begin{align}\label{eq:Technique:security}
&\pr{\size{\Ec(1^n,X_R^n,Y_{-R}^n,T^n,R^n)-\ip{X_R^n,Y^n_R}}< \ell\text{ }\mid\text{ }\size{\Out^n_\Ac - \Out^n_\Bc} < \ell} \ge  c\cdot \ell /\sqrt{n}
\end{align}
That is,  \Ec predicts the key non-trivially when conditioning on \emph{agreement}.  Assuming such \Ec exists, the natural criterion for a triplet $(x,y,t= (\Out,\cdot))$ to be  non-trivial, is that $\Out$ is close to $\ip{x,y}$  (which by definition implies that $\Out^n_\Ac - \Out^n_\Bc$ is small). But as mentioned above, to be a useful criterion we should be able to identify such a triplet while missing $x_i$ (or $y_i$).\footnote{It is tempting to ignore the missing coordinate and to decide whether a triplet is non-trivial by comparing  $\ip{x_{-i},y_{-i}}$ to $\Out$. It turns out, however, that taking this approach might create an over-fitting between the decision and the value of  $x_i$, which might result in a  very poor predictor. As we mention below, a  similar approach is useful  \wrt a more distinguished set of triplets.\label{fn:tech:1}}  We overcome this problem by assuming the transcript contains an $\eps$-DP estimation $\es$ of $\ip{x,y}$ with a small additive error, which allows making the above decision without knowing the missing coordinate. By composition of differential privacy, it follows that even with such an estimation, it is impossible to violate the privacy of the inner-product protocol $\Gamma$.\footnote{We remark that while it may be impossible to implement a protocol with such an accurate estimation, privacy still holds by composition.} So the new candidates  for non-trivial triplets are
$$\cG= \set{(x,y,t= (\es,\out,\cdot)) \colon \size{\out - \es} \le \ell}$$


Unfortunately, the set $\cG$ is still  not what we need: it might be that \Ec does very well on a small fraction of $\cG$, and very poorly elsewhere. Therefore, our next step is to identify those triplets $(x,y,t)\in \cG$ for which  \Ec does well. Concretely, those for which
\begin{align}\label{eq:Technique:hardcase:3}
&\beta_{x,y,t}\eqdef\ppr{R}{\size{\Ec(1^n,x_R,y_{-R},t,R)-\ip{x_R,y_R}}< \ell } \ge  c\cdot \ell /2\sqrt{n}
\end{align}
A simple argument yields that the density of $\cG'\eqdef \set{(x,y,t) \in \cG\colon \beta_{x,y,t} \ge c\cdot \ell /2\sqrt{n}}$  in $\cG$ is at least $c\cdot \ell /2\sqrt{n}$. But how can we identify the triplets of $\cG'$, while missing a coordinate? The idea is to try an estimate $\beta_{x,y,t}$ without having, for instance,   $x_i$. That is, using 
\begin{align}\label{eq:Technique:hardcase:4}
&\beta^i_{x,y,t}=\ppr{R|_{R_1=0}}{\size{\Ec(1^n,x_R,y_{-R},t,R)-\ip{x_R,y_R}}< \ell }
\end{align}
As mentioned in \cref{fn:tech:1}, using such  estimate might cause the decision whether $(x,y,t) \in \cG'$ to be strongly \emph{dependent} on $x_i$, the bit that the estimator is missing. This is unfortunate, since our reconstruction algorithm is only guaranteed to reconstruct \emph{most} entries, and the above estimator may use the reconstruction algorithm only  on indexes that it fails to reconstruct. Luckily, it turns out that $\beta^i_{x,y,t}$ is ``not  too far'' from the desired $\beta_{x,y,t}$ for all but at most $1/\sqrt{n}$ of the indexes. And when focusing on triplets in the (identifiable) set $\cG$, a careful analysis yields that the above dependency is not too harmful. More details in \cref{sec:KAProtocol,sec:CondensingSV}

\subsubsection{Reconstructing  Slightly Good Triplets}\label{sec:tech:Reconstruction}
Our goal is  to  find  an efficient algorithm $\Dc$ that  given  $(x_{-i},y)$ (or $(x,y_{-i})$) and $t$ as input, and an oracle access to an estimator \Ec that is slightly accurate on the triplet $s=(x,y,t)$,   computes a (non-trivial) prediction of the missing element $x_i$ (or of $y_i$). Similarly to the highly accurate protocols  case,  see \cref{sec:Technique:Accurate}, we would like to determine a set of  values $\set{\alpha_i^\tuple}_{i\in [n]}$ such that:
\begin{enumerate}
	\item $\ppr{i \la [n]}{\sign(\alpha_i^\tuple) = x_i \cdot y_i}$ is sufficiently larger than $1/2$ \hfil(\ie the analog of \cref{clm:Technique:ReconstrructionHadm}), and\label{sec:tech:rec1}
	
	\item $\alpha_i^\tuple = \alpha_{i,\cY}^\tuple + \alpha_{i,\cX}^\tuple$, where $\alpha_{i,\cY}^\tuple$ can be computed without knowing $y_i$, and $\alpha_{i,\cX}^\tuple$ can be computed without knowing $x_i$.\footnote{In \cref{sec:Technique:Accurate}, we defined $\alpha_i^\tuple = \alpha_{i,\cY}^\tuple - \alpha_{i,\cX}^\tuple$ (i.e., with minus instead of plus) since it was more suitable for the specific $\alpha_i^\tuple$ that we considered there. In general, there is nothing special about the minus, and we can always switch between the cases by considering $(-\alpha_{i,\cX}^\tuple)$ as the part that is independent of $x_i$ (rather than $\alpha_{i,\cX}^\tuple$).}\label{sec:tech:rec2} 
\end{enumerate}
In particular,  we search for a function $\est$ such that

\begin{align}\label{eq:Technique:new-alpha}
	\alpha_i^\tuple
	\eqdef 	\eex{R}{\est^\Ec(i,x,y,t,R)}
	= \frac12 \paren{\underbrace{ \eex{R|_{R_i=1}}{\est^\Ec(i,x,y,t,R)}}_{\alpha_{i,\cY}^\tuple } +  \underbrace{\eex{R|_{R_i=0}}{\est^\Ec(i,x,y,t,R)}}_{ \alpha^\tuple_{i,\cX}}}
\end{align}
for $R\gets \zn$, satisfy the above requirements.\footnote{In \cref{sec:Technique:Accurate}, we implicitly used $\est^\Ec(i,x,y,t,r) = 2\cdot (-1)^{r_i + 1} \cdot \Ec(x_r,y_{-r},t,r)$ and, assuming that $\Ec$ is highly accurate,  showed that it satisfies the above requirements. We do not know whether this $\est$ satisfies the above requirements \wrt  slightly accurate  $\Ec$.}~\footnote{A reconstruction method from a somewhat accurate estimator for the inner-product functionality was presented by \citet{dinur2003revealing},  who showed a method for revealing most of the entries of a vector $z$ given an oracle access to  an algorithm \Ec that accurately estimates $\iprod{z,r}$ for $0.51$ fractions of of the $r$'s. This method, however, can only be carried out \emph{efficiently} \wrt $\Ec$ that is  accurate on $1 - \Omega(1/n)$ fraction of the $r$'s (\cite{DY08}).  \citet{DMT07}  improved over the above,  presenting  an efficient  reconstruction estimator that does well for  given access to an estimator that does well on  $0.77$ fraction of the $r$'s. Both methods, however,  are not suitable for estimators that are accurate for less than a constant fraction of the $r$'s (as we are aiming for in  \cref{eq:Technique:hardcase:3}). Furthermore, there is no clear way how to turn the reconstruction algorithms  presented by these  methods to  satisfy the second requirement above.}

For ease of notation, in the following we assume that the domain of the vector $r$ sent in  \cref{prot:Technique:main} is  $\mon$ (rather than $\zn$.) For such $r\in \mon$, let   $r^+ \eqdef \set{i \colon r_i =1}$, and let $r^{-} \eqdef  [n] \setminus r^{+}$.   Recall, see \cref{sec:tech:IdentifyGoodTupples}, that  \wlg, the transcript contains a part $e$ that is an  $\eps$-DP estimation of $\ip{x,y}$. In the following we  assume for simplicity  that  $\size{\es -  \ip{x,y}} \leq \ell$ (and not only with high probability). Towards defining the desired function $g^\Ec$, we   define the following function $f^\Ec$: 
\begin{align}
	f^{\Ec}(i,x_{r^+},y_{r^-},t = (\es,\cdot),r) \eqdef 2 \cdot \Ec(i,x_{r^+},y_{r^-},t,r) - \es
\end{align} 
\remove{
An easy calculation yields that for every $r \in \mon$
\begin{align*}
	\size{	f^{\Ec}(i,x_{r^+},y_{r^-},t,r) - \ip{x\cdot y,r}}
	\leq \ell + 2\cdot \size{\Ec(i,x_{r^+},y_{r^-},t,r) - \ip{x_{r^+},y_{r^+}}}
\end{align*}
Combining  the above with
}
Since $\Ec$ is a good estimator of $\ip{x_{r^+},y_{r^+}}$ (followed by \cref{eq:Technique:hardcase:3}), it holds that  
\begin{align}\label{eq:technique:f-assumption}
	\ppr{r \la \oo^n}{\size{f^\Ec(1^n,x_{r^+},y_{r^-},t,r)-\ip{x \cdot y, r}}< 3\ell } \ge  c\cdot \ell /2\sqrt{n}
\end{align}
That is, $f$ estimates $\ip{x \cdot y, r}$ well. In the following, let
 \begin{align*}
	\delta_i^{\Ec}(x,y,t,r) \eqdef f^\Ec(x_{r^+},y_{r^-},t,r) - \iprod{x_{-i} \cdot y_{-i}, r_{-i}}
\end{align*}

Note that if $f$ would have  computed  $\ip{x \cdot y, r}$ perfectly, then $\delta_i(x,y,t,r) = x_iy_i  r_i$, and the function $g$ defined by $g(i,x,y,t,r) \eqdef \delta_i(x,y,t,r) \cdot r_i$ would have satisfies \reqref{sec:tech:rec1} (it is clear, see below, that $g$ also satisfies  \reqref{sec:tech:rec2}). While we do not have such a strong guarantee about $f$,  we manage to prove that taking some additive offset  of $\delta_i$ yields a good enough  $g$. Specifically,  for $k\in \Z$, consider the   function $\est^\Ec_k$ defined by 
\begin{align}
	\est^{\Ec}_k(i,x,y,t,r) \eqdef \begin{cases} (\delta_i^{\Ec}(x,y,t,r) - k) \cdot r_i & \delta_i^{\Ec}(x,y,t,r) \in \set{k-1,k+1} \\ 0 & \text{otherwise;}\end{cases},
\end{align}
Namely, $\est^{\Ec}_k$ checks whether  $f(i,x_{r^+},y_{r^-},\cdot )$ might  be  off by  exactly  $k$ in estimating $\ip{x,y}$. If positive, it  assumes this is the case and predicts  $x_iy_i$, accordingly. In all other cases, $\est^{\Ec}_k$ takes no risks an outputs $0$. Of course, even if the check is positive, it might be that $f(i,x_{r^+},y_{r^-},\cdot )$ is off by $k-2$ or by $k+2$, and in this case $\est_k$ is wrong.

Since  $\delta_i^{\Ec}(x,y,t,r)$ can be computed without knowing $y_i$ if $r_i = 1$, and  without knowing $x_i$ if $r_i = -1$,  the function $\est^\Ec_k$, for each $k$, satisfies \reqref{sec:tech:rec2}. We conclude the proof by arguing that  for some $k$, the function $\est^\Ec_k$ satisfies \reqref{sec:tech:rec1}.\footnote{Actually, this $k$, whose value might depend on $(x,y,t)$, has to be efficiently computable. We ignore this concern  from this high-level description.} That is,
\begin{align}\label{eq:Reconstruction:1}
\eex{i \la [n],r\gets \mon}{x_i \cdot y_i \cdot \est_k^\Ec(i,x,y,t,r)}	>0
\end{align}

Hereafter,  we remove $\Ec$ from notation,   remove $x,y,t$ from the inputs of $\est_k$ and $ \delta_i$,  and remove $x_{r^+},y_{r^-},t$ from the inputs of $f$. We also let $z \eqdef  x \cdot y$ (coordinate-wise product), and let  $R$ be uniformly  distributed  $\oo^n$.\footnote{We remark that under this simplyfing notation, the goal now is essentially to show that for some $k$, estimating the sign of $\eex{r\gets \mon}{\est_k(i,r)}$ (which has oracle access to the estimator $\Fc(r) \eqdef f^\Ec(r)$ of $\ip{z,r}$) is a good reconstruction for \cref{thm:reconstruction:intro}. We note that \cref{eq:Reconstruction:1} is weaker than what is required in  \cref{thm:reconstruction:intro}, but we ignore this concern for the purpose of this high-level description.}

 Let $A_k$ be the event $\set{f(R) = \iprod{z, R} +k}$, and let $B_k^i$ be the event $\set{f(R) = \iprod{z_{-i},R_{-i}} - z_i R_i +k}$.  In words, $A_k$  is the event that $f$  accurately computes $\iprod{z, R}$ with offset $k$ (\ie  $\est^{\Ec}_k$ is correct), and  $B_k^i$ is the event that $f$  is not off by  $k$, but seems so when $z_i$ is not given (\ie $\est^{\Ec}_k$ is wrong). By definition, $\est_k(i,r)=z_i$ for $r \in A_k$ (\ie  $r$'s with $f(r) = \iprod{z, r} +k$), equals to $-z_i$ for $r \in B_k^i$, and equals to zero for all other $r$'s. Therefore,
\begin{align}
	z_i \cdot \eex{R}{\est_k(i,R)} =  \ppr{R}{A_k} - \ppr{R}{B_k^i}
\end{align}
We next argue that  $ \ppr{R}{A_k} -\eex{i \gets [n]}{\ppr{R}{B_k^i}}>0$ for some $k$,  yielding that  $g_k$ satisfies \cref{eq:Reconstruction:1}. In the following, let $a_k \eqdef \ppr{R}{A_k}$ and $b_k \eqdef \eex{i \gets [n]}{\ppr{R}{B_k^i}}$. We make the following  key observation:  for any $k\in \Z$ it holds that
\begin{align}\label{eq:Reconstruction:2}
	\size{b_k -  \frac12(a_{k-2} + a_{k+2})} \leq \mu,\text{ for } \mu \in  O(1/n)  
\end{align}
That is,  the probability of the ``bad'' event $B_k^i$ is essentially the average of the probabilities of the good events  $A_{k-2}$ and $A_{k+2}$.\footnote{\cref{eq:Reconstruction:2} is over simplified, and we refer to \cref{sec:reconstruction} for the actual statement and proof. But very intuitively,   (a close variant of) \cref{eq:Reconstruction:2} holds since, by definition,   the event $B^i_k$ occurs if and only if: (1) $A_{k+2}$ occurs and $z_i R_i = -1$, or (2) $A_{k-2}$ occurs and $z_i R_i = 1$. For a uniformly chosen $i$, the probability of (1) is (roughly) $\in a_{k+2} \cdot \paren{1/2 \pm O(1/\sqrt{n})}$, and the probability of (2) is  (roughly) $\in a_{k-2} \cdot \paren{1/2 \pm O(1/\sqrt{n})}$. \cref{eq:Reconstruction:2}   now follows since ``typically''  $a_{k-2},a_{k+2}\in\Theta(1/\sqrt{n})$.}

To conclude the argument, assume towards a contraction that  all  $k$'s are ``bad'': $a_k$ is not larger than $b_k$ (otherwise we are done). Under this assumption, \cref{eq:Reconstruction:2} yields that for every $k$:
\begin{align}\label{eq:Reconstruction:3}
	a_{k+2} \geq 2 a_k  - a_{k-2} - \mu
\end{align}
Let  $k^\ast \eqdef \argmax_{k \in \bbZ} \set{a_k}$. \cref{eq:technique:f-assumption} yields that $a_{k^*} \geq \frac{c}{12 \sqrt{n}}$. By \cref{eq:Reconstruction:3}, we deduce that $a_{k^*+2} \geq \frac{c}{12 \sqrt{n}} - \mu$, that $a_{k^*+4} \geq \frac{c}{12 \sqrt{n}} - 2\mu$, and so forth. Hence, for large enough  $c$,  
the sequence $\set{a_{k^*}, a_{k^* + 2}, \ldots}$ contains many large values, whose   sum is  more than one, in contradiction to the fact that they denote  probabilities of disjoint events.  We conclude that at least one $k$ is not bad, making $g_k$ is the desired function. More details in  \cref{sec:reconstruction}.

\remove{
\subsubsection{\Nnote{Move or delete}}

\begin{proof}[Proof of \cref{lemma:overview:dist_to_dp}]
	First, note that since $X$ is uniform, it holds that $X\equiv \hX$. Thus, 
	\begin{align}\label{eq:overview:Da_claim_assumption}
	\pr{\cD(\hX,M(X))=1}= \pr{\cD(X,M(\hX))=1}
	\end{align}
	
	Compute,
	\begin{align}
	\pr{\cD(\hX,M(X))=1}& = \pr{\cD(X,M(\hX))=1}\\\nonumber
	&=\eex{x\gets X}{\pr{\cD(x,M(\hx))=1}}\\\nonumber
	&=\eex{x\gets X}{\sum_{t\in \Supp(M(\hx))}\pr{M(\hx)=t}\cdot \pr{\cD(x,t)=1}}.\\\nonumber
	\end{align}
	By differential privacy, $\pr{M(\hx)=t} \leq e^{\eps}\cdot \pr{M(x)=t}$. Therefore, 
	\begin{align}
	\pr{\cD(\hX,M(X))=1}&=\eex{x\gets X}{\sum_{t\in \Supp(M(\hx))}\pr{M(\hx)=t}\cdot \pr{\cD(x,t)=1}}\\\nonumber
	&\leq  e^{\eps}\cdot \eex{x\gets X}{\sum_{t\in \Supp(M(\hx))}\pr{M(x)=t}\cdot \pr{\cD(x,t)=1}}\\\nonumber
	&=  e^{\eps}\cdot \pr{\cD(X,M(X))=1}\\\nonumber
	&=  e^{\eps}\cdot \alpha.\\\nonumber
	\end{align}
\end{proof}
}

%% file: Preliminaries.tex
\section{Preliminaries}\label{sec:Preliminaries}
\subsection{Notations}
We use calligraphic letters to denote sets, uppercase for random variables, and lowercase for values and functions. Let $\poly$ stand the set of all polynomials.  Let $\negl$ stand for a negligible function.

For $x\in \bbR$, let $\floor x$ [\resp $\ceil x$] denote the closest integer which is smaller [\resp larger] than $x$, and let $\lfloor x \rceil$ denote the closes integer to $x$ (rounding of $x$).
For $n \in \N$, let $[n] \eqdef \set{1,\ldots,n}$, and for $a<b \in \Z$ let $\iseg{a,b}\eqdef [a,b]\cap \Z$. Given a vector $v\in \Sigma^n$, let $v_i$ denote its \ith entry. For a set $\cI \subseteq [n]$, let $v_\I$ be the \emph{ordered sequence} $(v_i)_{i\in \I}$, let  $v_{-\cI} \eqdef v_{[n] \setminus \I}$, and let $v_{- i} \eqdef v_{-\set{i}}$ (\ie $(v_1,\ldots,v_{i-1},v_{i+1},\ldots,v_n)$). For $v\in \mon$, let $v\flipi\eqdef (v_1,\ldots,v_{i-1}, -v_i,v_{i+1},\ldots,v_n)$. For $\rr \in \oo^n$, let $\rr^+ \eqdef \set{i\in [n] \colon \rr_i =1}$ and let $\rr^- \eqdef [n] \setminus \rr^+ $. For two vectors $\px = (x_1,\ldots,x_n)$ and $\py = (y_1,\ldots,y_n)$, let $\px \cdot \py \eqdef (x_1 \cdot  y_1, \ldots, x_n \cdot y_n)$, and let $\ip{\px,\py} \eqdef \sum_{i=1}^n x_i y_i$. The vectors $\px$ and $\py$ are \emph{neighboring}, if they differ in exactly one entry. All logarithms considered here are in base $2$.

\subsection{Distributions and Random Variables}\label{sec:prelim:dist}
The support of a distribution $P$ over a finite set $\cS$ is defined by $\Supp(P) \eqdef \set{x\in \cS: P(x)>0}$. For a (discrete) distribution $D$ let $d\from D$ denote that $d$ was sampled according to $D$. Similarly,  for a set $\cS$, let $x \from \cS$ denote that $x$ is drawn uniformly from $\cS$. For a finite set $\cX$ and a distribution $C_X$ over $\cX$, we use the capital letter $X$ to denote the random variable that takes values in $\cX$ and is sampled according to $C_X$. The {\sf statistical distance} (\aka {\sf~variation distance}) of two distributions $P$ and $Q$ over a discrete domain $\cX$ is defined by $\sdist{P}{Q} \eqdef \max_{\cS\subseteq \cX} \size{P(\cS)-Q(\cS)} = \frac{1}{2} \sum_{x \in \cS}\size{P(x)-Q(x)}$. 


\begin{definition}[Strong Santha-Vazirani sources]\label{def:SV}
	The random variable $X$ over $\oo^n$ is an {\sf$\alpha$-strong Santha-Vazirani source} (denoted $\alpha$-strong $\SV$) if for every $i\in [n]$ and $x_{-i}\in\oo^{n-1}$ it holds that:
	$$
	\alpha \leq \frac{\pr{X_i=1\mid X_{-i}=x_{-i}}}{ \pr{X_i=-1\mid X_{-i}=x_{-i}} } \leq1/\alpha.
	$$
\end{definition}

\paragraph{Computation Santha-Vazirani sources.}
\begin{definition}[Computational strong Santha-Vazirani sources]\label{def:Comp_SV}
	The random variable ensemble $X=\set{X_\kappa}_{\kappa\in \N}$ over $\oo^n$ is an {\sf$\alpha(\kappa)$-strong computational Santha-Vazirani source} (denoted $\alpha$-strong $\CSV$) if for every \ppt$\Ac$, $i\in [n]$ and $x_{-i}\in\oo^{n-1}$, the following holds for every large enough $\kappa$: 
	$$
	\alpha(\kappa) \leq \frac{\pr{\Ac(1^\kappa,(X_\kappa)_{-i})=(X_\kappa)_{i}\mid X_{-i}=x_{-i}}}{ \pr{\Ac(1^\kappa,(X_\kappa)_{-i})=-(X_\kappa)_{i}\mid X_{-i}=x_{-i}} } \leq1/\alpha(\kappa).
	$$
\end{definition}

\subsection{Algorithms}
We consider both uniform and non-uniform algorithms (\ie Turing machines). Let \ppt stand for probabilistic polynomial time, and \pptm stand for \ppt (uniform) algorithm. Oracle access to a deterministic algorithm, means access to its input/output function. When using oracle access to a randomized algorithm, the caller has to set random coins for the call.
Oracle access to a distribution $D$ is just an oracle access to a no-input randomized function, in which the output distributed according to $D$. A distribution ensemble $\cD = \set{D_n}_{n \in \bbN}$ is called \emph{efficiently samplable} if there exists a \pptm $\Ac$ such that for every $n \in \bbN$, the output of $\Ac(1^n)$ is distributed according to $D_n$.

If the coins are not specified, it means that they are sampled uniformly at random. We denote an algorithm $\Ac$ with advice $z$, by $\Ac_z$.

\subsection{Two-Party Protocols}\label{sec:protocol}
A two-party protocol $\Pi=(\Ac,\Bc)$ is \ppt if the running time of both parties is polynomial in their input length. We let $\Pi(x,y)(z)$ denote a random execution of $\Pi$ on a common input $z$, and private inputs $x,y$. We assume \wlg that a protocol has a common output (part of its transcript).



\begin{definition}[$(\alpha,\gamma)$-Accurate protocol]
	A two-party protocol $\Pi$ with private inputs is {\sf $(\alpha,\gamma)$-accurate for the function $f$}, if for any inputs $x,y\in\mon$,  $\pr{\size{\out(T)-f(x,y)}\leq \alpha}\ge \gamma$, where $T$ is the transcript of $\Pi(x,y)$ and $\out(T)$ is the designated common output.

	A two-party protocol $\Pi$ that gets security parameter $1^\kappa$ as its common input is $(\alpha,\gamma)$-accurate for $f$ if $\Pi(\cdot,\cdot)(1^\kappa)$ is $(\alpha(\kappa),\gamma(\kappa))$-accurate for $f$, for every $\kappa\in \N$.
\end{definition}

\begin{definition}[Oracle-aided protocols]\label{def:ChannelAidedProtocol}
	In a two-party protocol $\Pi$ with oracle access to a {\sf protocol} $\Psi$, denoted $\Pi^\Psi$, the parties make use of the \textit{next-message function} of $\Psi$.\footnote{The function that on a partial view of one of the parties, returns its next message.} In a two-party protocol $\Pi$ with oracle access to a {\sf channel} $\CXYT$, denoted $\Pi^C$, the parties can jointly invoke $\CXYT$ for several times. In each call, an independent triplet $(x,y,t)$ is sampled according to $\CXYT$, one party gets $x$, the other gets $y$, and $t$ is added to the transcript of the protocol. 
\end{definition}

\subsection{Differential Privacy}\label{sec:prelim:DP}
We use the following standard definition of (information theoretic) differential privacy, due to \citet{DMNS06}. For notational convenience, we focus on databases over $\oo$.
\begin{definition}[Differentially private mechanisms]\label{def:mech}
	A randomized function $f\colon\oo^n\mapsto \zs$ is an {\sf $n$-size, $(\eps,\delta)$-differentially private mechanism} (denoted $(\eps,\delta)$-$\Dp$) if for every neighboring $w,w'\in \oo^n$ and every function $g\colon \zs\mapsto \zo$, it holds that 
	$$
	\pr{g(f(w))=1}\leq \pr{g(f(w'))=1}\cdot e^\eps +\delta.
	$$ 	
	If $\delta=0$, we omit it from the notation.
\end{definition}

\paragraph{The Laplace mechanism}
The most ubiquitous differential private mechanism is the so-called Laplace mechanism. For $\sigma \geq 0$, the \textit{Laplace distribution with parameter $\sigma$, denoted $\Lap(\sigma)$}, is defined by the probability density function $p(z) = \frac1{2 \sigma} \exp\paren{-\frac{\size{z}}{\sigma}}$. 

\begin{fact}\label{fact:laplace-concent}
	Let $\eps > 0$. If $X \la \Lap(1/\eps)$ then for all $t > 0: \quad \pr{\size{X} > t/\eps} \leq e^{-t}$.
\end{fact}

\begin{definition}[Laplace mechanism for the inner-product functionality over $\oo^n \times \oo^n$] 
	For $\eps>0$, the mechanism $\IP_{\eps}$ is defined by  $\IP_{\eps}(x,y)=\ip{x,y}+\lfloor w \rceil $, where $w \gets \Lap(2/\eps)$.
\end{definition}
\begin{theorem}[\cite{DMNS06}]\label{theorem:LapIP}
	For every $\eps >0$ it holds that $\IP_{\eps}$ is $\eps$-$\Dp$.\footnote{The original definition proposed by \cite{DMNS06} did not round the value of the Laplace distribution. 
		However, by the definition of differential privacy, any post-processing (function) applied on the output of the mechanism does not effect the $\Dp$ property of the mechanism. Specifically, if $f$ is an $\eps$-$\Dp$ mechanism, then for every function $g$, the mechanism $g(f(\cdot))$ is also $\eps$-$\Dp$. Thus, by taking $g$ to be the rounding function, $\IP_{\eps}(x,y)=\lfloor \ip{x,y}+ \gamma \rceil=\ip{x,y}+ \lfloor \gamma \rceil$ is also $\eps$-$\Dp$. 
	}
\end{theorem}

\subsubsection{Computational Differential Privacy}
There are several ways for defining computational differential privacy (see \cref{sec:intro:Related Work}). We use the most relaxed version due to \citet{BNO08}.
\begin{definition}[Computational differentially private mechanisms]\label{def:ComMech}
	A randomized function ensemble $f=\set{f_\kappa\colon\oo^{\isize(\kappa)}\mapsto \zs}$ is an {\sf $\isize$-size, $(\eps,\delta)$-computationally differentially private} (denoted $(\eps,\delta)$-$\CDP$) if for every poly-size circuit family $\set{\Ac_\kappa}_{\kappa\in \N}$, the following holds for every large enough $\kappa$ and every neighboring $w,w'\in\oo^{\isize(\kappa)}$:
	$$
	\pr{\Ac_\kappa(f_\kappa(w))=1}\leq \pr{\Ac_\kappa(f_\kappa(w'))=1}\cdot e^{\eps(\kappa)} +\delta(\kappa).
	$$ 
If $\delta(\kappa) = \negl(\kappa)$, we omit it from the notation. 
\end{definition}

\remove{
	\begin{definition}[Computing over a mechanism's output] 
		Function $\out$ is {\sf $(\alpha,\gamma)$-accurate for the function $f$ \wrt an $s$-size mechanism $M$}, if for every $w\in\zo^s$ $\pr{\size{\out(M(w))-f(w)}\leq \alpha}\ge \gamma$. Function $\out$ is {\sf $(\alpha,\gamma)$-accurate for the function $f$ \wrt $s$-size mechanism ensemble $M= \set{M_n}_{n\in \N}$}, if for every $n$ it is $(\alpha(n),\gamma(n))$-accurate for $f$ \wrt $M_n$.
	\end{definition}
}

\subsection{Channels}\label{sec:prelim:Channel}

A channel is a distribution over triplets $(X,Y,T)$, as defined below.

\begin{definition}[Channels]\label{def:channel} A {\sf channel} $\CXYT$ of size $\isize$ over alphabet $\Sigma$ is a probability distribution over $\Sigma^\isize \times \Sigma^\isize \times\zo^\ast$. The ensemble $\CXYT= \set{\CXYTk}_{\kappa\in \N}$ is an $\isize$-size channel ensemble, if for every $\kappa\in \N$, $\CXYTk$ is an $\isize(\kappa)$-size channel. We denote a channel of size one by a \emph{single-bit} channel. 
	
	We refer to $X$ and $Y$ as the {\sf local outputs}, and to $T$ as the {\sf transcript}.	A part of $T$ is marked as the {\sf designated (common) output}, denoted by $\out(T)$.
\end{definition}
Unless said otherwise, the channels we consider are over the alphabet $\Sigma = \oo$. We naturally identify channels with the distribution that characterize their output.

\begin{definition}[The channel of a protocol]\label{def:ChannlOfProtocol}
	For a no-input two-party protocol $\Pi= (\Ac,\Bc)$, we associate the channel $C_\Pi$, defined by $\C_\Pi= \CXYT$, where $X$, $Y$ and $T$ are the local output of $\Ac$, the local output of $\Bc$ and the protocol's transcript (respectively), induced by the random execution of $\Pi$. The designated output of $C_\Pi$ is set to the common output of $\Pi$, if such exists.
	
	For a two-party protocol $\Pi$ that gets a security parameter $1^\kappa$ as its (only, common) input, we associate the channel ensemble $ \set{C_{\Pi(1^\kappa)}}_{\kappa\in \N}$. 
\end{definition}

\begin{definition}[$(\alpha,\gamma)$-Accurate channel]\label{def:accurate-func}
	Channel $\CXYT$ is {\sf $(\alpha,\gamma)$-accurate for the function $f$}, if $\ppr{\CXYT}{\size{\out(T)-f(X,Y)}\leq \alpha}\ge \gamma$. Channel ensemble $\CXYT= \set{\CXYTk}_{\kappa\in \N}$ is  $(\alpha,\gamma)$-accurate for  $f$ if $\CXYTk$ is $(\alpha(\kappa),\gamma(\kappa))$-accurate for $f$, for every $\kappa \in \N$.
\end{definition}

\subsubsection{Differentially Private Channels}\label{sec:DPChannel}
Differentially private channels are naturally defined as follows:
\begin{definition}[Differentially private channels]\label{def:DPChannel}
	An $\isize$-size channel $\CXYT$ is {\sf$(\eps,\delta)$-differentially private} (denoted $(\eps,\delta)$-$\Dp$) if there exists a $2\isize$-size $(\eps,\delta)$-$\Dp$ mechanism $M$ 
	such that $(X,Y,T)\equiv(X,Y,M(X,Y))$. 
\end{definition}

\begin{definition}[Computational differentially private channels]\label{def:CDPChannel}
	A channel ensemble $\CXYT= \set{\CXYTk}_{\kappa\in \N}$ is {\sf$(\eps,\delta)$-computationally differentially private} (denoted $(\eps,\delta)$-$\CDP$) if there exists an $(\eps,\delta)$-$\CDP$ mechanism ensemble $M=\set{M_\kappa}_{\kappa\in \N}$ such that $(X_\kappa,Y_\kappa,T_\kappa)\equiv(X_\kappa,Y_\kappa,M_\kappa(X_\kappa,Y_\kappa))$ for every $\kappa\in\N$.
\end{definition}

We use the following properties of differentially private channels. We state the properties using efficient black-box reductions. Thus, they are applicable for both information-theoretic and computational differential privacy.

\remove{
\paragraph{Conditioning maintains differentially privacy.}
\begin{proposition}\label{prop:conditioning on DP:IT}
	Let $C= \CXYT$ be $(\eps,\delta)$-$\Dp$ channel. Then for every function $f$, the channel $C' = (\CXYT \mid\set{f(T)=1})$ is  $(\eps,\delta' \eqdef \delta/\pr{{f(T)=1}})$-$\Dp$. 
	
	Furthermore, security proof is black-box: there exits an oracle-aided \pptm $\Ac$ such that for any algorithm $\Ac'$ violating the  $(\eps,\delta')$-$\Dp$ of $C'$, algorithm $\Ac^{f,\Ac'}$ violates the  $(\eps,\delta)$-$\Dp$ of $C$. 
\end{proposition}
\begin{proof}
	Assume there exist function $f$, algorithm $\Ac'$, and neighboring $w,w'\in\zo^{2\isize}$, for $\isize$ being the size of $C$, such that
	\begin{align}\label{eq:dp break IT}
	&\pr{\Ac'(T)=1|(X,Y)=w, f(T)=1}
	> \pr{\Ac'(T)=1\mid (X,Y)=w', f(T)=1}\cdot e^{\eps} + \delta'
	\end{align}
	Define $\Ac^{f,\Ac'}(t)$ to output $1$ if $f(t)= \Ac'(t)=1$, and $0$ otherwise. Since
	\begin{align}
	\pr{\Ac^{f,\Ac'}(T)=1\mid f(T)\neq 1}=0,
	\end{align}
	it holds that 
	\begin{align*}
	\pr{\Ac^{f,\Ac'}(T)=1\mid(X,Y)=w} 
	&= \pr{\Ac^{f,\Ac'}(T)=1\mid(X,Y)=w,f(T)=1}\cdot \pr{f(T)=1}\\
	&= \pr{\Ac'(T)=1\mid(X,Y)=w,f(T)=1}\cdot \pr{f(T)=1}\\
	&>(\pr{\Ac'(T)=1\mid (X,Y)=w', f(T)=1}\cdot e^{\eps} + \delta') \cdot \pr{f(T)=1}\\
	&=\pr{\Ac^{f,\Ac'}(T)=1\mid (X,Y)=w'}\cdot e^{\eps} + \delta.
	\end{align*}
	The inequality holds by \cref{eq:dp break IT}. Thus, $\Ac^{f,\Ac'}$ violates the $\eps$-$\Dp$ privacy of $C$. 
\end{proof}
}

\paragraph{Composition.}
\begin{proposition}[Composition of differentially private channels.]\label{prop:Composition:IT} 
	Let $M_0$ and $M_1$ be $\isize$-size mechanisms, and let $\hM$ be the mechanism by $\hM(w)\eqdef(M_0(w),M_1(w))$. If $M_0$ is $\eps_0$-$\Dp$ and $M_1$ is $(\eps_1,\delta)$-$\Dp$, then $\hM$ is $(\eps_0+\eps_1,\delta)$-$\Dp$. 
	
	Furthermore, the proof is black-box: there exists an oracle-aided poly-time algorithm $\hf$ such that for any algorithm $f$ violating the  $(\eps_0+\eps_1,\delta)$-$\Dp$ of $\hM$,  there exists $a\in \oo^\isize$ such that either $\hf$, with advice $a$, violates the $\eps_0$-$\Dp$ of $M_0$, or it violates the $(\eps_1,\delta)$-$\Dp$ of $M_1$. 
\end{proposition}
\begin{proof}
	Assume towards a contradiction that $\hM$ is not $(\eps_0 + \eps_1,\delta)$-DP. Then, by definition, there exists a function $f$, and neighboring $w,w'\in\oo^{\isize}$ such that,
	\begin{equation}\label{eq:comp:contradict}
	\pr{f(\hM(w'))=1} > e^{\eps_0+\eps_1} \cdot \pr{f(\hM(w))=1}+\delta.
	\end{equation}
	In the following, let $f_0(x)\eqdef f(M_0(w'),x)$ and $f_1(x)\eqdef f(x,M_1(w))$.
	Compute
	\begin{align*}
	\pr{f(\hM(w'))=1}
	&= \pr{f(M_0(w'),M_1(w'))=1}\\
	&=\pr{f_0(M_1(w'))=1}\\
	&\leq e^{\eps_1}\cdot\pr{f_0(M_1(w))=1}+\delta\\
	&=e^{\eps_1}\cdot\pr{f_1(M_0(w'))=1}+\delta\\
	&\leq e^{\eps_0+\eps_1}\cdot\pr{f_1(M_0(w))=1}+\delta\\
	&=e^{\eps_0+\eps_1}\cdot\pr{f(\hM(w))=1}+\delta,
	\end{align*}
	in contradiction to \cref{eq:comp:contradict}. The two inequalities follow from the $\Dp$ property of $M_0$ and $M_1$. Thus we get a contradiction. The black-box property holds by considering $\widehat{f}$ to be either $f_0$ (with advice $w'$) or $f_1$ (with advice $w$).
\end{proof}

\paragraph{Composing SV source with \Dp mechanism.}
\begin{proposition}\label{prop:breaking the dp}
	There exists a poly-time oracle-aided algorithm $\Ac$ such that the following holds. Let $X$ be $e^{-\eps_1}$-strong-SV source over $\mon$, let $\Mc$ be a $(\eps_2,\delta)$-\Dp mechanism, let $\eps\eqdef \eps_1+\eps_2$, and let $\Dc$ be an algorithm such that 	
	\begin{align*}
	\pr{\Dc(i,X,\Mc(X))=1} > e^{\eps}\cdot\pr{\Dc(i,X\flipi,\Mc(X))=1}+\delta.
	\end{align*}
	Then there exists $z\in[n]\times\oo^n$, such that $\Ac^\Dc$ with advice $z$, violates the $(\eps_2,\delta)$-\Dp of $\Mc$. 
\end{proposition}
\begin{proof}
	Since $X$ is strong-SV, for every $x\in \mon$ it holds that $\pr{X=x}\leq e^{\eps_1}\cdot \pr{X\flipi=x}$. It follows that 
	\begin{align*}
	\pr{\Dc(i,X,\Mc(X\flipi))=1}
	&= \ppr{z \la X}{\Dc(i,z,\Mc(z\flipi))=1}\\
	&\leq e^{\eps_1}\cdot \ppr{z \la X\flipi}{\pr{\Dc(i,z,\Mc(z\flipi))=1}}\\
	&= e^{\eps_1}\cdot \pr{\Dc(i,X\flipi,\Mc(X))=1}.
	\end{align*}
	By combining it with the assumption on $\Dc$, we obtain that
	\begin{align*}
	\ppr{i\gets [n]}{\Dc(i,X,\Mc(X))=1} 
	&> e^{\eps}\cdot\pr{\Dc(i,X\flipi,\Mc(X))=1}+\delta\\
	&\geq e^{\eps_2}\cdot\ppr{i\gets [n]}{\Dc(i,X,\Mc(X\flipi))=1} + \delta
	\end{align*}
	By an averaging argument, there exists $x\in \mon$ and $i\in [n]$ such that 
	\begin{align*}
	\pr{\Dc(i,x,\Mc(x))=1} > e^{\eps_2}\cdot\pr{\Dc(i,x,\Mc(x\flipi))=1}+\delta.
	\end{align*}
	Let  $\Ac^\Dc$ be the algorithm that given advice $z\in [n]\times\mon$ and input $w$, outputs $\Dc(z,w)$. It follows that $\Ac$ with advice $(i,x)$ violets the $(\eps_2,\delta)$-\Dp of $\Mc$, \wrt the neighboring $x,x\flipi \in \mon$.
\end{proof}

\subsection{Two-Party Differential Privacy}\label{sec:DP}
In this section we formally define distributed differential privacy mechanism (\ie protocols). 

\begin{definition}\label{def:DP}
	A two-party protocol $\Pi=(\Ac,\Bc)$ is {\sf $(\eps,\delta)$-differentially private}, denoted $(\eps,\delta)$-$\Dp$, if the following holds for every algorithm $\Dc$: let $\V^\Pc(x,y)(\kappa)$ be the view of party $\Pc$ in a random execution of $\Pi(x,y)(1^\kappa)$. Then for every $\kappa,n \in \N$, $x\in \oo^n$ and neighboring $y,y'\in\oo^n$:
	\begin{align*}
	\pr{\Dc(V^\Ac(x,y)(\kappa))=1}\le \pr{\Dc(V^\Ac (x,y')(\kappa))=1}\cdot e^{\eps(\kappa)}+\delta(\kappa),
	\end{align*} 
	and for every $y\in \oo^n$ and neighboring $x,x'\in\oo^{n}$:
	\begin{align*}
	\pr{\Dc(V^\Bc(x,y)(\kappa))=1}\le \pr{\Dc(V^\Bc (x',y)(\kappa))=1}\cdot e^{\eps(\kappa)}+\delta(\kappa).
	\end{align*}

	Protocol $\Pi$ is $(\eps,\delta)$-$\Dp$ {\sf against external observer} if we limit the above $\Dc$ to see only the protocol {\sf transcript}.

	Protocol $\Pi$ is {\sf $(\eps,\delta)$-computational differentially private}, denoted $(\eps,\delta)$-$\CDP$, if the above inequalities only hold for a non-uniform \ppt $\Dc$ and large enough $\kappa$. We omit $\delta = \negl(\kappa)$ from the notation. 
\end{definition}

\Enote{simulation-based}
\begin{remark}[Comparison with simulation-based definition of computational differential privacy]\label{rem:comDPChannel}
	An alternative stronger definition of computational differently privacy, known as \textit{simulation based computational differential privacy}, stipulates that the distribution of each party's view is computationally indistinguishable from a distribution that preserves privacy in an information-theoretic setup. 
	\cref{def:DP} is \emph{weaker} than the above, and thus proving lower bound on a protocol that achieves this weaker guarantee (as we do in this work) is a stronger bound.
\end{remark}

\paragraph{The randomized response protocol for IP.}
The randomized response method of \cite{warner1965randomized} can be used in order to construct a protocol for the inner-product. This protocol achieves $\eps$-\Dp and $(c_\eps\sqrt{n},1/2)$-accuracy, for every $\eps>0$ and some constant $c_\eps$ (dependent on $\eps$)\cite{MPRV09}.

\begin{protocol}[$\Pi = (\Ac,\Bc)$]\label{prot:randomResponse}
	\item Parameter: $n$, $\eps$.
	\item $\Ac$'s private input: $\px \la \oo^n$
	\item $\Bc$'s private input: $\py \la \oo^n$
	\item Operation:
	\begin{enumerate}
		
		\item Let $p\eqdef \frac{e^{\epsilon}}{e^{\epsilon}+1}- \half$. \Ac samples $\hat{x}$, a noise verson of $x$: for every $i\in [n]$, \Ac sets $\hat{x}_i$ to be $x_i$ with probability $\half+p$ and $-x_i$ with probability $\half-p$, independently.
		
		\item \Bc computes $z\eqdef 1/(2p)\cdot\sum_{i=1}^n y_i\cdot \hat{x}_i+ \Lap(1/(p\cdot \eps))$ and send $z$ to \Ac.
		
		\item Both parties output $z$.
	\end{enumerate}
\end{protocol}

\begin{proposition}\label{prop:randResp}
	Let $\Pi$ be \cref{prot:randomResponse}. For every $\epsilon>0$ there exists a constant $c_\eps$ such that the following holds. For every $n\in \N$, $\Pi_{n,\eps}$ is a $\eps$-\Dp protocol with $(c_\epsilon\sqrt{n},1/2)$-accuracy for IP.
\end{proposition}

\subsection{Key Agreement}\label{sec:prlim:KA}
We start with defining the information-theoretic case. 

\begin{definition}[Key-agreement channel]\label{def:KeyAgreementChannel} 
	The following properties are associate with a channel $C= \CXYT$:
	
	\begin{description}
		\item[Agreement:] $C$ has {\sf $\alpha$-agreement} if $\pr{X=Y}\ge \alpha$.
		
		\item[Leakage:] $C$ has {\sf $\delta$-leakage} if $\pr{f(T)=X} \le \delta$ for {\sf every} function (\ie ``eavesdropper'') $f$.
		
		\item[Equality-leakage:] $C$ has {\sf $\delta$-equality-leakage} if $\pr{f(T)=X \mid X= Y} \le \delta$ for every function $f$.
		
	\end{description}
	
	A single-bit, $\alpha$-agreement, $\delta$-leakage channel is called an {\sf $(\alpha,\delta)$-key agreement}. An $\alpha$-agreement, $\delta$-equality-leakage channel  is called an {\sf $(\alpha,\delta)$-key-agreement-with-equality-leakage}. 	
\end{definition}


\paragraph{Amplification.}
We use the following amplification result, implicit in \cite{HolensteinThesis06}, for key-agreement channels with equality-leakage.

\def \thmKAmpThomas{
	Let $\alpha>\delta\in (0,1]$ be constants. There exists a \ppt, oracle-aided, two-party protocol $\Pi$ such that the following holds. Let $C$ be a single-bit, $(\alpha,\delta)$-key-agreement with equality-leakage channel. Then the channel $\tC$ induced by $\Pi^C(1^\kappa)$ is a single-bit, $(1-2^{-\kappa},1/2+2^{-\kappa})$-key agreement.
	
	Furthermore, the security proof is black-box: there exists an oracle-aided  \Ec such that for every single-bit channel $C$ with $\delta$-agreement, and  an algorithm \tE violating the $(1/2+2^{-\kappa} +\beta))$-leakage of $\tC$ for some $\beta>0$, algorithm $\Ec^{C,\tE}(\kappa,\beta)$ runs in time $\poly(\kappa,1/\beta)$ and violates the $\delta$-leakage of $C$.
	\Enote{A note for us: the $\poly(1/\beta)$ is not explicitly mentioned in \cite{HolensteinThesis06}. However, from reading the proof, it eventually relies on Theorem 6.9 there (uniform hard-core lemma - measure version), where in this theorem the $\poly(1/\beta)$ ($\beta$ is $\gamma$ there) is clearly described.}
}

\begin{theorem}[Key-agreement amplification, implicit in \cite{HolensteinThesis06}]\label{thm:KaAmp:Hol}
	\thmKAmpThomas
\end{theorem}

\paragraph{Combiners.}
We use the following key-agreement ``combiner''.

\begin{theorem}[Key-agreement combiner \cite{harnik2005robust}]\label{Theorem:Combiners}
	There exists a \ppt, oracle-aided, two-party protocol $\Pi$ such that the following holds: let $\cC = \set{C_i}_{i\in[\ell]}$ be a set of channels such that at least one of them is a single-bit $(3/4,1/2+\delta)$-key-agreement for some $\delta>0$. Then the channel $\tC$ induced by $\Pi^{\set{C_i}_{i\in [\ell]}}(1^\kappa,1^\ell)$ is a single-bit $(1-2^{-\kappa},1/2+\delta\cdot p(\kappa))$-key-agreement, for some universal $p\in \poly$.

	Furthermore, the security proof is black-box: there exists an oracle-aided \ppt \Ec such that for every single-bit channel family $\cC = \set{C_i}_{i\in[\ell]}$, every index  $i\in[\ell] $ such that $C_i$ has $3/4$-agreement, and every algorithm \tE that violates the $(1/2+\delta\cdot p(\kappa))$-security of $\tC$, algorithm $\Ec^{\cC,\tE}(1^\kappa,1^\ell,i)$ violates the $(1/2+\delta)$-security of $C_i$.
\end{theorem}

\subsubsection{Key-Agreement Protocols}
We now define the computational notion for key-agreement protocols and channel ensembles. 

\begin{definition}[Computational key-agreement channels and protocols]\label{def:CompKey-agreementChannel} 
	The following properties are associate with a channel ensemble $C=\set{\CXYTk}_{\kappa\in \N}$ 
	\begin{description}
		\item[Agreement:] $C$ has {\sf $\alpha$-agreement} if $\pr{X_\kappa=Y_\kappa}\ge \alpha(\kappa)$.
		
		\item[Leakage:] $C$ has {\sf $\delta$-leakage} $\pr{\Fc(T_\kappa)=X_\kappa} \le \delta(\kappa)$ for {\sf every} \ppt (\ie ``eavesdropper'') $\Fc$ and a large enough $\kappa\in\N$.
		
		\item[Equality-leakage:] $C$ has {\sf $\delta$-equality-leakage} if $\pr{f(T_\kappa)=X_\kappa \mid X_\kappa= Y_\kappa} \le \delta(\kappa)$ for {\sf every} \ppt (\ie ``eavesdropper'') $\Fc$ and a large enough $\kappa\in\N$.
		
	\end{description}
	A single-bit, $\alpha$-agreement, $\delta$-leakage channel is called an {\sf $(\alpha,\delta)$-key agreement}. An $\alpha$-agreement, $\delta$-equality-leakage channel  is called an {\sf $(\alpha,\delta)$-key-agreement-with-equality-leakage}.

	A two-party protocol $\Pi=(\Ac,\Bc)$ is an {\sf $(\alpha,\delta)$-key agreement protocol}, if its associate channel ensemble $ \set{C_{\Pi(1^\kappa)}}_{\kappa\in \N}$ is an $(\alpha,\delta)$-key agreement channel ensemble.
\end{definition}

\subsection{Basic Probability Bounds}\label{sec:prelim:BasicProbFacts}


\begin{fact}[Hoeffding's Inequality]\label{fact:Hoeff}
	Let $X_1,\ldots,X_n$ be independent random variables, each $X_i$ is bounded by the interval $[a_i,b_i]$, and let $\bar{X} = \frac1{n}\cdot \sum_{i=1}^n X_i$. Then for every $t \geq 0$:
	\begin{align*}
	\forall t \geq 0: \quad \pr{\bar{X} - \ex{\bar{X}} \geq t}, \text{ }\pr{\bar{X} - \ex{\bar{X}} \leq -t}
	\leq \exp\paren{-\frac{2 n^2 t^2}{\sum_{i=1}^n (b_i - a_i)^2}}
	\end{align*}
\end{fact}

The following propositions are proven in \cref{sec:missing-proofs:proving-exp-of-abs,sec:missing-proofs:proving-boundMultDist,sec:missing-proofs:proving-Raz}, respectively.

\def\propExpOfAbs{
	Let $n \in \bbN$ be larger than some universal constant, and let $X = \size{X_1 + \ldots + X_n}$, where the $X_i$'s are i.i.d., each takes $1$ w.p. $1/2$ and $-1$ otherwise. Then for event $E$, it holds that
	\begin{align*}
	\pr{E} \cdot \ex{X \mid E} \leq 4\sqrt{n}
	\end{align*}
}

\begin{proposition}\label{proposition:exp-of-abs}
	\propExpOfAbs
\end{proposition}

\def\propBoundMultDist{
	Let $R$ be an uniform random variable over $\zn$, and $E$ some event s.t. $\ppr{R}{E} \geq 1/n$. Then for every $q>0$ it holds that
	\begin{align*}
	\ppr{i \gets [n]}{\exists b\in \zo \text{ s.t. } \ppr{R|_{R_i=b}}{E}\notin (1\pm 2q)\cdot \ppr{R}{E}}\le \log n/(n\cdot q^2).
	\end{align*}
}

\begin{proposition}\label{lemma:boundMultDist}
	\propBoundMultDist
\end{proposition}

\def\propRaz{
	Let $R$ be uniform random variable over $\zn$,and let $I$ be uniform random variable over $\cI \subseteq [n]$, independent of $R$. Then $SD(R|_{R_I=1},R|_{R_I=0})\leq 1/\sqrt{\size{\cI}}$.
}

\begin{proposition}\label{prop:Raz}
	\propRaz
\end{proposition}

%% file: The_KE_Protocol.tex
\newcommand{\hCT}{\widehat{C}_{\OA\OB\hT_{\ell}}}
\newcommand{\lCT}{\overline{C}_{\OA\OB\overline{T}_{\ell}}}
\newcommand{\hell}{{\widehat{\ell}}}
\newcommand{\hBB}{\widehat{\OB}}
\newcommand{\hBA}{\widehat{\OA}}
\newcommand{\oBB}{\overline{\OB}}
\newcommand{\oBA}{\overline{\OA}}
\newcommand{\oPi}{\overline{\Pi}}
\newcommand{\oT}{\overline{T}}
\newcommand{\oD}{\overline{D}}
\newcommand{\DLap}{\Delta_\Lap}
\newcommand{\amax}{a_{\max}}
\newcommand{\Combiner}{\MathAlgX{Comb}}
\newcommand{\Amplifier}{\MathAlgX{Amp}}
\newcommand{\oDelta}{\overline{\Delta}}
\newcommand{\ModA}{Z_{\sf mod}}
\newcommand{\oEveDP}{\overline{\EveDP}}
\newcommand{\pn}{n}
\newcommand{\con}{{c}}
\newcommand{\hPsi}{\widehat{\Psi}}
\newcommand{\hPi}{\widehat{\Pi}}
\newcommand{\neps}{{2}}

\section{Key Agreement  from Differentially Private Inner Product}\label{sec:KAProtocol}
In this section we prove that differentially private protocols (and channels) that estimate the inner product ``well'', can be used to construct a key agreement protocol.  We start,  \cref{sec:KAProtocol:IT},  with the information-theoretic case, in which the privacy holds information-theoretically (\ie against unbounded observers). In \cref{sec:KAProtocol:Comp}, we extend the result to the computational case.

\subsection{The Information-Theoretic Case}\label{sec:KAProtocol:IT}
The starting point in the information-theoretic case is a differentially private channel (\ie a triplet of random variables) that estimates the inner product well. For such channels,   we prove the following result.
\begin{theorem}[Key-agreement from differentially private channels estimating the inner product]\label{thm:KAProtocol:IT}
	There exists an oracle-aided \ppt protocol $\Lambda$ and a universal constant $c>0$ such that the following holds for every $\eps_1,\eps_2>0$: let $\CXYT$ be an $\pn$-size, $(\eps_1,1/\pn^2)$-$\Dp$ channel over $\oo$, such that $(X,Y)$ is an $e^{-\eps_2}$-strong SV source over $\oo^{2n}$, and let $\eps\eqdef\eps_1+\eps_2$ . If $C$ is   an $(\mu,~c\cdot e^{\con\cdot\eps}\cdot \mu/\sqrt \pn)$-accurate channel for the inner-product functionality, for some $\mu\ge \log \pn$, then the channel induced by  $\Lambda^C(1^\kappa)$ is   $(1-2^{-\kappa},\nfrac12 +2^{-\kappa})$-key agreement. \footnote{Requiring that $(X,Y)$ has  ``enough'' of entropy is mandatory.   For instance, perfectly accurate, perfect \DP  (\ie $(0,0)$-\DP) channels exist \emph{unconditionally} for  $0$-entropy (\ie fixed) $(X,Y)$, or more generally,  for $X$ and $Y$ that most of their coordinates are fixed.}  \footnote{It seems provable that the $(\eps_1, 1/n^2)$-$\Dp$ can be improved to $(\eps_1, O(1/n))$-$\Dp$. However, since it complicates the (already rather long) proof, we chose to prove the slightly weaker variant of this theorem, stated here. }
\end{theorem}
As mentioned in the introduction, \cref{thm:KAProtocol:IT} immediately yields that inner-product is a condenser for, independent,  strong Santha-Vazirani sources. 

\begin{corollary}[Inner-product is a good condenser for strong SV sources, \cref{cor:intro:condenser} restated]\label{cor:KAProtocol} 
	There exist universal constants $c_1, c_2>0$ such that the following hold for every independent   $e^{-\eps}$-{\em strong} SV sources $X$ and $Y$ of size $\pn$. 
	\begin{itemize}
		\item $\Hmin(\iprod{X,Y})\geq \log\paren{{\sqrt{\pn}}/{e^{c_1 \eps}  c_1\log \pn}}$, and
		\item$\Hmin(\iprod{X,Y} \bmod c_2 \sqrt{\pn}) \geq \log\paren{{\sqrt{\pn}}/{e^{c_1 \eps}  c_1\log \pn}}$.
		
	\end{itemize} 
\end{corollary}
\begin{proof}[Proof of \cref{cor:KAProtocol}]
	We only prove the second item (the proof of the first item follows by similar means).
	\remove{
		\Nnote{consider omit the first item proof and say it is similar to the second}
		We start with the proof of first item. Assume toward contradiction that \cref{cor:KAProtocol} does not holds. It follows that there exist two independent $e^{-\eps}$ strong-SV sources $X$ and $Y$ and $z\in \N$, such that $\pr{\ip{X,Y}=z} > e^{4\eps}\cdot c'\cdot \log(\pn)/\sqrt \pn$. 
		Consider the two-party protocol $\Pi(\Ac,\Bc)$ in which party $\Ac$ draws a random independent sample from $X$, party $\Bc$ draws a random independent sample from $Y$, and both parties output $z$ (regardless of their samples). And  note that the channel associated with $\Pi$, $C=C_\Pi$ is defined by the distribution of   $(X,Y,z)$. By definition, $C$  is  an $\eps$-$\Dp$ channel, that is $(1,e^{4\eps}\cdot c'\cdot \log(\pn)\sqrt \pn)$-accurate for the inner-product functionality (where $\out(\cdot)$ is a constant function that outputs $z$). Thus, by \cref{thm:KAProtocol:IT}, there exists a channel $C'$ that is  $(1-2^{-\kappa},\frac12+2^{-\kappa})$-key agreement, however, this is a contradiction since such key-agreement protocols and channels cannot exist unconditionally. 
	}	
	Let $\tPi$ be the $\nfrac12$-\Dp randomized-response protocol for the inner-product (\cref{prot:randomResponse}) with accuracy $(c\cdot \sqrt{n}, \nfrac12 )$ for some constants $c$, and  assume towards contradiction that \cref{cor:KAProtocol} does not hold. It follows that there exist two independent $e^{-\eps}$-strong SV sources $X$ and $Y$, and $z\in \iseg{0,c\sqrt{n}-1}$, such that 
	\begin{align}\label{eq:z_with_large_prop}
	\pr{\ip{X,Y}  \equiv z \bmod c\sqrt{n}} > e^{4\eps}\cdot c'\cdot \log(\pn)/\sqrt \pn
	\end{align}
	Consider the following two-party protocol $\Pi=(\Ac,\Bc)$:  $\Ac$ draws $x\gets X$,  $\Bc$ draws $y\gets Y$, and the  parties interact in $\tPi(x,y)$ to get a common output $\tout$.  \Ac then sends $s \gets \set{-1,0,1}$ to \Bc, and both parties output $\out \eqdef (\lfloor \tout/c\sqrt{\pn} \rfloor + s )c\sqrt{\pn}  +z$.

	Let $X,Y,S,\tOut$ and $\Out$, be the values of $x,y,s,\tout$ and $\out$, in a  random execution of $\Pi$.
	It is not hard to verify that if $\size{\ip{x,y}-\tout}\leq c\sqrt{\pn}$, then $\floor{\tout/c\sqrt{\pn}} = \lfloor \ip{x,y}/c\sqrt{\pn} \rfloor \pm 1$.
	Therefore,  by the accuracy  of $\tPi$ (\cref{prop:randResp}), for every $x,y$ it holds that:
	\begin{align}\label{eq:the_1/6_inequality}
	\pr{\floor{\tOut/c\sqrt{\pn}}  + S = \floor{\ip{X,Y}/c\sqrt{\pn}} \mid (X,Y)=(x,y)} \ge 1/6
	\end{align}
	Note that for every $x,y$ with $\ip{x,y} \equiv z \mod c\sqrt{\pn}$, it holds that $\ip{x,y} = \floor{\ip{x,y}/c\sqrt{n}}\cdot c\sqrt{n} + z$.
	Hence, by \cref{eq:z_with_large_prop,eq:the_1/6_inequality}  
	\begin{align*}
	\pr{\Out = \ip{X,Y}}
	\geq e^{4\eps}\cdot c'\cdot \log(\pn)/6\sqrt \pn.
	\end{align*}
	By definition, the  channel $C$ induced by  $\Pi$  is distributed according to     $(X,Y,(\tT,S))$, for  $\tT$ being the transcript of $\tPi(X,Y)$. Since $\tPi$ is $\nfrac12$-$\Dp$, then so is $C$.  Thus, by \cref{thm:KAProtocol:IT}, assuming that the constant $c'$ is large enough, there exists a channel $C'$ that is  $(1-2^{-\kappa},\nfrac12 +2^{-\kappa})$-key agreement. Such channels, however, do not  exist unconditionally. 
\end{proof}

We prove \cref{thm:KAProtocol:IT}  using the following transformation  that  utilizes  a \Dp channel that estimates  the inner-product functionality well, to create  a key-agreement-with-equality-leakage protocol (over non-boolean  domain).
\begin{protocol}[$\Pi^{C}_{\pn,\ell} = (\Ac,\Bc)$]\label{protocol:MainProtocol}
	\item Oracle: $\pn$-size channel $\CXYT$.
	\item Parameters: $\pn,\ell$.
	\item Operation:
	\begin{enumerate}
		
		\item The parties (jointly) call the channel  $\CXYT$. 
		
		Let $x$, $y$, and $t$, be the output of $\Ac$ and $\Bc$, and the common transcript of this call,  respectively. 
		
		\item  $\Ac$ samples $v \la [\ell]$ and  $\rr \la \oo^\pn$, and sends $(v, \px_{r^{+}}, \rr)$ to $\Bc$. 
		
		
		\item  $\Bc$ sends $\py_{\rr^{-}}$ to $\Ac$.
		
		\item   $\Ac$ sets $u_\Ac =  \ip{\px_{\rr^{-}}, \py_{\rr^{-}}}$, and (locally) outputs $o_\Ac =  \floor{\frac{u_\Ac - v}{\ell}}\cdot \ell$.
		
		\Bc sets $u_\Bc = \out(t)- \ip{\px_{\rr^{+}}, \py_{\rr^{+}}}$, and (locally) outputs $o_\Bc = \floor{\frac{u_\Bc - v}{\ell}}\cdot \ell$.
	\end{enumerate}
\end{protocol} 
The following lemma,  which is the main technical contribution of this section, states that for the right choice of parameters,  the channel induced by \cref{protocol:MainProtocol}  is a weak key agreement.  
\begin{lemma}[Main lemma, information theoretic case]\label{lem:KAProtocol:IT}
	There exists a  constant $c>0 $ such that the following holds for every  $\eps_1,\eps_2,\delta>0$: let  $\CXYT$ be an   $\pn$-size, $(\eps_1,1/\pn^2)$-$\Dp$ channel over  $\oo$, such that $(X,Y)$ is $e^{-\eps_2}$-strong SV over $\oo^{2n}$, and let $\eps\eqdef\eps_1+\eps_2$ and $\beta\eqdef c\cdot e^{\con\cdot\eps}$.   If $\CXYT$ is $(\mu,\beta \cdot\mu/\sqrt \pn)$-accurate for the inner-product functionality, for some $\mu\ge \log \pn$,   then  there exists $ \hell\ge\mu$ such that channel induced by $\Pi^C_{\pn,\hell}$ is a $(\alpha,\alpha/2^{15})$-key-agreement-with-equality-leakage, for $\alpha\eqdef(\beta\cdot\hell)/(8\sqrt{\pn})$.  
	
	Furthermore, the above is proved in a {\sf black-box way}: there exists an oracle-aided  \ppt \EveDP such that for any deterministic algorithm \Eve that breaks the  above stated  equality-leakage of  $\Pi^C_{\pn,\hell}$, there exists an advice string $a \in \oo^{3\pn}$ such that $\EveDP^{C,\Eve}$, with advice $a$, violates the $(\eps_1,1/\pn^2)$-$\Dp$ property of $C$.
\end{lemma}
We prove \cref{lem:KAProtocol:IT} below, but first use it for proving \cref{thm:KAProtocol:IT}.
\paragraph{Proving \cref{thm:KAProtocol:IT}.}
In addition to  \cref{lem:KAProtocol:IT}, we  make use of the following key-agreement amplification theorem, proven in  \cref{sec:KAApmlification}, that yields that for the correct value of $\hell$, the channel implied by $\Pi_{\pn,\hell}^C$ can be amplified into a full-fledged key agreement.

\def \thmKAmplificationN{
	There exists an oracle-aided two-party protocol $\Phi$ such that the following holds for every $\alpha\in (0,1]$. Let  $C$ be an $n$-size, $(\alpha,\alpha/2^{15})$-key-agreement-with-equality-leakage channel. Then the channel $\tC$ induced by $\Phi^C(\kappa,n,\alpha)$  is a single-bit,  $(1-2^{-\kappa},1/2+2^{-\kappa})$-key agreement. The running time of $\Phi^C(\kappa,n,\alpha)$ is $\poly(\kappa,n,1/\alpha)$.
	
	Furthermore, the security proof is black-box: there exists a \ppt oracle-aided  \Ec such that  for every $n$-size channel $C$ with $\alpha$-agreement, and every algorithm \tE that violates  the  $(1/2+2^{-\kappa}+\beta)$-equality-leakage of $\tC$, for some $\beta>0$, algorithm  $\Ec^{C,\tE}(\kappa,n,\alpha,\beta)$ violates the equality-leakage of  $C$, and runs in time $\poly(\kappa,n,1/\alpha,1/\beta)$
}

\begin{theorem}[Key-agreement amplification]\label{thm:key-agreement-amp}
	\thmKAmplificationN
\end{theorem}

Equipped with the above results, we are ready to prove  \cref{thm:KAProtocol:IT}.
\begin{proof}[Proof of \cref{thm:KAProtocol:IT}]
	Let  $\Amplifier$ be the protocol guaranteed by \cref{thm:key-agreement-amp}. For  $\ell\in[\pn]$, let $\widehat{C}_\ell$ be the channel induced by $\Pi^C_{\pn,\hell}$ (\cref{protocol:MainProtocol}). Let $c$ be the constant from  \cref{lem:KAProtocol:IT}, and let  $\widetilde{C}_\ell$ be the channel induced by $\Amplifier^{\widehat{C}_\ell}(1^\kappa,\alpha_\ell)$ for $\alpha_\ell=(e^{\con\cdot\eps}\cdot c\cdot\ell)/(8\sqrt{\pn})$. By  \cref{lem:KAProtocol:IT}, there  exists $\hell\in [\pn]$ such that $\widehat{C}_\hell$ is a $(\alpha_\hell,\alpha_\hell/2^{15})$-key-agreement-with-equality-leakage channel. Thus,  \cref{thm:key-agreement-amp} yields that  $\widetilde{C}_\hell$ is a  $(1-2^{-\kappa},1/2+2^{-\kappa})$-key-agreement channel. Using \cref{Theorem:Combiners} to combine the channels $\set{\widetilde{C}_\ell}_{\ell \in [\pn]}$  into a single channel, yields the desired  (full-fledged) key-agreement channel.  
\end{proof}

The rest of this section is dedicated to proving \cref{lem:KAProtocol:IT}.

\paragraph{Proving \cref{lem:KAProtocol:IT}.}
In the following fix  $\kappa\in \N$. For $\ell \in\N$, the following random variables are associated  with a random execution of $\Pi^C_{\pn,\ell}(1^\kappa)$:  let $(X,Y,T)$ be the  output of the call to $\CXYT$ done by the parties,  let $R$ and $V_\ell$ be the value of $r$ and $v$ sent in the execution, let $\OA$ and $\OB$, be the local outputs of \Ac and \Bc, respectively. Finally, let $\hT_\ell \eqdef (X_{R^+},Y_{R^-},T,R,V_\ell)$, and  let $\hC_{\OA\OB\hT_{\ell}}$ denote the channel defined by  the distribution of $(\OA,\OB,\hT_\ell)$.    The proof of the lemma makes  use  of the main result of  \cref{sec:CondensingSV}, stated below.  (In the following recall that $z\flipi = (z_{<i},-z_i ,z_{>i}) $, \ie \ith bit is flipped.)

\def \theoremNoam{
	There exist constants $c_1,c_2>0$ and a poly-time oracle-aided algorithm $\EveDP$ such that the following holds:  let $\pn \in \N$, $\eps\ge 0$ and $\ell\geq \log \pn$,  and let $D$ be a distribution over $\mon\times \mon \times\Ss$. Then for every function $f$ such that
	\begin{align*}
	\ppr{\stackrel{(x,y,t)\gets D}{ r\gets\mon}}{\size{f(\transF)-\ip{x\cdot y, r}}\leq \ell}\geq e^{ c_1\cdot\eps}\cdot c_2\cdot \ell/\sqrt{\pn},
	\end{align*}
	it holds that
	\begin{align*}
	&\ppr{\stackrel{(x,y,t) \gets D}{i \gets [2n]}}{\EveDP^{D,f}(i,(x,y)\flipi, t) = 1} < e^{-\eps}\cdot \ppr{\stackrel{(x,y,t) \gets D}{i \gets [2n]}}{ \EveDP^{D,f}(i,(x,y), t) = 1} - 1/\pn.
	\end{align*}
}

\begin{theorem}[Estimation  to Distinguishing]\label{thm:CondensingSV}
	\theoremNoam
\end{theorem}

Informally,   the existence of  an  adversary $\Eve$ that  violates  the equality-leakage of $\hC_{\OA\OB\hT_{\ell}}$  yields that there exists an algorithm $f$ such that 
$$
\pr{\size{f^\Eve(T,R)-\ip{X\cdot Y,R}}<3\ell \mid \BuckA=\BuckB} >  e^{ c_1\cdot\eps}\cdot c_2\cdot \ell/\sqrt{\pn}.
$$
Very superficially, the above should have allowed us  to use \cref{thm:CondensingSV} for violating the differential privacy of $D \eqdef (C\mid \BuckA=\BuckB)$.  The conditioning on the event $\set{\OA = \OB}$ in the definition of $D$, however, poses two problems: the first  is that there is no guarantee that  $D$ is differentially private (even though $C$ is), and thus the predictor guaranteed by \cref{thm:CondensingSV} does  not yield a contradiction. The second issue is that after the conditioning, the random variable $R$ might no longer be uniform and independent of the other parts $D$ (as required by \cref{thm:CondensingSV}). To overcome these challenges, we consider a \emph{different} distribution that is (1) differentially private, and (2) we have a good inner-product estimator for (with independent and uniform $R$ meeting the requirements in \cref{thm:CondensingSV}). See formal proof below.  

\begin{proof}[Proof of \cref{lem:KAProtocol:IT}]
	Let $\ell\ge\mu$ be such that there exists a (deterministic)  adversary $\Eve$ that  violates  the equality-leakage of $\hC_{\OA\OB\hT_{\ell}}$. That is, 
	\begin{align}\label{eq:Equality leak adv}
	\pr{\Eve(\hT_\ell)=\OA \mid \OA=\OB} >\alpha/2^{15}= \frac{\beta \cdot \ell}{2^{18}\sqrt{\pn}}
	\end{align}
	
	Recall that $\hT_\ell= (X_{R^+},Y_{R^-},T,R,V_\ell)$. By the definition of $\Pi^C_{n,\ell}$, the event $\set{\OA=\OB}$ implies the designated output of the call to $C$ is close to  $\ip{X,Y}$. That is, 
	\begin{align}\label{claim:not equality if delta is far}
	\set{\OA=\OB} \implies \set{\size{\out(T)-\ip{X,Y}} < \ell}
	\end{align}   
	In addition, note that the event $\set{\size{\out(T)-\ip{X,Y}} < \ell}$ implies that $2\cdot(O_A + V_{\ell})-\out(T)$ and $\ip{X \cdot Y,R}$ are at distance  at most $3\ell$. Indeed, 
	\begin{align*}
	\size{2\cdot(O_A + V_{\ell})-\out(T) - \ip{X \cdot Y,R}} 
	&= \size{2\cdot\paren{\floor{\frac{\ip{X_{R^-},Y_{R^-}} - V_{\ell}}{\ell}}\cdot \ell + V_{\ell}}-\out(T) - \ip{X \cdot Y,R}}\\
	&\leq 2\ell + \size{2\cdot \ip{X_{R^-},Y_{R^-}} -\out(T) - \ip{X \cdot Y,R}}\\
	&< 3\ell + \size{2\cdot \ip{X_{R^-},Y_{R^-}} -\ip{X,Y} - \ip{X \cdot Y,R}}\\
	&= 3\ell.
	\end{align*}
	Therefore, by combining  \cref{eq:Equality leak adv,claim:not equality if delta is far}, we obtain that  $f^{\Eve}(\hT_\ell)\eqdef 2(\Eve(\hT_\ell)+V_\ell)-\out(T)$ is an accurate estimation for $\ip{X\cdot Y,R}$. Specifically, for every such $\ell$:
	\begin{align}\label{eq:estimation from leak}
	\pr{\size{f^\Eve(\hT_\ell)-\ip{X\cdot Y,R}}<3\ell \mid \BuckA=\BuckB} > \frac{\beta \cdot \ell}{ 2^{18}\sqrt{\pn}}
	\end{align}
	
	
	Let $\IP_\neps(x,y)\eqdef \ip{x,y}+\lfloor w \rceil$ for  $w\from \Lap(1)$, \ie  $\IP_\neps$ is the Laplace mechanism defined at \cref{theorem:LapIP} and $\lfloor w \rceil$ being the rounding of $w$ to its closes integer. Let  $\DLap$ be the random variable, jointly distributed with $\hC$, defined by 
	\begin{align}
	\DLap=\size{\out(T)-\IP_\neps(X,Y)}	
	\end{align}
	
	We make use of the following key claim, proven below.  
	\begin{claim}\label{claim:different channel}
		There exists an integer $\hell\ge\mu$ and a constant $c>0$ such that the following holds:
		\begin{enumerate} 
			\item $\ppr{\hC_{\OA\OB\hT_{\hell}}}{\OA=\OB}\ge \beta\cdot \hell/8\sqrt{\pn}$. \label{item:claim:agreement} 
			
			\item For every function $f$ such that
			$\pr{\size{f(\hT_\hell)=\ip{X\cdot Y,R}}<3\hell \mid \BuckA=\BuckB} > \frac{ \beta \cdot \hell}{2^{18}\cdot\sqrt{\pn}}$, it holds that
			$\pr{\size{f(\hT_\hell)=\ip{X\cdot Y,R}}<3\hell \mid \DLap<\hell} > \frac{\beta \cdot \hell}{10\cdot2^{22}\cdot\sqrt{\pn}}.$ \label{item:claim:different}
			
			\item
			$\pr{\DLap<\widehat{\ell}}\ge 2/\pn$.\label{item:claim:happens}
			
		\end{enumerate}  
	\end{claim}
	Let $\hell\in\N$ be value  guaranteed by \cref{claim:different channel}. \cref{claim:different channel}(\ref{item:claim:agreement})  yields that  the channel $\hC_{\OA\OB\hT_{\hell}}$ has $(\beta\cdot\hell/8\sqrt{\pn})$-agreement. By \cref{eq:estimation from leak,claim:different channel}(\ref{item:claim:different}), it holds that 
	\begin{align}\label{eq:esti for diff}
	\pr{\size{f^\Eve(\hT_\hell)-\ip{X\cdot Y,R}}<3\hell \mid \DLap<\hell} > \frac{ \beta \cdot \hell}{10\cdot 2^{22}\cdot\sqrt{\pn}}
	\end{align}
	Consider the  function $g = g^{f^\Eve}$ that on input $(r,x_{r^+},y_{r^-},t)$:  (1) samples  $v\from[\hell]$, and (2) outputs $f^\Eve(x_{r^+},y_{r^{-}},t,r,v)$. Since conditioned on $\set{\DLap<\hell}$ the value of both $R$ and $V$ in $\hT_\hell$ are uniform and independent of all other parts of the transcript, \cref{eq:esti for diff} yields that
	\begin{align}\label{eq:esti for diff:g}
	\ppr{r\gets\mon}{\size{g(r,X_{r^+},Y_{r^-},T)-\ip{X\cdot Y, r}}\leq 3\hell \mid \DLap<\hell}
	>\frac{ \beta \cdot \hell}{10\cdot 2^{22}\cdot\sqrt{\pn}}
	\end{align}
	Let $C_{\Lap}$ be the channel $\paren{X,Y,(T,P \eqdef  \IP_2(X,Y))}$ and let $D$ be the distribution $(C_{\Lap}\mid {\DLap<\hell})$.  \cref{eq:esti for diff:g} yields that
	\begin{align}\label{eq:esti for diff:tC}
	\ppr{\stackrel{(x,y,(t,p))\gets D}{ r\gets\mon}}{\size{g(r,x_{r^+},y_{r^-},t)-\ip{x\cdot y, r}}\leq 3\hell}
	>\frac{ \beta \cdot \hell}{10\cdot 2^{22}\cdot\sqrt{\pn}}
	\end{align}
	Hence, there exists a fixed value of $v\in[\hell]$ such that above holds \wrt $g_v$, the variant of $g$ with $v$ hardwired. 
	Recall that $\beta= e^{\con\cdot\eps}\cdot c$. Taking  $c\ge3\cdot e^{2\cdot c_1} \cdot  c_2\cdot 10\cdot2^{22}$, yields that
	\begin{align*}
	\ppr{\stackrel{(x,y,(t,p))\gets D}{ r\gets\mon}}{\size{g_v(\transF)-\ip{x\cdot y, r}}\leq 3\hell}
	>e^{ c_1\cdot(\eps+2)}\cdot c_2\cdot (3\hell)/\sqrt{\pn}.
	\end{align*}
	Thus by \cref{thm:CondensingSV}, it holds that
	\begin{align}\label{eq:KAProtocol:IT:1}
	&\ppr{\stackrel{(x,y,(t,p)) \gets D}{i \gets [2\pn]}}{\EveDP^{D,g_v}(i,(x,y)\flipi, t) = 1} < e^{-(\eps+2)}\cdot \ppr{\stackrel{(x,y,(t,p)) \gets D}{i \from [2\pn]}}{ \EveDP^{D,g_v}(i,(x,y), t) = 1} - 1/\pn
	\end{align}
	for \EveDP being  the  poly-time algorithm guaranteed by \cref{thm:CondensingSV}. Let \EveDPI be the poly-time algorithm that given $(x,y,(t,p))$, outputs $\EveDP(x,y,t)$ if $\size{p-\out(t)}<\hell$, and abort otherwise.  \cref{claim:different channel}(\ref{item:claim:happens} yields that \EveDPI does not abort with probability at least $2/\pn$.  Furthermore, since the decision of \EveDPI whether  to abort or not is a function of  the transcript $(t,p)$,  it holds that 
	\begin{align}\label{eq:KAProtocol:IT:dist}
	&\ppr{\stackrel{(x,y,(t,p)) \gets C_{\Lap}}{i \gets [2\pn]}}{\EveDPI^{D,g_v}(i,(x,y)\flipi, (t,p)) = 1}\\
	&~~ < e^{-(\eps+2)}\cdot \ppr{\stackrel{(x,y,(t,p)) \gets C_{\Lap}}{i \from [2\pn]}}{ \EveDPI^{D,g_v}(i,(x,y), (t,p)) = 1} - 2/\pn^2.\nonumber
	\end{align}
	
	Recall that $\eps=\eps_1+\eps_2$ and that $(X,Y)$ is a strong $e^{-\eps_2}$-$\SV$ source. Thus by combining \cref{eq:KAProtocol:IT:dist,prop:breaking the dp}, we deduce that $C_{\Lap}$ is \emph{not}  $(\eps_1+2,1/\pn^2)$-$\Dp$. Specifically, there exists an advise $z = (i,(x,y))\in[n]\times\oo^{2n}$ such that the oracle-aided algorithm $\EveDP_z(t)\eqdef \EveDPI^{D,g_v}(z,t)$ violates the $(\eps_1+2,1/\pn^2)$-$\Dp$ of $C_{\Lap}$.

By \cref{claim:different channel}(\ref{item:claim:happens}),   oracle access to   $C_{\Lap}$ suffices for  efficiently  emulating  (with negligible probability of failure) the distribution  $D$. Hence, there exits a deterministic, poly-time algorithm, that uses only oracle access to $\C_{\Lap}$ and $g_v$,  for   violating the     $(\eps_1+2,1/\pn^2)$-$\Dp$  of $C_{\Lap}$.
	
	Finally, since $\IP_2$ is a $2$-$\Dp$ mechanism (see \cref{theorem:LapIP}),  by differential privacy composition (see \cref{prop:Composition:IT}) there exists a distinguisher with an advise $a\in\oo^\pn$ and an oracle access to $C$ and $\EveDP_z$, that violates the $(\eps_1,1/\pn^2)$-$\Dp$ of the (original) channel $C$. Putting it all together, we get an oracle-aided \ppt that given oracle access to $C$ and $\Ec$, and the advice $(z,v,a) \in \oo^{3n}$, violates the $(\eps_1,1/\pn^2)$-$\Dp$ of the  channel $C$.
	
\end{proof}

\subsubsection{Proving  \cref{claim:different channel}}\label{Proof of different channel claim}

Let $\Delta \eqdef \size{\out(T) - \ip{X,Y}}$, let   $\cA\eqdef\set{a\in[\pn]\colon\pr{\Delta<a}\ge \frac{a\cdot\beta}{\sqrt{\pn}}}$,   and   let $\amax\eqdef \max(\cA)\leq \sqrt{n}$. We  prove that \cref{claim:different channel} holds for the choice: 
\begin{align}\label{eq:good ell}
\hell=2 \cdot \amax
\end{align}
Since, by the accuracy of the channel,  it holds that $\mu\in\cA$, we deduce that  $\hell\ge \mu$.

We will make use of the following claims:

\begin{claim}\label{claim:agreement vs noise}
	$\pr{\BuckA=\BuckB }\ge \frac14 \cdot \pr{\Delta < \hell}$.
\end{claim}

\begin{claim}\label{claim:from equality to delta}
	Let $f$ be a function such that
	$\pr{\size{f(\hT_\hell)=\ip{X\cdot Y,R}}<3\hell \mid \BuckA=\BuckB} > \frac{\beta \cdot \hell}{2^{18}\cdot\sqrt{\pn}}$. Then 
	$\pr{\size{f(\hT_\hell)=\ip{X\cdot Y,R}}<3\hell \mid \Delta<\hell} > \frac{\beta \cdot \hell}{2^{20}\cdot\sqrt{\pn}}.$
\end{claim}

\begin{claim}\label{claim:DeltaToDlap}
	$\frac{\pr{\Delta< \hell }}{\pr{\DLap< \hell }}\ge 1/5$.
\end{claim}


The proof \cref{claim:agreement vs noise,claim:from equality to delta,claim:DeltaToDlap}  is given below, but first we will use the above claims to prove \cref{claim:different channel}.
\begin{proof}[Proof of \cref{claim:different channel}.]
	Since $\hell=2\cdot \amax$, it holds that
	\begin{align}\label{eq:prob for delta}
	\pr{\Delta<\hell}\ge \pr{\Delta<\amax} \ge\frac{\amax\cdot \beta}{\sqrt{\pn}}\ge \frac{\hell \cdot \beta}{2\cdot\sqrt{\pn}}
	\end{align}
	
	Thus, by \cref{eq:prob for delta,claim:agreement vs noise},  we prove \cref{item:claim:agreement} in the claim statement. Recall that $\Delta = \size{\out(T) - \ip{X,Y}}$ and that $\DLap=\size{\out(T)-\ip{X,Y}- \Gamma }$,  where $\Gamma=\lfloor W \rceil$  and $W$ is sampled from $\Lap(1)$. Note that
	\begin{align}\label{eq:Delta=DLap}
	(\DLap|{\set{\Gamma=0}}) \equiv \Delta
	\end{align}
	In addition, the definition of $\Lap(1)$ readily yields that
	\begin{align}\label{eq:lapCons}
	\pr{\Gamma=0}=\pr{\size{W} \leq \half}\geq 1-e^{-1/2} > \half
	\end{align}
	The second inequality holds by \cref{fact:laplace-concent}.
	It follow that 
	$$\pr{\DLap<\hell}>\pr{\Delta<\hell\mid \Gamma=0}\cdot\pr{\Gamma=0}>\frac{\amax \cdot \beta}{\sqrt{\pn}}\cdot \frac{1}{2}>\frac{2}{\pn},$$
	which satisfies  \cref{item:claim:happens} in the claim.
	Finally, compute
	\begin{align*}
	\lefteqn{\pr{\size{f(\hT_\hell)=\ip{X\cdot Y,R}}<3\hell \mid \DLap<\hell}}\\
	&> \pr{\size{f(\hT_\hell)=\ip{X\cdot Y,R}}<3\hell \mid \DLap<\hell, \Gamma=0}\cdot \pr{\Gamma=0\mid  \DLap<\hell}\\
	&= \pr{\size{f(\hT_\hell)=\ip{X\cdot Y,R}}<3\hell \mid \Delta<\hell}\cdot \frac{\pr{  \DLap<\hell \mid\Gamma=0}\pr{\Gamma=0}}{\pr{\DLap<\hell}}\\
	&= \pr{\size{f(\hT_\hell)=\ip{X\cdot Y,R}}<3\hell \mid \Delta<\hell}\cdot \frac{\pr{  \Delta<\hell}\pr{\Gamma=0}}{\pr{\DLap<\hell}}\\
	&>\pr{\size{f(\hT_\hell)=\ip{X\cdot Y,R}}<3\hell \mid \Delta<\hell}\cdot \frac{1}{10}.
	\end{align*}
	The first and third equalities follow from \cref{eq:Delta=DLap} and the fact that the event $\set{\Gamma=0}$ is independent from $X,Y$ and $\hT_\hell$. The last inequality  holds by \cref{claim:DeltaToDlap} and \cref{eq:lapCons}. Combing the above  inequality with \cref{claim:from equality to delta}, proves \cref{item:claim:different} in \cref{claim:different channel}.
\end{proof}

The remainder of this section is dedicated to proving \cref{claim:agreement vs noise,claim:from equality to delta,claim:DeltaToDlap}. We start by proving  \cref{claim:DeltaToDlap}.

\paragraph{Proving  \cref{claim:DeltaToDlap}.}
\begin{proof}[Proof of \cref{claim:DeltaToDlap}]
	Recall that $\Delta = \size{\out(T) - \ip{X,Y}}$ and that $\DLap=\size{\out(T)-\ip{X,Y}-\Gamma}$,
	where $\Gamma=\lfloor W \rceil$  and $W$ is sampled from $\Lap(1)$.
	It holds that
	\begin{align}\label{eq: bound on DLap}
	\pr{\DLap<\hell}&=\pr{\DLap<\hell,\Delta<2\cdot\hell}+ \pr{\DLap<\hell,\Delta\ge 2\cdot\hell}\\\nonumber
	&\leq\pr{\Delta<2\cdot\hell}+\pr{\size{\Gamma}>\hell}.\nonumber
	\end{align}
	The inequality follows since the event $\set{\DLap<\hell,\Delta\ge 2\cdot\hell}$ implies the event $\set{\size{\Gamma}>\hell}$. 
	Since $\hell\ge\mu\ge \log(\pn)$, we deduce by \cref{fact:laplace-concent} that
	\begin{align}\label{eq:propability_of_Gamma_large}
	\pr{\size{\Gamma}>\hell} = \pr{\size{\Lap(1)} > \hell} \le \frac{1}{\pn}
	\end{align} 
	On the other hand, since we set 	$\hell=2\max\set{a\in[\pn]\colon\pr{\Delta<a}\ge \frac{a\cdot\beta}{\sqrt{\pn}}}$, it holds that
	\begin{align}\label{eq:bound on delta}
	4/n<\pr{\Delta<2\cdot\hell}<4\cdot\pr{\Delta<\hell/2}
	\end{align}
	Combining \cref{eq:bound on delta,eq:propability_of_Gamma_large,eq: bound on DLap}, we conclude that
	\begin{align*}\nonumber
	\frac{\pr{\Delta< \hell }}{\pr{\DLap< \hell }}&\ge	\frac{\pr{\Delta< \hell }}{\pr{\Delta<2\cdot\hell}+\pr{\size{\Gamma}>\hell}} \ge \frac{\pr{\Delta< \hell/2 }}{4\cdot\pr{\Delta<\hell/2}+1/\pn} \ge \frac{1}{5}.
	\end{align*}
\end{proof}

\paragraph{Proving  \cref{claim:agreement vs noise,claim:from equality to delta}.}

We make use of the following claim.

\begin{claim}\label{claim:Delta and equality}
	$\pr{\OA=\OB\mid \Delta<\hell/2}\ge 1/2$.
	\remove{
		\begin{enumerate}
			\item $\pr{\BuckA=\BuckB \mid \Delta \geq \hell } = 0$, and \label{item:noteq}
			\item . \label{item:more than half}
		\end{enumerate}
	}
\end{claim}

\begin{proof}[Proof of 	\cref{claim:Delta and equality}]
	Let $\oDelta\eqdef\out(T)-\ip{X,Y}$  and note that $\Delta=\size{\oDelta}$. Let
	$Z=\ip{X_{R^-},Y_{R^-}}-V$, by construction it holds that
	\begin{align}
	&\OA=\floor{\frac{Z}{\hell}}\cdot\hell  & \text{and}&    & \OB=\floor{\frac{\oDelta+Z} {\hell}}\cdot \hell,
	\end{align}
	\remove{which implies that  $\OA=\OB$ iff $\floor{\frac{Z}{\hell}}=\floor{\frac{\oDelta+Z} {\hell}}$. It follows that 
		\begin{align*}
		\pr{\OA=\OB \mid \Delta \geq \hell } = \pr{\floor{\frac{Z}{\hell}}=\floor{\frac{\oDelta+Z} {\hell}}\mid \Delta \geq \hell }=0,
		\end{align*}
		concluding the proof of \cref{item:noteq}. }

	Let $\ModA\eqdef (Z \mod \hell)$. Since $V$ is uniform over $[\hell]$ and independent from $X$ and $Y$, 
	and since $\size{\bbZ \cap [0,\hell/2)} = \size{\bbZ \cap [\hell/2,\hell)} = \hell/2$ (holds since $\hell/2$ is an integer), 
	we deduce that
	\begin{align}
	\pr{\ModA<\hell/2}=\pr{\ModA\ge\hell/2}=\half
	\end{align}
	Moreover,
	if $\ModA<\hell/2$ and $0\le\oDelta<\hell/2 $, then $\floor{\frac{Z}{\hell}}=\floor{\frac{\oDelta+Z} {\hell}}$ (and $\OA=\OB$).
	Thus,
	\begin{align}\label{eq:Delta and equality:1}
	\lefteqn{\pr{\OA=\OB\mid 0\le\oDelta< \hell/2}}\\
	&\ge \pr{\OA=\OB\mid 0\le\oDelta< \hell/2,\ModA< \hell/2}\cdot \pr{\ModA <\hell/2\mid 0\le\oDelta< \hell/2}\nonumber\\
	&= 1\cdot\pr{\ModA <\hell/2\mid 0\le\oDelta< \hell/2}\nonumber\\
	&= \pr{\ModA<\hell/2}\nonumber\\
	&= 1/2.\nonumber
	\end{align}
	The penultimate equation holds since $V$ is uniform over $[\hell]$ and independent of $\oDelta$. A similar argument yields that
	\begin{align}\label{eq:Delta and equality:2}
	\lefteqn{\pr{\OA=\OB\mid \hell/2\le\oDelta< 0}}\\
	&\ge \pr{\OA=\OB\mid \hell/2\le\oDelta< 0,\ModA\ge \hell/2}\cdot \pr{\ModA \ge\hell/2\mid \hell/2\le\oDelta< 0}\nonumber\\
	&\ge \pr{\ModA\ge\hell/2}= 1/2.\nonumber
	\end{align}
	Combining \cref{eq:Delta and equality:1,eq:Delta and equality:2}, yields that 	
	\begin{align*}
	\pr{\OA=\OB\mid \Delta< \hell/2}&\ge \pr{\OA=\OB\mid 0\le\oDelta< \hell/2}\cdot \pr{0\le\oDelta< \hell/2\mid \Delta<\hell/2}\\
	&+\pr{\OA=\OB\mid \hell/2\le\oDelta< 0}\cdot \pr{\hell/2\le\oDelta< 0\mid \Delta<\hell/2}\\
	&\ge 1/2,
	\end{align*}
	which concludes the proof of \cref{claim:Delta and equality}.
\end{proof}

\paragraph{Proving  \cref{claim:agreement vs noise}.}
\begin{proof}[Proof of \cref{claim:agreement vs noise}]
	Recall that by definition $\amax$ is the largest element in the set  $\cA=\set{a\in[\pn]\colon\pr{\Delta<a}\ge \frac{a\cdot \beta}{\sqrt{\pn}}}$  and that $\hell=2 \cdot \amax$. Thus,
	\begin{align}\label{Item:great >1/2}
	\pr{\Delta< \hell/2\mid \Delta<\hell}&=\pr{\Delta< \amax\mid \Delta<2\cdot \amax}\\\nonumber
	&=\frac{\pr{\Delta< \amax}}{\pr{\Delta< 2\cdot \amax }}\\
	&>1/2.\nonumber 
	\end{align}
	The inequality holds since  otherwise,  we have that $\pr{\Delta< \amax}\le2\cdot \pr{\Delta< 2\cdot \amax}$ which  (since $\amax\in\cA$) implies that $\hell=2\cdot \amax\in\cA$, contradicting the maximality of $\amax$.
	
	By \cref{Item:great >1/2} and  \cref{claim:Delta and equality}, it follows that:
	\begin{align*}
	\pr{\OA=\OB}&=\pr{\OA=\OB\mid \Delta<\hell}\cdot\pr{\Delta < \hell}\\
	&\ge\pr{\OA=\OB\mid \Delta<\hell/2}\cdot\pr{\Delta<\hell/2\mid \Delta<\hell}\cdot\pr{\Delta < \hell}\\
	&\ge \frac14 \cdot \pr{\Delta < \hell}.
	\end{align*}
\end{proof}

\paragraph{Proving \cref{claim:from equality to delta}.}

\begin{proof}[Proof of \cref{claim:from equality to delta}]
	The claim immediately holds by observing that:
	\begin{align}\label{eq:Eve_upper_bound}
	\lefteqn{\pr{\size{f(\hT_\hell)=\ip{X\cdot Y,R}}<3\hell \mid \OA=\OB}} \\
	&=\frac{1}{\pr{\OA=\OB}}\cdot\pr{\size{f(\hT_\hell)=\ip{X\cdot Y,R}}<3\hell \land (\OA=\OB)}\nonumber\\
	&=\frac{1}{\pr{\OA=\OB}}\cdot\pr{\size{f(\hT_\hell)=\ip{X\cdot Y,R}}<3\hell\land  (\OA=\OB) \land (\Delta < \hell)}\nonumber\\
	&\leq\frac{1}{\pr{\OA=\OB}}\cdot\pr{\size{f(\hT_\hell)=\ip{X\cdot Y,R}}<3\hell\land (\Delta < \hell)}\nonumber\\
	&=\frac{\pr{\Delta < \hell}}{\pr{\OA=\OB}}\cdot\pr{\size{f(\hT_\hell)=\ip{X\cdot Y,R}}<3\hell\mid \Delta < \hell}\nonumber\\
	&\leq 4 \cdot \pr{\size{f(\hT_\hell)=\ip{X\cdot Y,R}}<3\hell\mid \Delta < \hell}.
	\end{align}
	The second equality holds by \cref{claim:not equality if delta is far}, and the last one by  \cref{claim:agreement vs noise}.
\end{proof}

\subsection{The Computational Case}\label{sec:KAProtocol:Comp}
In this section we state and prove our results for the computational case:   \CDP (computational differential private) protocols that estimate the inner product well. For such protocols, we prove the following result.

\begin{theorem}[Key-agreement from differentially private channels estimating the inner product, the computational case, restatement of \cref{thm:intro:main}]\label{thm:KAProtocol:Comp}
	There exists an oracle-aided protocol $\Lambda$ and a universal constant $c>0$, such that the following holds for every protocol $\Psi$ that is $\eps$-\CDP against external observer. If $\Psi$ is  $(\mu(\kappa),e^{\con\cdot\eps(\kappa)}\cdot c\cdot \mu(\kappa)/\sqrt{n(\kappa)})$-accurate for the inner-product functionality on inputs of length $n$, for some $\mu(\kappa)\ge \log n(\kappa)$, then $\Lambda^\Psi$ is a  (full fledged) key-agreement protocol. \footnote{The theorem extends to $(\eps(\kappa),1/\pn(\kappa)^2)$-$\CDP$ channels.} \footnote{The theorem extends accurate on \emph{average} protocols: \ie the probability of inaccuracy is small over uniformly chosen inputs.}
\end{theorem}

\cref{thm:KAProtocol:Comp} is an immediate  corollary of the following  key lemma. Let \Amplifier be the key-agreement amplifier  guaranteed by \cref{thm:key-agreement-amp}, and let \Combiner be the key-agreement combiner guaranteed by \cref{Theorem:Combiners}. To avoid notational  cluttering, in the following we omit $\kappa$ when clear from the context.

\begin{lemma}[Main lemma,  the computational case]\label{lem:KAProtocol:Comp}
	There exists a constant $c>0$ such that the following holds: let  $C = \set{C_\kappa}_{\kappa\in \N}$ be an  $n$-size, $\eps$-\CDP channel ensemble, that is  $(\mu,e^{\con\cdot\eps}\cdot c\cdot \mu/\sqrt{n})$-accurate for the inner-product functionality, for some $\mu\ge \log n$, and let $\Pi$ be according to  \cref{protocol:MainProtocol}.  Let $\Gamma_{n,\ell} \eqdef \Pi_{n,\ell}^{C_\kappa}$, let $\Gamma^\Amplifier_{n,\ell}  \eqdef \Amplifier^{\Gamma_{n,\ell}}(\kappa,n,\alpha(\ell))$,  for  $\alpha(\ell)\eqdef (\beta\cdot\ell)/(8\sqrt{n})$, and  let $\Gamma^\Combiner  \eqdef \Combiner^{\set{\Gamma^\Amplifier_{n,\ell}}_{\ell \in [n]}}(1^\kappa,1^n)$. Then $\Gamma^\Combiner$ is a key-agreement protocol. 
\end{lemma}
We prove \cref{lem:KAProtocol:Comp} by using (the ``information theoretic'')  \cref{lem:KAProtocol:IT} to show  that for  the right choice of $\ell$, protocol $\Gamma_{n,\ell} =\Pi_{n,\ell}^{C_\kappa}$ is a weak key-agreement protocol, and hence, procedure  \Amplifier turns it into a full-fledged key agreement $\Gamma^\Amplifier_{n,\ell}$. It follows that applying the above procedure  for  all  $\ell\in [n]$, yields the set of protocols $\set{\Gamma^\Amplifier_{n,\ell}}_{\ell \in [n]}$ that \emph{contains} a    key-agreement protocol. Applying \Combiner on this set, yields the desired  key-agreement protocol $\Gamma^\Combiner$.   \cref{lem:KAProtocol:Comp} is formally proved below, but we first use it for proving \cref{thm:KAProtocol:Comp}.

\paragraph{Proving \cref{thm:KAProtocol:Comp}.}

For using \cref{lem:KAProtocol:Comp}, we first convert protocol   $\Psi$ into a (no private input) protocol such that the  \CDP-channel it induces,    accurately estimate the inner-product functionality. Such a transformation is simply the following protocol that invokes $\Psi$  over uniform inputs, and each party locally outputs its input. 
\begin{protocol}[$\hPsi = (\hAc,\hBc)$]\label{prot:EDPtoSV}
	\item Common input: $1^\kappa$.
	\item Operation:
	\begin{enumerate}
		
		\item  $\hAc$ samples $x \gets \oo^{\pn(\kappa)}$ and  $\hBc$ samples $y\gets \oo^{\pn(\kappa)}$. 
		
		\item The parties interact in a random execution protocol $\Psi(1^\kappa)$, with   $\hAc$ playing the role of $\Ac$ with private input $x$, and $\hBc$ playing the role of $\Bc$ with private input $y$.

		\item  $\hAc$ locally outputs  $x$ and  $\hBc$ locally outputs $y$. 
	\end{enumerate}
\end{protocol}

\newcommand{\hhC}{\hC}

Let $\hhC$ be the channel ensemble induced by $\hPsi$, letting its designated output (the function $\out$) be the designated output of the embedded execution of $\Psi$. The following fact is immediate by definition.

\begin{proposition}\label{prop:EDP to SV}
	The channel ensemble   $\hhC$  is   $(\eps,\delta)$-$\CDP$,  and has the same accuracy for computing the inner product as protocol $\Psi$ has.
\end{proposition}

\begin{proof}[Proof of \cref{thm:KAProtocol:Comp}]
	Immediate by \cref{lem:KAProtocol:Comp,prop:EDP to SV} 
\end{proof}

\newcommand{\tEc}{\widetilde{\Ec}}

\subsubsection{Proving \cref{lem:KAProtocol:Comp} }
\begin{proof}[Proof of \cref{lem:KAProtocol:Comp}]
	Assume towards a contradiction that there exits \ppt \Ec that for infinity often $\kappa\in\N$  breaks the security of $\Gamma^\Combiner(1^\kappa,1^n)$ with probability $1/p(\kappa)$, for some $p\in \poly$. Fix such  $\kappa\in\N$ and omit it from notation when clear from the context. The proof follows by the following steps:

	\begin{enumerate}
		
		\item Recall that $\Gamma^\Combiner  =\Combiner^{\set{\Gamma^\Amplifier_{\ell}}_{\ell \in [n]}}(1^n)$, and let  $\hEc$ be the  \pptm  (\ie black-box reduction)   guaranteed by \cref{Theorem:Combiners}. By the contradiction assumption, for every $\ell\in[n]$, $\hE(\ell)\eqdef \hEc^{\Ec,{\set{\Gamma^\Amplifier_{\ell}}_{\ell \in [n]}}}(1^n,\ell)$ violates the $1/p'(\kappa)$-secrecy of $\Gamma^\Amplifier_{\ell}$, for some $p'\in \poly$.
		
		\item Recall that $\Gamma^\Amplifier_{\ell}  = \Amplifier^{\Gamma_{n,\ell}}(n,\alpha(\ell))$, and let $\tEc$ be the algorithm guaranteed by \cref{thm:key-agreement-amp}. By the above,  for every $\ell\in[n]$, $\tEc(\ell) \eqdef \tEc^{\hEc(\ell),{\Gamma_{\ell}}}(1^\kappa,1/2p'(\kappa))$ runs in polynomial time, and violates the $\alpha(\ell)$-secrecy of  $\Gamma_{n,\ell}$, for $\alpha(\ell)\eqdef (2^{\con\cdot\eps(\kappa)}\cdot c'\cdot\ell)/(8\sqrt{n})/15$.

		\item For each $\ell \in [n]$, use polynomial number of sampling to find, with save but negligible failure probability, a random string $r_\ell$ such that  $\tEc(\ell;r_\ell)$ violates the $\alpha(\ell)$-secrecy of  $\Gamma_{n,\ell}$. Let $\Ec^\ast$ be deterministic algorithm that on input $\ell$  acts like $\tEc(\ell;r_\ell)$.

		\item By \cref{lem:KAProtocol:IT}, recalling that $\Gamma_{n,\ell} =\Pi_{n,\ell}^{C_\kappa}$,  there exists a \pptm $\Ac^{\Ec{^\ast}}$  such the the following holds: there exits $\hell\in[n]$ and (advise) $a_\kappa\in\oo^{2n}$ such that $\Ac^{\Ec{^\ast}}(1^\kappa,\hell,a_\kappa)$ violates the $(\eps,1/n^2)$-$\Dp$ of $C_\kappa$.
	\end{enumerate}
	Since we assumed (toward contradiction) that the above holds for infinitely many $\kappa$'\pn,  the algorithm that for every $\kappa\in\N$, gets $(\hell,a_\kappa)$ as non-uniform advice and runs $\Ac^{\Ec{^\ast}}(1^\kappa,\hell,a_\kappa)$, violates the assume $\eps$-$\CDP$ of the ensemble $C$. This  concludes the proof.
\end{proof}  

%% file: CondensingSantaVazirani.tex
\newcommand{\tO}{{\widetilde{O}}}

\section{Condensing Santa-Vazirani Source using Source-Dependent Seed}\label{sec:CondensingSV}

In this section, we prove \cref{thm:CondensingSV}, restated below. Recall that for a string $z\in \mo^{2n}$ and an index $i\in [2n]$, we denote $z_{<i>} \eqdef (z_1,\dots,z_{i-1}, -z_i, z_{i+1}, \dots, z_{2n})$.
\begin{theorem}[Estimation to Distinguishing][Restatement of \cref{thm:CondensingSV}]\label{thm:CondensingSVRes}
	\theoremNoam
\end{theorem}

That is, given an oracle to a function $f$ that estimates the inner product $\ip{x\cdot y, r}$ well, \EveDP distinguishes, for most $i$'s, between $(x,y)$ and its variant in which the \ith bit is flipped. \cref{thm:CondensingSVRes} immediately yields the following corollary, proven in \cref{sec:missing-proofs:CondensingSV}. 

\def\defCondCor{
	Let $\Cc\colon (\mon)^3 \mapsto \Z$ be defined by $\Cc(x,y,r) \eqdef \ip{x\cdot y,r}$. Then for every $\eps>0$ and any $e^{-\eps}$-strong SV source $(X,Y)$ over $(\mon)^2$ and $R\gets \mon$, it holds that for every $0 \leq \delta \leq 1$: $$\ppr{(x,y,r)\gets (X,Y,R)}{\Hmin(\Cc(X,Y,R)|_{(R,X_{R^+},Y_{R^-})=(r,x_{r^+},y_{r^-})}) \geq \log \paren{\frac{\delta \sqrt{n}}{c_2 \cdot e^{c_1\eps} \cdot \log n}}}\geq 1-\delta,$$ where $c$ is the constant from \cref{thm:CondensingSVRes}.
}
\begin{corollary}[Restatement of \cref{thm:intro:condenser}]\label{cor:CondensingSV}
	\defCondCor\footnote{A similar statement holds for $\Cc(x,y,r) \eqdef \ip{x\cdot y,r} \mod \sqrt{n}\cdot \log n$.
		}\footnote{\Nnote{Added:}Since the proof is by black-box reduction, it automatically applies  to computational strong SV sources.	}
\end{corollary}
Namely, the inner product is a good strong seeded condenser for such SV source, even when significant seed related information (\ie $X_{R^+},Y_{R^-}$) is leaked. Since clearly $\Hmin(\Cc(X,Y,R)) \leq \log {\sqrt n} + O(1)$, the above result is tight up to $c_1\eps + \loglog n$ additive term. The rest of this section is devoted for proving \cref{thm:CondensingSVRes}. The proof uses the following key lemma (which in turn proven using the main result of \cref{sec:reconstruction}).

\begin{lemma}\label{lemma:distinguisher_for_distribution}
	There exist \ppt algorithms $\Ac_1$, $\Ac_2$ and $\Ac_3$, and $n_0\in \N$ such that the following holds for every $n\geq n_0$, $\ell \in \N$ and $\eps\ge 0$: let $Q$ be a distribution over $\mon \times \mon \times \Ss$ and let $f$ be a function such that for every $(x,y,t)\in \Supp(Q)$:
	\begin{align*}
	&\ppr{\ptr \gets \mon}{\size{f(\rr, \px_{\rr^+}, \py_{\rr^-},t) - \ip{\px \cdot \py, \rr}} < \ell} \geq \frac{1024\cdot e^\eps \cdot \ell}{\sqrt{n}}.
	\end{align*}
	Then $\exists \Ac \in \set{\Ac_1,\Ac_2,\Ac_3}$ such that 
	\begin{enumerate}
		\item $\ppr{(x,y,t)\gets Q, i\gets [2n]}{ \Ac^{f}(i,(x,y),t,\ell) = 1} \geq e^{-\eps}/16$, and
		
		\item $\ppr{(x,y,t)\gets Q, i\gets [2n] }{\Ac^{f}(i,(x,y)\flipi,t,\ell)= 1} \leq 1/2\cdot e^{-\eps}\cdot \ppr{(x,y,t)\gets Q, i\gets [2n]}{ \Ac^{f}(i,(x,y),t,\ell) = 1}$.
	\end{enumerate}
\end{lemma}
That is, \cref{lemma:distinguisher_for_distribution} essentially proves \cref{thm:CondensingSVRes} for a distribution $Q$ for which $f$ is a good estimator of $\ip{x\cdot y, R}$ for \emph{all} $(x,y,t)\gets Q$. We prove \cref{lemma:distinguisher_for_distribution} below, and use it to prove \cref{thm:CondensingSVRes} in \cref{sec:proveCondensingThm}. 

\subsection{Proving \cref{lemma:distinguisher_for_distribution}}\label{sec:istinguisher_for_distribution}
We prove \cref{lemma:distinguisher_for_distribution} using the following theorem, proved in \cref{sec:reconstruction}.

\def\defEstimator{
	Let $n,\ell \in \N$, let $\lambda > 0$ and let $\zz \in \oo^n$. A function $f\colon \oo^n \mapsto \Z$ is an {\sf $(\lambda,\ell)$-estimator of $\ip{\zz,\cdot}$} if 
	\begin{align*}
		\ppr{\rr \la \oo^n}{\size{ f(\rr) - \ip{\zz,\rr}} < \ell} \geq \frac{\lambda \ell}{\sqrt{n}}\enspace.
	\end{align*}
}

\begin{definition}[Inner-product  estimator]\label{def:estimator}
	\defEstimator
\end{definition}

\def\EInd{j}

\def \theoremEliadi{
	There exists a \pptm $\Pc$ that outputs a value in $\set{-1,0,1}$ such that the following holds for large enough $n \in \N$: let $\zz \in \oo^n$, let $\lambda \geq 64$, and let $f$ be an $(\lambda,\ell)$-estimator of $\ip{\zz,\cdot}$. Then with probability at least $(1-\frac{4096}{\lambda^2})$ over $\EInd \gets [n]$, it holds that
	\begin{align*}
		\zz_\EInd \cdot \eex{\rr \la \oo^n}{\Pc(\EInd ,\zz_{-\EInd },\rr,f(\rr),\ell)} \geq \frac{\lambda}{8 n^{1.5}}\enspace,
	\end{align*}
	where the expectation is also over the randomness of $\Pc$.
}

\begin{theorem}\label{thm:reconstruction}
	\theoremEliadi
\end{theorem}

Let $\ell$, $\eps$, $Q$ and $f$ be as in \cref{lemma:distinguisher_for_distribution}, and let \Pc be \pptm guaranteed in \cref{thm:reconstruction}. The following algorithm reconstructs the \jth bit of $(x\cdot y)$, for $(x,y,t)\gets Q$, given only oracle access to the function:
\begin{align}\label{eq: def of G}
\Gc_{x,y,t}(j,r) \eqdef \Pc(j,(\px \cdot \py)_{-j}, f(\rr,\px_{\rr^+},\py_{\rr^-},t), \ell).
\end{align}

\begin{algorithm}[The reconstruction algorithm \RecBit]\label{alg:CondensingSV:Rec}
	
	\item Oracle: $\Gc_{x,y,t}$.
	
	\item Input: $j \in [n]$.
	
	\item Operation:~
	\begin{enumerate}

		\item Sample uniform $(r_1,...,r_{n^4}) \gets (\mo^n)^{n^4}$, let $\cR \eqdef \set{r_i}_{i\in [n^4]}$. 
		
		\item For every $r \in \cR$, let $g_{x,y,t}(j,r) = \Gc_{x,y,t}(j,r)$.
		
		\item Return $\sign(\eex{r\gets \cR}{g_{x,y,t}(j,r)})$.
		
	\end{enumerate}
	
\end{algorithm}
 We next prove that \RecBit has good success probability in reconstructing $(x\cdot y)_j$, for $i \gets [n]$ and $(x,y,t) \gets Q$. 
\begin{claim}\label{clm:condensingSV:Rec_succ}
	For large enough $n\in \N$, it holds that
	\begin{align*}\
		&\ppr{(x,y,t)\gets Q, j \gets [n]}{\RecBit^{\Gc_{x,y,t}}(j)=x_j\cdot y_j} \geq 1- e^{-2\eps}/16.
	\end{align*}
\end{claim}
\begin{proof}
	We assume \wlg that $\eps \leq \log n$ (otherwise the claim follows trivially). The proof is immediate by the Hoffeding bound, \cref{thm:reconstruction} and the observation that $f_{x,y,t}(r) \eqdef f(\rr,\px_{\rr^+},\py_{\rr^-},t)$ is an $(1024\cdot e^{\eps},\ell)$-estimator of $\ip{x\cdot y,\cdot}$.
\end{proof}
The next claim essentially yields that \Rec distinguishes between $(x,y)$ and $(x,y)\flipi$, for some $i \in [2n]$.

\begin{claim}\label{clm:condensingSV:Rec_fail}
	 For every $n\in \N$, at least one of the following holds: 
	\begin{enumerate}
		\item $\ppr{(x,y,t)\gets Q, j \gets [n]}{\RecBit^{\Gc_{x,y,t}}(j)=x_j\cdot y_j} \le 3/4$.
		
		\item $\ppr{(x,y,t)\gets Q, j \gets [n]}{\RecBit^{\Gc_{x\flipj,y,t}}(j)=-x_j\cdot y_j} \le 3/4$.
		
		\item $\ppr{(x,y,t)\gets Q, j \gets [n]}{\RecBit^{\Gc_{x,y\flipj,t}}(j)=-x_j\cdot y_j} \le 3/4$.
		
		\item $\ppr{(x,y,t)\gets Q, j \gets [n]}{\RecBit^{\Gc_{x\flipj,y\flipj,t}}(j)=x_j\cdot y_j} \le 3/4$.		
	\end{enumerate}
\end{claim}
\begin{proof}
	By definition of $\Gc$ (see \cref{eq: def of G}), for every $r$ with $r_j=1$, it holds that
	\begin{align}\label{eq:condensing:recostruct:c1}
		&\Gc_{x,y,t}(j,r) \equiv \Gc_{x,y\flipj,t}(j,r), \text{ \ \ and\ \ }\Gc_{x\flipj,y,t}(j,r) \equiv \Gc_{x\flipj,y\flipj,t}(j,r)
	\end{align}
	Similarly, for every $r$ with $r_j=-1$, it holds that
	\begin{align}\label{eq:condensing:recostruct:c2}
		&\Gc_{x,y,t}(j,r) \equiv \Gc_{x\flipj,y,t}(j,r), \text{ \ \ and\ \ }\Gc_{x,y\flipj,t}(j,r) \equiv \Gc_{x\flipj,y\flipj,t}(j,r)
	\end{align}

Fix $(x,y,t) \in \Supp(Q)$, $j\in [n]$, and the randomness of \RecBit (including the part uses in the call to \Gc). We prove that for at least one of the possible assignments to $(u,v) \in \set{(x,y), (x\flipj,y), (x,y\flipj),(x\flipj, y\flipj)}$, algorithm $\Rec^{\Gc_{u,v,t}}(j)$ fails to output the value of $u_j\cdot v_j$. 
	
Let $\cR$ be the value sampled by $\Rec(j)$, and for a pair $(u,v)$ and $r\in \cR$, let $g_{u,v,t}(j,r)$ be the value set by $\Rec(j)$ (all values \wrt the above fixing). Note that for a pair $(u,v)$, it holds that 
	\begin{align}
		&\eex{r\gets \cR}{g_{u,v,t}(j,r)}\\ \nonumber
		& = \ppr{r\gets \cR}{r_j=-1}\cdot\eex{r\gets \cR|_{r_j=-1}}{g_{u,v,t}(j,r)}+\ppr{r\gets \cR}{r_j=1}\cdot\eex{r\gets \cR|_{r_j=1}}{g_{u,v,t}(j,r)}.
	\end{align}
Let
	\begin{align*}
		&\alpha^{u,v}_{j,\cX} \eqdef \ppr{r\gets \cR}{r_j=-1}\cdot\eex{\rr \la \oo^n|_{r_j=-1}}{g_{u,v,t}(j,r)}, \text{\ \ \ and}\\
		&\alpha^{u,v}_{j,\cY} \eqdef \ppr{r\gets \cR}{r_j=1}\cdot\eex{\rr \la \oo^n|_{r_j=1}}{g_{u,v,t}(j,r)}.
	\end{align*}	
\cref{eq:condensing:recostruct:c1} yields that
	\begin{align}
		&\alpha^{x,y}_{j,\cY} = \alpha^{x,y\flipj}_{j,\cY}, \text{ \ \ and\ \ }\alpha^{x\flipj,y}_{j,\cY} = \alpha^{x\flipj,y\flipj}_{j,\cY}
	\end{align}
	Similarly, \cref{eq:condensing:recostruct:c2} yields that
	\begin{align}
		&\alpha^{x,y}_{j,\cX} = \alpha^{x\flipj,y}_{j,\cX}, \text{ \ \ and\ \ }\alpha^{x,y\flipj}_{j,\cX} = \alpha^{x\flipj,y\flipj}_{j,\cX}
	\end{align}
	Combining the above two equations, we get that 
	\begin{align}\label{eq:condensingSV:rhombus}
		(\alpha^{x,y}_{j,\cX}+\alpha^{x,y}_{j,\cY}) + (\alpha^{x\flipj,y\flipj}_{j,\cX}+\alpha^{x\flipj,y\flipj}_{j,\cY}) = (\alpha^{x\flipj,y}_{j,\cX}+\alpha^{x\flipj,y}_{j,\cY}) + (\alpha^{x,y\flipj}_{j,\cX}+\alpha^{x,y\flipj}_{j,\cY})
	\end{align}
	By definition, $\RecBit^{\Gc_{u,v}}(j)$ outputs $1$ iff $\alpha^{u,v}_{j,\cX}+\alpha^{u,v}_{j,\cY} \geq 0$. Assume towards a contradiction that, for the fixed randomness above, for every $(u,v) \in \set{(x,y), (x\flipj,y), (x,y\flipj),(x\flipj, y\flipj)}$ it holds that $\RecBit^{\Gc_{u,v}}(j)=u_j\cdot v_j$. Assume for simplicity that $x_j\cdot y_j =1$ (the case $x_j\cdot y_j =-1$ is symmetric). It follows that $(\alpha^{x,y}_{j,\cX}+\alpha^{x,y}_{j,\cY}) + (\alpha^{x\flipj,y^y}_{j,\cX}+\alpha^{x\flipj,y\flipj}_{j,\cY})\geq 0$ and $(\alpha^{x\flipj,y}_{j,\cX}+\alpha^{x\flipj,y}_{j,\cY}) + (\alpha^{x,y\flipj}_{j,\cX}+\alpha^{x,y\flipj}_{j,\cY})<0$, in contradiction to \cref{eq:condensingSV:rhombus}.
	
	Since for every fixing of $(x,y,t)$, $j$, and its randomness, \Rec errs on at least one of the cases appearing in the claim statements, we conclude that for (at least) one of the cases, it errs with probability at least $1/4$, over a random choice of $(x,y,t)$, $j$ and its random coins. 
\end{proof}

Equipped with the above claim, we prove \cref{lemma:distinguisher_for_distribution} \wrt  algorithms $\Ac_1,\Ac_2,\Ac_3$, defined below using the following algorithm. 
\begin{algorithm}[The algorithm \Ac]\label{alg:CondensingSV:dist}
	
	\item Oracle: $f$.
	
	\item Input: $i \in [2n], (x,y) \in \mon \times \mon , t\in \Ss, \ell \in [n], \cI \subseteq [2n]$.
	
	\item Operation: If $i \notin \cI$, output 0. Otherwise, 
	\begin{enumerate} 	
		
		\item Let $j = \begin{cases} i & i \leq n \\ i-n & i > n\end{cases}$.
		\item Emulate $\RecBit^{\Gc_{x,y,t}}(j)$, for $\Gc_{x,y,t}\eqdef \Pc(j,(\px \cdot \py)_{-j}, f(\rr,\px_{\rr^+},\py_{\rr^-},t), \ell)$. Let $d$ be its output.

		\item If $d \neq x_j\cdot y_j$, output $1$. Otherwise, output $0$.
		
	\end{enumerate}
	
\end{algorithm}
Let $j(i) \eqdef i$ if $i \leq n$ and $i-n$ otherwise, and let 
\begin{itemize}
	\item $\Ac_1(i,(u,v),t,\ell)\eqdef A(i,(u,v)\flipi,t,\ell,[n])$,
	\item $\Ac_2(i,(u,v),t,\ell)\eqdef A(i,(u,v)\flipi,t,\ell,[2n]\setminus [n])$, and
	\item $\Ac_3(i,(u,v),t,\ell)\eqdef A(i,(u_{<j(i)>},v_{<j(i)>}),t,\ell,[n])$. 
\end{itemize}
 
\begin{proof}[Proof of \cref{lemma:distinguisher_for_distribution}]
	Let $n_0\in \N$ be large enough for \cref{thm:reconstruction}. Let $n\ge n_0$, and let $Q$,$f$, $\eps$ and $\ell$ be as in 	\cref{lemma:distinguisher_for_distribution}. \cref{clm:condensingSV:Rec_succ} yields that 	
	\begin{align*}
		&\ppr{(x,y,t)\gets Q, j \gets [n]}{\RecBit^{\Gc_{x,y,t}}(j)=x_j\cdot y_j} \geq 1- e^{-2\eps}/16 > 3/4.
	\end{align*}
	Thus, \cref{clm:condensingSV:Rec_fail} yields that (at least) one of the following holds:
	\begin{enumerate}
		\item $\ppr{(x,y,t)\gets Q, j \gets [n]}{\RecBit^{\Gc_{x\flipj,y,t}}(j)=-x_j\cdot y_j} \le 3/4$.

		\item $\ppr{(x,y,t)\gets Q, j \gets [n]}{\RecBit^{\Gc_{x,y\flipj,t}}(j)=-x_j\cdot y_j} \le 3/4$.
		
		\item $\ppr{(x,y,t)\gets Q, j \gets [n]}{\RecBit^{\Gc_{x\flipj,y\flipj,t}}(j)=x_j\cdot y_j} \leq 3/4$.		
	\end{enumerate}
	
	The proof continues by case analysis:

	\paragraph{Case $1$: $\ppr{(x,y,t)\gets Q, j \gets [n]}{\RecBit^{\Gc_{x\flipj,y,t}}(j)=-x_j\cdot y_j} <1-e^{-\eps}/8$.} In this case $\Ac_1$ fulfills the requirements of the lemma. Indeed,
	\begin{align*}
		\ppr{(x,y,t)\gets Q, i\gets [2n], a \gets \Ac_1^{f}(j,(x,y)\flipi,t,\ell)}{ a = 1}&= \ppr{(x,y,t)\gets Q, i\gets [2n]}{i \in [n] \land \RecBit^{\Gc_{x,y,t}}(i)\neq x_i\cdot y_i} \\
		& \leq e^{-2\eps}/32,
	\end{align*}
	and similarly, 
	\begin{align*}
		\ppr{(x,y,t)\gets Q, i\gets [2n], a \gets \Ac_1^{f}(i,(x,y),t,\ell)}{ a = 1}&= \ppr{(x,y,t)\gets Q, i\gets [2n]}{i \in [n] \land \RecBit^{\Gc_{(x,y)\flipi,t}}(i)\neq -x_i\cdot y_i} \\
		& \geq e^{-\eps}/16.
	\end{align*}

	\paragraph{Case $2$:  $\ppr{(x,y,t)\gets Q, j \gets [n]}{\RecBit^{\Gc_{x,y\flipj,t}}(j)=-x_j\cdot y_j} <1-e^{-\eps}/8$.} This case is analogous to the previous one, taking $\Ac_2$ instead of $\Ac_1$.
	
	\paragraph{Case $3$: $\ppr{(x,y,t)\gets Q, j \gets [n]}{\RecBit^{\Gc_{x\flipj,y\flipj,t}}(j)=x_j\cdot y_j} \le 3/4$.}
	We show that assuming case $1$ does not hold, $\Ac_3$ fulfills the requirements of the lemma. Indeed
	\begin{align*}
		\ppr{(x,y,t)\gets Q, i\gets [2n], a \gets \Ac_3^{f}(i,(x,y)\flipi,t,\ell)}{ a = 1}&= \ppr{(x,y,t)\gets Q, i\gets [2n]}{i \in [n] \land \RecBit^{\Gc_{x,y\flipi,t}}(i)\neq -x_i\cdot y_i} \\
		& \leq e^{-\eps}/16,
	\end{align*}
	and similarly 
	\begin{align*}
		\ppr{(x,y,t)\gets Q, i\gets [2n], a \gets \Ac_3^{f}(i,(x,y),t,\ell)}{ a = 1}&= \ppr{(x,y,t)\gets Q, i\gets [2n]}{i \in [n] \land \RecBit^{\Gc_{x\flipi,y\flipi,t}}(i)\neq x_i\cdot y_i} \\
		& \geq 1/8.
	\end{align*}
\end{proof}

\subsection{Proving \cref{thm:CondensingSVRes}}\label{sec:proveCondensingThm}
In this section we use \cref{lemma:distinguisher_for_distribution} for proving \cref{thm:CondensingSVRes}. Throughout this section, let $n,\eps, \ell, D$ and $\EveF$ be as in \cref{thm:CondensingSVRes},  let $\Ac_1, \Ac_2$ and $\Ac_3$ be the algorithms guaranteed by \cref{lemma:distinguisher_for_distribution}, let $n_0$ be the constant guaranteed by \cref{lemma:distinguisher_for_distribution}, let $m \eqdef e^{2\eps} \cdot 1000$, let $c\eqdef 2^{30}\cdot n_0$, let $c_1 \eqdef 10$ and let $c_2 \eqdef c^3$.

We prove that the following algorithm, for the right choice of parameters, fulfills the requirements of \cref{thm:CondensingSVRes}.  

\begin{algorithm}[The distinguisher \EveDP]\label{alg:CondensingSV:EveDP}
	
	\item Oracle: $\EveF$.
	
	\item Parameters: $\hl,\hv$, $d\in \set{1,2,3}$.
	
	\item Input: $i \in [2n]$, $(x',y')\in \oo^{2n}$ and $t\in\Ss$.
	
	\item Operation:~
	\begin{enumerate} 	
		
		
		\item Let $(j,b) \eqdef \begin{cases} (i,-1) & i \leq n \\ (i-n,1) & i>n\end{cases}$.
		
		\item Sample uniform $(r_1,...,r_{n^5}) \gets (\mo^n)^{n^5}$, conditioned on $(r_k)_j=b$ for every $k\in[n^5]$.  Let $\cR \eqdef \set{r_k}_{k\in [n^5]}$, and let $q \eqdef \ppr{r \gets \cR}{\size{f(\transFp)-\ip{(x'\cdot y')_{-j}, r_{-j}}} \leq \hl }$.\label{step:estimate_succ}

		\item If $ q \le\hv$, abort.
		
		Else, output $\Ac_d^{f}(i,(x',y'),t,\hl+1)$. 
		
	\end{enumerate}
	
\end{algorithm}

Recall that, given \EveDP aims to distinguish between $(x,y)$ and $(x,y)\flipi$ (in which the \ith bit is flipped). \EveDP starts by trying to figure out whether $f$ is a good estimator of $\ip{\px \cdot \py, \rr}$, for a random $r$.  Since \EveDP does not know the right value of $(x,y)_i$, it invokes $f$ only on inputs that do not contain the missing bit $(x,y)_i$.  If \EveDP finds out that $f$ is a good estimator, it uses $\Ac_d$ for telling whether $(x',y')_i= (x,y)_i$. It easily follows from \cref{lemma:distinguisher_for_distribution} that if for every $s=(x,y,t)$, \EveDP could have computed the success of $f$ on (truly) random $r$, which requires knowing $(x,y)_i$, that is to compute 
\begin{align}
	p^\tuple_\hl \eqdef \ppr{r\gets \mo^n}{\size{\EveF(\transF)-\ip{x\cdot y, r}}\leq \hl},
\end{align}
then, for the right choice of $d$, it would have fulfilled the requirement of \cref{thm:CondensingSVRes}. The crux of our proof is showing that, for most $i$'s,  the difference between $p^\tuple_\hl$ and the computed $q$ is unlikely to affect \EveDP's answer. We do the latter by considering a second distinguisher \EveDPI, an idealized variant of \EveDP that (miraculously) manages to computes a value that is in a sense in-between the value $q$ computed by \EveDP and the above $p^\tuple_\hl$, and used that instead of 
the value of $q$. Note that the value of $q$ can be written as 
\begin{align}
	&q = \ppr{r\gets \cR}{\size{f(\transF)-\ip{x\cdot y, r}-b\cdot x_j\cdot y_j} \leq \hl }
\end{align}
where $j$ and $b$ are the functions of $i$ computed by \EveDP. Hereafter, we use $j(i)$ and $b(i)$ as the values of $j$ and $b$ (respectively) for a given input $i$ (\ie $j(i)\eqdef i$ and $b(i)\eqdef -1$ if $i \leq n$ or $i-n$ and $1$ otherwise). Algorithm \EveDPI manages to computes the value 
\begin{align}
	&q^\tuple_{\hl,i} \eqdef \ppr{r\gets \mo^n}{\size{f(\transF)-\ip{x\cdot y, r}-b(i)\cdot x_{j(i)}\cdot y_{j(i)}} \leq \hl }
\end{align}
I.e.,  $r$ is chosen uniformly, without the restriction that $r_j=b$. (Note that, without knowing $(x,y)_i$, \EveDP cannot calculate $q^\tuple_{\hl,i}$.)  We start, \cref{sec:CondensingSV:Idealized}, by proving that \EveDPI is a good distinguisher, and in \cref{sec:CondensingSV:NonIdealized} extend the proof to the real distinguisher \EveDP. 

\subsubsection{Analyzing the Idealized Distinguisher}\label{sec:CondensingSV:Idealized}
In this section, we prove that \cref{thm:CondensingSVRes} holds with respect to the idealized algorithm \EveDPI. We start by making two observation regarding the probabilities $p_\hl^\tuple$ considered above (\ie the probability that $f$ estimates the inner product on $\tuple$ with error at most $\hl$). In the following,  recall that $\iseg{a,b}\eqdef [a,b]\cap \Z$.  The first claims states that  $p_\hl^\tuple$ is large with high probability over  $\tuple \gets D$. 
\begin{claim}\label{clm:CondensingSV:MainClaim:size_of_g:number}
	$\ppr{\tuple\gets D}{p^\tuple_\ell \geq e^{c_1\eps}\cdot c_2\ell/2\sqrt{n}} \geq e^{c_1\eps}\cdot c_2\ell/2\sqrt{n}$.
\end{claim}

\begin{proof}[Proof of \cref{clm:CondensingSV:MainClaim:size_of_g:number}]
	Recall that,
	\begin{align}
	&\eex{\tuple \gets D}{p^\tuple_{\ell}}= \ppr{\tuple\gets D, r\gets\mo^n}{\size{\EveF(r,x_r,y_{-r},t)-\ip{x\cdot y, r}}\leq \ell}> e^{c_1\eps}c_2\ell/\sqrt{n}.
	\end{align}
	Hence, 
	\remove{
		We show that,
		\begin{align}\label{eq:lemma:IP_to_DP_attacker:eq1}
		\ppr{\tuple\gets D}{p^\tuple_{\ell}\geq e^{4\eps}c_2\ell/2\sqrt{n}} \geq e^{4\eps}c_2\ell/2\sqrt{n}.
		\end{align}
		Indeed, 
	}
	\begin{align*}
	e^{c_1\eps}c_2\ell/\sqrt{n} &< \eex{\tuple \gets D}{p^\tuple_{\ell}}\\
	&\leq \ppr{\tuple\gets D}{p^\tuple_{\ell}\geq e^{c_1\eps}c_2\ell/2\sqrt{n}} \cdot 1 + \ppr{\tuple\gets D}{p^\tuple_{\ell}< e^{c_1\eps}c_2\ell/2\sqrt{n}}\cdot e^{c_1\eps}c_2\ell/2\sqrt{n}\\
	&\leq \ppr{\tuple\gets D}{p^\tuple_{\ell}\geq e^{c_1\eps}c_2\ell/2\sqrt{n}} + e^{c_1\eps}c_2\ell/2\sqrt{n}.
	\end{align*}
\end{proof}

The second claim states that for the right value of $\hl$, the probability that $p_{\hl-1}^\tuple$ is larger than the threshold $\hv$ is very close to the probability that $p_{\hl+1}^\tuple$ is larger than this threshold. 
\begin{claim}\label{claim:good_ell}
For every $\hv  \le e^{4\eps}\cdot c_2\ell/2\sqrt{n}$ there exists $\hl\in \iseg{\ell+1, \ell+m\cdot \ceil{\log n}}$ such that:
	\begin{align*}
	 \ppr{\tuple\gets D}{p_{\hl-1}^\tuple \ge \hv} \le \ppr{\tuple\gets D}{p_{\hl+1}^\tuple \ge \hv}\leq (1+2/m) \cdot \ppr{\tuple\gets D}{p_{\hl-1}^\tuple \ge \hv}.
	\end{align*}
\end{claim}

\begin{proof}[Proof of \cref{claim:good_ell}.]
	Since $\hv \leq e^{4\eps}c_2\ell/2\sqrt{n}$, \cref{clm:CondensingSV:MainClaim:size_of_g:number} yields that
	\begin{align}
	&\ppr{\tuple\gets D}{p^\tuple_{\ell} \ge \hv}\geq e^{4\eps}c_2\ell/2\sqrt{n} \geq 1/\sqrt{n}
	\end{align}
	Assume toward contradiction that for every $\hl \in \iseg{\ell+1,\ell+m\ceil{\log n}}$, it holds that 
	\begin{align}
	&	\ppr{\tuple\gets D}{p^\tuple_{\hl+1} > \hv}> (1+2/m) \cdot \ppr{\tuple\gets D}{p^\tuple_{\hl-1} \ge \hv}
	\end{align}
	But, it would have followed that 
	\begin{align*}
	&1\ge \ppr{\tuple\gets D}{p^\tuple_{\ell+m\ceil{\log n}} > \hv}> (1+2/m)^{(m/2) \cdot \log n} \cdot \ppr{\tuple\gets D}{p^\tuple_{\ell} \ge \hv} \ge n \cdot (1/\sqrt{n}) > 1.
	\end{align*}
\end{proof}

We now prove that \cref{thm:CondensingSVRes} holds with respect to the idealized algorithm \EveDPI, formally stated in the following claim.
\newcommand{\EveDPIf}{\EveDPI^{f}}
\begin{claim}[\EveDPI is a good distinguisher]\label{claim:CondensingSV:Idealized}
It holds that
\begin{align*}
\ppr{\stackrel{(x,y,t) \gets D}{i \gets [2n]}}{\EveDPI^{D,f}(i,(x,y)\flipi, t) = 1} < 0.9\cdot e^{-\eps}\cdot \ppr{\stackrel{(x,y,t) \gets D}{i \gets [2n]}}{ \EveDPI^{D,f}(i,(x,y), t) = 1} - 2/n.
\end{align*}
\end{claim}
\begin{proof}
	Recall that $n_0$ be the constant guaranteed by \cref{lemma:distinguisher_for_distribution}, and that $c\eqdef 2^{30}\cdot n_0$, $c_1 \eqdef 10$ and $c_2 \eqdef c^3$. In addition, let $\hv\in [ e^{4\eps}c\ell/4\sqrt{n}, e^{4\eps}c\ell/2\sqrt{n}]$, and let $\hl$ be the value guaranteed by \cref{claim:good_ell}. We assume \wlg that $n\ge e^{4\eps}c_2 > n_0$ and $\eps \leq 1/20\cdot \log n$ (otherwise, the proof holds trivially as $ e^{ c_1\cdot\eps}\cdot c_2\cdot \ell/\sqrt{\pn}>1$).
	
	We start by upper bounding the probability that \EveDPI (with the above choice of $\hv$ and $\hell$) abort. By definition, it holds that $p^\tuple_{\hl-1} \le p^\tuple_{\hell}$ for every $\tuple\in \Supp(D)$. Moreover, by the triangle inequality, 
	\begin{align}\label{eq:CondensingSV:1}
	q^\tuple_{\hl,i} \in 	[p^\tuple_{\hl-1},p^\tuple_{\hl+1}]
	\end{align} 
	for every $i\in [2n]$. By \cref{clm:CondensingSV:MainClaim:size_of_g:number}, $\ppr{\tuple\gets D}{p^\tuple_{\hl-1} \geq \hv} \ge e^{c_1\eps}c_2\ell/2\sqrt{n} \ge e^{4\eps}c\ell/2\sqrt{n}$, and we conclude that 
	\begin{align} \label{eq:CondensingSV:2}
	\ppr{\tuple = (x,y,t) \gets D,i\gets [2n]}{\EveDPIf(i,(x,y),t)\text{ not abort}} 
	&= \ppr{\tuple \gets D, i \la [2n]}{q^\tuple_{\hl,i} \ge \hv}\\
	&\ge \ppr{\tuple \gets D}{p^\tuple_{\hl-1} \ge \hv}\nonumber\\
	&\ge e^{4\eps}c\ell/2\sqrt{n}.\nonumber 
	\end{align}
	The equality is by definition of \EveDPI. We next want to use \cref{lemma:distinguisher_for_distribution} in order to show that there exists $d \in \set{1,2,3}$ such that $\Ac_d$ is a good distinguisher for the distribution of $(x,y,t)$ sampled from $D$ conditioned on \EveDPI not aborting. In order to use the above lemma, we first need to show that $i$ is close to be uniform when conditioning on no abort. That is, we argue that the value $i$ is close to being independent of the decision taken by  \EveDPI whether to abort or not. Let $\cB = \set{\tuple\colon \: p^\tuple_{\hl+1}\geq\hv \: \land \: p^\tuple_{\hl-1}<\hv}$. By  \cref{claim:good_ell,eq:CondensingSV:2}, 
	\begin{align}\label{eq:CondensingSV:size}
	&\ppr{\tuple\gets D}{\tuple\in \cB} \le 2/m \cdot \ppr{\tuple\gets D}{p^\tuple_{\hl-1} \geq \hv} \leq 2/m\cdot \ppr{\stackrel{\tuple\gets D}{ i\gets [2n]}}{\EveDPIf(i,(x,y),t) \text{ not abort}}
	\end{align}
	Thus,
	\begin{align} \label{eq:CondensingSV:3}
	&\ppr{\stackrel{\tuple= (x,y,t) \gets D}{i \gets [2n]}}{\tuple\in \cB \mid \EveDPIf(i,(x,y),t) \text{ not abort}} \le2/m 
	\end{align}
	
	We next observe that for every $\tuple = (x,y,t)\notin \cB$, it holds that
	\begin{align}\label{eq:CondensingSV:4}
	i \gets [2n]|_{\EveDPIf(i,(x,y)_{-i},t) \text{ not abort}} \equiv 	i \gets [2n]
	\end{align}
	Indeed, let $\tuple = (x,y,t)\notin \cB$ be such that $\EveDPIf(i,(x,y)_{-i},t)$ does not abort for some $i\in[2n]$. By assumption, $q^\tuple_{\hl,i} \ge \hv$. Thus by \cref{eq:CondensingSV:1}, $p^\tuple_{\hl+1} \ge \hv$. Hence, by definition of $\cB$,  $p^\tuple_{\hl-1} \ge \hv$. Therefore by \cref{eq:CondensingSV:1}, $q^\tuple_{\hl,i'} \ge \hv$ for every $i'$, and \EveDPI does not abort on every $i$.

	We are left to show that the distribution of $(x,y,t)$ in
	\begin{align}\label{eq:CondensingSV:dist}
	&((x,y,t),i \gets D\times [n])|_{\EveDPIf(i, (x,y)_{-i},(x,y)_i,t) \text{ not abort} \, \land \, (x,y,t) \notin \cB}
	\end{align}
	fulfills the requirements of \cref{lemma:distinguisher_for_distribution}. Let $Q$ be the distribution of $(x,y,t)$ in \cref{eq:CondensingSV:dist}, and notice that by \cref{eq:CondensingSV:4} we get that the distribution in \cref{eq:CondensingSV:dist} is equal to $Q\times I$, where $I$ is the uniform distribution over $[2n]$. Note that by construction of \EveDPI, the above distribution is independent from the value of $(x',y')_i$. Also by construction, $\EveDPIf(\tuple,j)$ does not abort only if $q^\tuple_{\hl,i}\geq \hv$. Thus, in this case we obtain by \cref{eq:CondensingSV:1} that $p^\tuple_{\hl+1}\geq \hv$.
	The choice of $c$ and the fact that $\hl+1\leq 2m\ell$ yields that if $\EveDPIf(\tuple,j)$ does not abort, then $\tuple$ satisfies the conditions of \cref{lemma:distinguisher_for_distribution} \wrt length parameter $\ell'=\hl +1$. Thus, by \cref{lemma:distinguisher_for_distribution} there exists $d\in \set{1,2,3}$ such that 
	\begin{align}\label{eq:reconstruction:1}
	&\ppr{\stackrel{\tuple\gets Q}{ i\gets [2n]}}{\Ac_d^{f}(i,(x,y),t,\ell') = 1} \geq e^{-\eps}/16
	\end{align}
	and,
	\begin{align}\label{eq:reconstruction:2}
	&\ppr{\stackrel{\tuple\gets Q}{ i\gets [2n]}}{\Ac_d^{f}(i,(x,y)\flipi,t,\ell') = 1} \le 1/2\cdot e^{-\eps}\cdot \ppr{\stackrel{\tuple\gets Q}{ i\gets [2n]}}{\Ac_d^{f}(i,(x,y),t,\ell') = 1}
	\end{align}

	We now use the above observations above to conclude the claim. We first bound $\ppr{\stackrel{\tuple\gets D}{ i\gets [2n]}}{\EveDPIf(i,(x,y),t) = 1 }$. Compute,
	\begin{align}\label{eq:CondensingSV:6}
	\lefteqn{\ppr{\stackrel{\tuple\gets D}{ i\gets [2n]}}{\EveDPIf(i,(x,y),t) = 1 }}\\
	&\ge {\ppr{\stackrel{\tuple\gets D}{ i\gets [2n]}}{\EveDPIf(i,(x,y),t) = 1 \lland \tuple\notin\cB}}\nonumber\\
	& = \ppr{\stackrel{\tuple\gets D}{ i\gets [2n]}}{\EveDPIf(i,(x,y),t) \text{ not abort } \lland \tuple\notin\cB}\cdot\ppr{\stackrel{\tuple\gets Q}{ i\gets [2n]}}{\Ac^{f}(i,(x,y),t,\ell') = 1}\nonumber\\
	& \ge \ppr{\stackrel{\tuple\gets D}{ i\gets [2n]}}{\EveDPIf(i,(x,y),t) \text{ not abort}}(1- 2/m) \cdot \ppr{\stackrel{\tuple\gets Q}{ i\gets [2n]}}{\Ac^{f}(i,(x,y),t,\ell') = 1}.\nonumber
	\end{align}
	The first equality holds by the construction of $Q$, and the last inequality by \cref{eq:CondensingSV:3}. Similarly,
	\begin{align}\label{eq:CondensingSV:8}
	\lefteqn{\ppr{\stackrel{\tuple\gets D}{ i\gets [2n]}}{\EveDPIf(i,(x,y)\flipi,t) = 1 }}\\\nonumber
	&\le \ppr{\stackrel{\tuple\gets D}{ i\gets [2n]}}{\EveDPIf(i,(x,y)\flipi,t) \text{ not abort } \lland \tuple\notin\cB}\cdot\ppr{\stackrel{\tuple\gets Q}{ i\gets [2n]}}{\Ac^{f}(i,(x,y)\flipi,t,\ell') = 1}\\
	&~~ + \ppr{\stackrel{\tuple\gets D}{ i\gets [2n]}}{\EveDPIf(i,(x,y)\flipi,t) \text{ not abort} \lland \tuple \in\cB}\cdot 1
	\nonumber\\\nonumber
	&= \ppr{\stackrel{\tuple\gets D}{ i\gets [2n]}}{\EveDPIf(i,(x,y),t) \text{ not abort} \lland \tuple\notin\cB}\cdot\ppr{\stackrel{\tuple\gets Q}{ i\gets [2n]}}{\Ac^{f}(i,(x,y)\flipi,t,\ell') = 1}\nonumber\\
	&~~ + \ppr{\stackrel{\tuple\gets D}{ i\gets [2n]}}{\EveDPIf(i,(x,y),t) \text{ not abort} \lland \tuple \in\cB}\cdot 1
	\nonumber\\\nonumber
	& \leq \ppr{\stackrel{\tuple\gets D}{ i\gets [2n]}}{\EveDPIf(i,(x,y),t) \text{ not abort}}\cdot(\ppr{\stackrel{\tuple\gets Q}{ i\gets [2n]}}{\Ac^{f}(i,(x,y)\flipi,t,\ell') = 1}+2/m).
	\end{align}
	The equality holds by the observation that the decision to abort is independent of $(x',y')_i$, and the last inequality by \cref{eq:CondensingSV:3}. Combining \cref{eq:reconstruction:2,eq:CondensingSV:6,eq:CondensingSV:8}, yields 
	\begin{align}\label{eq:CondensingSV:9}
	\lefteqn{\ppr{\stackrel{\tuple\gets D}{ i\gets [2n]}}{\EveDPIf(i,(x,y)\flipi,t) = 1 }}\\\nonumber
	& \leq \ppr{\stackrel{\tuple\gets D}{ i\gets [2n]}}{\EveDPIf(i,(x,y),t) \text{ not abort}}\cdot(\ppr{\stackrel{\tuple\gets Q}{ i\gets [2n]}}{\Ac^{f}(i,(x,y)\flipi,t,\ell') = 1}+2/m)\\\nonumber
	& \leq \ppr{\stackrel{\tuple\gets D}{ i\gets [2n]}}{\EveDPIf(i,(x,y),t) \text{ not abort}}\cdot( 1/2\cdot e^{-\eps}\cdot\ppr{\stackrel{\tuple\gets Q}{ i\gets [2n]}}{\Ac^{f}(i,(x,y),t,\ell') = 1}+2/m)\\\nonumber
	& \leq e^{-\eps}/2\cdot (1-2/m)^{-1}\cdot\ppr{\stackrel{\tuple\gets D}{ i\gets [2n]}}{\EveDPIf(i,(x,y),t) = 1 }+2/m\cdot\ppr{\stackrel{\tuple\gets D}{ i\gets [2n]}}{\EveDPIf(i,(x,y),t) \text{ not abort}} \\\nonumber
	& \leq 0.75\cdot e^{-\eps}\cdot\ppr{\stackrel{\tuple\gets D}{ i\gets [2n]}}{\EveDPIf(i,(x,y),t) = 1 }+2/m\cdot\ppr{\stackrel{\tuple\gets D}{ i\gets [2n]}}{\EveDPIf(i,(x,y),t) \text{ not abort}}. \\\nonumber
	\end{align}
	The first equation holds by \cref{eq:CondensingSV:8}, the second by \cref{eq:reconstruction:2}, the third by \cref{eq:CondensingSV:6}, and the last one since $m\geq 1000$.

	We conclude the proof by showing that 
	\begin{align}\label{eq:condensing:leftover}
	2/m\cdot\ppr{\stackrel{\tuple\gets D}{ i\gets [2n]}}{\EveDPIf(i,(x,y),t) \text{ not abort}} < 0.1 \cdot e^{-\epsilon}\cdot \ppr{\stackrel{\tuple\gets D}{ i\gets [n]}}{\EveDPIf(i,(x,y),t) = 1 } - 2/n
	\end{align} 
	which yields the theorem. Indeed,

	\begin{align}\label{eq:CondesingSV:N:1}
	&\ppr{\stackrel{\tuple\gets D}{ i\gets [n]}}{\EveDPIf(i,(x,y),t) = 1 }\\ \nonumber
	&\geq\ppr{\stackrel{\tuple\gets D}{ i\gets [2n]}}{\EveDPIf(i,(x,y),t) \text{ not abort} \lland \tuple\notin \cB }\cdot \ppr{\stackrel{\tuple\gets Q}{ i\gets [2n]}}{\Ac^{f}(i,(x,y),t,\ell') = 1}\\ \nonumber
	&\geq \ppr{\stackrel{\tuple\gets D}{ i\gets [2n]}}{\EveDPIf(i,(x,y),t) \text{ not abort} \lland \tuple\notin \cB }\cdot e^{-\eps}/16\\ \nonumber
	&\geq\ppr{\stackrel{\tuple\gets D}{ i\gets [2n]}}{\EveDPIf(i,(x,y),t) \text{ not abort}}\cdot(1-2/m)\cdot e^{-\eps}/16\\\nonumber
	&> e^{-\eps}/32 \cdot \ppr{\stackrel{\tuple\gets D}{ i\gets [2n]}}{\EveDPIf(i,(x,y),t) \text{ not abort}}.
	\end{align}
The first equation holds by the definition of $Q$, the second by \cref{eq:reconstruction:1}, the third by \cref{eq:CondensingSV:3}, and the last inequality holds since $m>4$. \cref{eq:condensing:leftover} now follows since $n \ge c$, by the choice of $m$ and \cref{eq:CondensingSV:2}.
	
The above concludes the claim proof, apart from the fact that we need to find the right value of $\hl$ and $d$ hardwired into distinguisher \EveDPI. These values can be easily found, however, by trying all options of $(\hl,d)\in [n]\times [3]$. For each such pair, sample a polynomial number of samples from $D$, and by emulating $\EveDPI$ on them, estimating the prediction probability up to $o(1/n)$ error ( with overwhelming probability). Since there are only $O(n)$ possibilities for such value, the above can be done efficiently. 
\end{proof}

\subsubsection{Analyzing the Non-Idealized Distinguisher}\label{sec:CondensingSV:NonIdealized}
In this section, we use the above observations to lower bound the distinguishing advantage of (the non-idealized) algorithm \EveDP, and thus proving \cref{thm:CondensingSVRes}. Recall that \EveDP uses the value of 
\begin{align}
&q = \ppr{r\gets \cR}{\size{f(\transF)-\ip{x\cdot y, r}-b\cdot x_j\cdot y_j} \leq \hl },
\end{align}
rather than that of $q^\tuple_{\hl,i}$, used by its idealized variant \EveDPI considered above. For fixed $\hl$ and $\tuple=(x,y,t)$, let $Q^\tuple_{\hl,i}$ be the value of $q$ in a random execution $\EveDP^f(i,(x,y),t)$ (recall that \EveDP do not use $(x,y)_i$ in order to compute this value). 
In the following we assume $n\ge e^{4\eps}c$, as otherwise the theorem follows trivially. The following two claims will be useful in the proof of \cref{thm:CondensingSVRes}.

The first claim shows that small values added to the value of $q^\tuple_{\hl,i}$ are not likely to change the decision of \EveDP.
In the following, let $\ceps \eqdef e^{4\eps}c$ and $\alpha \eqdef \frac{\ell}{\sqrt{n}\cdot \log^3 n}$.

\def\defNonIdeal{
	There exists $\hv \in [\ceps\ell/4\sqrt{n} ,\ceps\ell/2\sqrt{n}] \cap \set{\ceps\ell/4\sqrt{n}+ k\cdot \alpha\colon k\in \N}$ such that for every $\hl\in \iseg{\ell+1,\ell+m\cdot \ceil{\log n}}$ it holds that
$$\ppr{\tuple\gets D, i\gets [2n]}{q^\tuple_{\hl,i} \in (\hv \pm \alpha)} \le 1/(\sqrt{\ceps}\log n) \cdot \ppr{\stackrel{\tuple\gets D}{ i\gets [2n]}}{q^\tuple_{\hl,i} \ge \hv +\alpha }.$$
}

\begin{claim}\label{clm:CondensingSV:NonIdeal}
	\defNonIdeal
\end{claim}
The proof of \cref{clm:CondensingSV:NonIdeal} is similar to the one of \cref{claim:good_ell}. For every fixing of $\hl$, it cannot holds for too many $\hv$ that $\ppr{\tuple\gets D, i\gets [2n]}{q^\tuple_{\hl,i} \in (\hv \pm \alpha)} > 1/(\sqrt{\ceps}\log n) \cdot \ppr{\stackrel{\tuple\gets D}{ i\gets [2n]}}{q^\tuple_{\hl,i} \ge \hv +\alpha }$, as otherwise it holds that $\ppr{\stackrel{\tuple\gets D}{ i\gets [2n]}}{q^\tuple_{\hl,i} \ge \ceps\ell/4\sqrt{n}}>1$. By our choice of the range of $\hv$, $[\ceps\ell/4\sqrt{n} ,\ceps\ell/2\sqrt{n}] \cap \set{\ceps\ell/4\sqrt{n}+ k\cdot \alpha\colon k\in \N}$, to be large enough, we can show that at least one $\hv$ in this range is good for every $\hl$.

The next claim states that $Q^\tuple_{\hl,i}$ is not too far from $q^\tuple_{\hl,i}$.
\begin{claim}\label{clm:CondensingSV:Approx_clear}
		For every $\tuple\in \Supp(D)$, $\hl \in \iseg{\ell+1,\ell+m\cdot \ceil{\log n}}$ and $\hv \leq \ceps\ell/2\sqrt{n}$, it holds that
	
\begin{align*}
\ppr{ i\gets [2n]}{\paren{Q^\tuple_{\hl,i}\ge \hv \land q^\tuple_{\hl,i}< \hv-\alpha} \lor  \paren{Q^\tuple_{\hl,i}< \hv \land q^\tuple_{\hl,i}\ge \hv+\alpha}}\leq 2/\sqrt{n}.
\end{align*}
\end{claim}
\cref{clm:CondensingSV:Approx_clear} follows by \cref{lemma:boundMultDist}. Recall that 
\begin{align*}
&q^\tuple_{\hl,i} \eqdef \ppr{r\gets \mo^n}{\size{f(\transF)-\ip{x\cdot y, r}-b(i)\cdot x_{j(i)}\cdot y_{j(i)}} \leq \hl }
\end{align*}
and that $Q^\tuple_{\hl,i}$ is an estimation of 
\begin{align*}
\ppr{r\gets \mo^n|_{r_{j(i)=b(i)}}}{\size{f(\transF)-\ip{x\cdot y, r}-b(i)\cdot x_{j(i)}\cdot y_{j(i)}} \leq \hl }.
\end{align*}
Thus, the main difference between $Q^\tuple_{\hl,i}$ to $q^\tuple_{\hl,i}$ is the expectation that taken only over $r$'s for which $r_{j(i)}=b(i)$.  Using \cref{lemma:boundMultDist}  it can be shown that for most values of $i$, $Q^\tuple_{\hl,i}$ and $q^\tuple_{\hl,i}$ are close.
We prove \cref{clm:CondensingSV:NonIdeal,clm:CondensingSV:Approx_clear} below, but first we use them to prove \cref{thm:CondensingSVRes}.

\paragraph{Proving \cref{thm:CondensingSVRes}.}
\begin{proof}[Proof of of \cref{thm:CondensingSVRes}]
	The proof goes by coupling $\EveDP$ with its idealized variant \EveDPI considered above.  Let $\hv$ be the value guaranteed by \cref{clm:CondensingSV:NonIdeal}, $\hl$ be the value guaranteed by \cref{claim:good_ell}, and let $(X,Y,T)\gets D$ and $I \gets [2n]$. Let $O$ and $\tO$ be the output in of random executions of $\EveDP(I,(X,Y),T)$ and $\EveDPI(I,(X,Y),T)$ respectively, using the same random tape for both executions. 
	We start with bounding the probability that $O \neq \tO$. By construction, the event $O \neq \tO$ implies that $Q^\tuple_{I} \ge \hv$ and $q^\tuple_{I} < \hv$, or $Q^\tuple_{I} < \hv$ and $q^\tuple_{I} \ge \hv$, omitting the subscript $\hl$ for clarity of the notation.
	Hence, 	
	\begin{align}\label{eq:CondensingSV:R:1}
	\pr{O \neq \tO} &\le \ppr{\stackrel{\tuple\gets D}{ i\gets [2n]}}{\paren{Q^\tuple_{\hl,i}\ge \hv \land q^\tuple_{\hl,i}< \hv} \lor  \paren{Q^\tuple_{\hl,i}< \hv \land q^\tuple_{\hl,i}\ge \hv}} \\\nonumber
	&\leq  \ppr{\stackrel{\tuple\gets D}{ i\gets [2n]}}{\paren{q^\tuple_{\hl,i} \in (\hv\pm \alpha)}\lor \paren{Q^\tuple_{\hl,i}\ge \hv \land q^\tuple_{\hl,i}< \hv-\alpha} \lor  \paren{Q^\tuple_{\hl,i}< \hv \land q^\tuple_{\hl,i}\ge \hv+\alpha}}\\\nonumber
	&\leq  \ppr{\stackrel{\tuple\gets D}{ i\gets [2n]}}{q^\tuple_{\hl,i} \in (\hv\pm \alpha)}+ \ppr{\stackrel{\tuple\gets D}{ i\gets [2n]}}{\paren{Q^\tuple_{\hl,i}\ge \hv \land q^\tuple_{\hl,i}< \hv-\alpha} \lor  \paren{Q^\tuple_{\hl,i}< \hv \land q^\tuple_{\hl,i}\ge \hv+\alpha}}\\\nonumber
	&\leq  1/(\sqrt{\ceps}\log n) \cdot\ppr{\stackrel{\tuple\gets D}{ i\gets [2n]}}{q^\tuple_{\hl,i} \ge \hv+ \alpha}+2/\sqrt{n}
	\end{align} 	
	Where the third inequality holds by the union bound and the last by \cref{clm:CondensingSV:NonIdeal,clm:CondensingSV:Approx_clear}.
	 By definition of \EveDPI, it holds that,
	\begin{align}\label{eq:CondensingSV:R:2}
	&\ppr{\tuple\gets D, i\gets [2n]}{q^\tuple_{\hl,i} \ge \hv +\alpha } \le \ppr{\stackrel{\tuple\gets D}{ i\gets [2n]}}{\EveDPI^f(i,(x,y),t) \text{ not abort}},
	\end{align} 
	and by \cref{eq:CondensingSV:2}
	\begin{align} 
	2/\sqrt{n} \leq (4/\ceps\ell)\cdot \ppr{\stackrel{\tuple\gets D}{ i\gets [n]}}{\EveDPI^f(i,(x,y),t)\text{ not abort}} 
	\end{align}
	Combining the above with the assumption that $\ell \ge \log n$, we get that,
	\begin{align}\label{eq:CondensingSV:R:3}
	\pr{O \neq \tO} \le 2/(\sqrt{\ceps}\log n)\cdot \ppr{\stackrel{\tuple\gets D}{ i\gets [n]}}{\EveDPI^f(i,(x,y),t)\text{ not abort}}
	\end{align}
	
	We now use the above to bound the distinguishing advantage of \EveDP. By \cref{eq:CondensingSV:R:3} we immediately get that 
	\begin{align}\label{eq:CondesingSV:N:2}
	& \ppr{\stackrel{\tuple\gets D}{ i\gets [n]}}{\EveDP^f(i,(x,y),t) = 1 } \\\nonumber
	&\geq \ppr{\stackrel{\tuple\gets D}{ i\gets [n]}}{\EveDPI^f(i,(x,y),t)=1} 
	- 2/(\sqrt{\ceps}\log n)\cdot \ppr{\stackrel{\tuple\gets D}{ i\gets [n]}}{\EveDPI^f(i,(x,y),t)\text{ not abort}}
	\end{align}
	Recall that by construction, the decision to call $\Ac_d$ is independent of $(x,y)_i$. Thus, using the same line of proof,
	\begin{align}
	&\ppr{\stackrel{\tuple\gets D}{ i\gets [n]}}{\EveDP^f(i,(x,y)\flipi,t) = 1 }\\\nonumber
	& \leq \ppr{\stackrel{\tuple\gets D}{ i\gets [n]}}{\EveDPI^f(i,(x,y)\flipi,t)=1} + 2/(\sqrt{\ceps}\log n)\cdot \ppr{\stackrel{\tuple\gets D}{ i\gets [n]}}{\EveDPI^f(i,(x,y),t)\text{ not abort}}
	\end{align}
	
	Observe that by the choice of $c$ and $m$ it holds that $4/(\sqrt{\ceps}\log n) \leq 1/m$. Combining the above with \cref{eq:CondensingSV:9}, we get,
	\begin{align}
	\lefteqn{\ppr{\stackrel{\tuple\gets D}{ i\gets [n]}}{\EveDP^f(i,(x,y)\flipi,t) = 1 }}\\\nonumber
	& \leq 0.75\cdot e^{-\eps}\cdot\ppr{\stackrel{\tuple\gets D}{ i\gets [n]}}{\EveDP^f(i,(x,y),t) = 1 }+(3/m)\cdot\ppr{\stackrel{\tuple\gets D}{ i\gets [2n]}}{\EveDPI^f(i,(x,y),t) \text{ not abort}} \\\nonumber
	\end{align}	
	We conclude the proof by showing that
	\begin{align}\label{eq:CondensingSV:11}
	3/m\cdot\ppr{\stackrel{\tuple\gets D}{ i\gets [2n]}}{\EveDPI^f(i,(x,y),t) \text{ not abort}} \leq 0.25\cdot e^{-\eps}\cdot  \ppr{\stackrel{\tuple\gets D}{ i\gets [n]}}{\EveDP^f(i,(x,y),t) = 1 } -2/n.
	\end{align}
	Indeed,
	\begin{align}
	&3/m\cdot\ppr{\stackrel{\tuple\gets D}{ i\gets [2n]}}{\EveDPI^f(i,(x,y),t) \text{ not abort}}\\ \nonumber
	& < 0.15 \cdot e^{-\epsilon}\cdot \ppr{\stackrel{\tuple\gets D}{ i\gets [n]}}{\EveDPI^f(i,(x,y),t) = 1 } - 2/n\\\nonumber
	& \leq 0.15 \cdot e^{-\epsilon}\cdot( \ppr{\stackrel{\tuple\gets D}{ i\gets [n]}}{\EveDP^f(i,(x,y),t) = 1 }+1/m\cdot \ppr{\stackrel{\tuple\gets D}{ i\gets [n]}}{\EveDPI^f(i,(x,y),t)\text{ not abort}}) - 2/n,
	\end{align} 
	where the first inequality holds by \cref{eq:condensing:leftover}, and the second by \cref{eq:CondesingSV:N:2} and since $2/(\sqrt{\ceps}\log n) \leq 1/m$. The above implies that 
			\begin{align}
	&(3/m-0.15\cdot e^{-\eps}/m)\cdot\ppr{\stackrel{\tuple\gets D}{ i\gets [2n]}}{\EveDPI^f(i,(x,y),t) \text{ not abort}}\\\nonumber
	&~~~\leq 0.15 \cdot e^{-\epsilon}\cdot \ppr{\stackrel{\tuple\gets D}{ i\gets [n]}}{\EveDP^f(i,(x,y),t) = 1 }- 2/n	
	\end{align}	
which easily yields	\cref{eq:CondensingSV:11} as $(3/m-0.15\cdot e^{-\eps}/m)\geq 2/m$.

 Similar to the ideal case, we need to find the right value of $\hl,\hv$ and $d$ hardwired into distinguisher \EveDP. As in the ideal case, these values can be found by trying all options of triplets $\hl,\hv,d$. For each such triplet, sample a polynomial number of samples from $D$, and by emulating $\EveDP$ on them, estimating the prediction probability up to $o(1/n)$ error (with overwhelming probability). Since by \cref{clm:CondensingSV:NonIdeal} $\hv \in [\ceps\ell/4\sqrt{n} ,\ceps\ell/2\sqrt{n}] \cap \set{\ceps\ell/4\sqrt{n}+ k\cdot \alpha\colon k\in \N}$, this can be done efficiently.
\end{proof}

\newcommand{\hQ}{\widehat{Q}}



\paragraph{Proving \cref{clm:CondensingSV:NonIdeal}.}
\begin{proof}[Proof of \cref{clm:CondensingSV:NonIdeal}.]
	Recall that $m = e^{2\eps}\cdot 1000$,  $c = 2^{30}\cdot n_0$, $\ceps = e^{4\eps}\cdot c$ and $\alpha = \frac{\ell}{\sqrt{n}\cdot \log^3 n}$, and let $d\eqdef \lfloor \ceps\ell/(4\cdot \alpha \cdot \sqrt{n}) \rfloor$. By the choice $d$ it holds that
	\begin{align}
	&d > 2\sqrt{c_\eps}\cdot m \cdot \log^3 n 
	\end{align}
	For every $\hl\in \iseg{\ell+1, \ell+m\log n}$, let $\cB_{\hl}$ be the set of $v\in [d]$ such that
	\begin{align*}
	\ppr{\tuple\gets D, i\gets [2n]}{q^\tuple_{\hl,i}\geq \ceps\ell/4\sqrt{n}+\alpha (v-1)}> (1+1/(\sqrt{\ceps}\log n))\ppr{\tuple\gets D, i\gets [2n]}{q^\tuple_{\hl,i}\geq \ceps\ell/4\sqrt{n}+\alpha (v+1)}.
	\end{align*}
	We need to show that there exists $v\in [d]$ such that $v \notin \cB_{\hl}$ for every $\hl\in \iseg{\ell+1, \ell+m\log n}$. 
	
	We start by showing that for every $\hl\in \iseg{\ell+1, \ell+m\log n}$, the size of $\cB_{\hl}$ is at most $2\sqrt{\ceps}\log^2 n$.
	The claim now follows since $\size{\cup_{\hl\in \iseg{\ell+1, \ell+m\log n}}\cB_{\hl}} \leq 2\sqrt{\ceps}\log^2 n\cdot m\log n< d$. To bound the size of $\cB_{\hl}$ we use a similar argument to the proof of \cref{claim:good_ell}. 
	
	To see the above, fix $\hl\in \iseg{\ell+1, \ell+m\log n}$. We start with showing that 
		\begin{align}\label{eq:good_v:lower_bound}
	\ppr{\tuple\gets D, i\gets [2n]}{q^\tuple_{\hl,i}\geq \ceps\ell/4\sqrt{n}+\alpha\cdot d} \geq 1/\sqrt{n}
	\end{align}
	 To see this, notice that
	\begin{align}\label{eq:good_v:eq1}
	 e^{c_1\eps}\cdot c_2\ell/2\sqrt{n} \geq  \ceps\ell/2\sqrt{n} \geq \ceps\ell/4\sqrt{n} +d\alpha.
	\end{align}
    and recall that by \cref{eq:CondensingSV:1} it holds that 
	$q^\tuple_{\hl,i} \in 	[p^\tuple_{\hl-1},p^\tuple_{\hl+1}]$. Since $\hl \geq \ell+1$, the last implies that
	\begin{align}\label{eq:good_v:eq2}
	q^\tuple_{\hl,i} \ge p^\tuple_{\hl-1} \ge p^\tuple_{\ell}.
	\end{align}
	Lastly, by \cref{clm:CondensingSV:MainClaim:size_of_g:number} it holds that $\ppr{\tuple\gets D}{p^\tuple_\ell \geq e^{c_1\eps}\cdot c_2\ell/2\sqrt{n}} \geq e^{c_1\eps}\cdot c_2\ell/2\sqrt{n}$. Together  with  \cref{eq:good_v:eq1}, we get that,
	\begin{align}\label{eq:good_v:eq3}
	\ppr{\tuple\gets D}{p^\tuple_\ell \geq \ceps\ell/4\sqrt{n}+\alpha\cdot d} \geq \ceps\ell/2\sqrt{n} \geq 1/\sqrt{n}.
	\end{align}
	Combining \cref{eq:good_v:eq2,eq:good_v:eq3} yields \cref{eq:good_v:lower_bound}.
	
	Next, assume toward contradiction that 
	$\size{\cB_{\hl}}>2\sqrt{\ceps}\log^2 n$. By monotonicity, for every $v\in [d]$ it holds that 
	\begin{align}\label{eq:good_v:monoton}
	&\ppr{\tuple\gets D, i\gets [2n]}{q^\tuple_{\hl,i}\geq \ceps\ell/4\sqrt{n}+\alpha (v-1)}\geq \ppr{\tuple\gets D, i\gets [2n]}{q^\tuple_{\hl,i}\geq \ceps\ell/4\sqrt{n}+\alpha (v+1)},
	\end{align}	
	and, by definition, for every $v\in \cB_{\hl}$ it holds that,
	\begin{align}\label{eq:good_v:in_b}
	&\ppr{\tuple\gets D, i\gets [2n]}{q^\tuple_{\hl,i}\geq \ceps\ell/4\sqrt{n}+\alpha (v-1)}\\\nonumber
	&\geq (1+1/(\sqrt{\ceps}\log n))\ppr{\tuple\gets D, i\gets [2n]}{q^\tuple_{\hl,i}\geq \ceps\ell/4\sqrt{n}+\alpha (v+1)}.
	\end{align}	
	Combining \cref{eq:good_v:monoton,eq:good_v:in_b} together with \cref{eq:good_v:lower_bound}, we get,
	\begin{align*}
	&\ppr{\tuple\gets D, i\gets [2n]}{q^\tuple_{\hl,i}\geq \ceps\ell/4\sqrt{n}}\\
	&\geq (1+1/(\sqrt{\ceps}\log n))^{\size{\cB_{\hl}}/2}\ppr{\tuple\gets D, i\gets [2n]}{q^\tuple_{\hl,i}\geq \ceps\ell/4\sqrt{n}+\alpha d}\geq n\cdot 1/\sqrt{n} > 1,
	\end{align*}	
	which cannot holds.
\end{proof}

\paragraph{Proving \cref{clm:CondensingSV:Approx_clear}.}
The proof of  \cref{clm:CondensingSV:Approx_clear} easily follows from the following claim. 
\def\defCondApprox{
	For every $\tuple\in \Supp(D)$ and $\hl \in \iseg{\ell+1,\ell+m\cdot \ceil{\log n}}$, it holds that 
	
	\begin{enumerate}
		\item $\ppr{i\gets [2n]}{\paren{Q^\tuple_{\hl,i} \leq \ceps\ell/\sqrt{n}} \: \land \: \paren{\size{Q^\tuple_{\hl,i} - q^\tuple_{\hl,i}} > \alpha}}< 1/\sqrt{n}$, and 
		
		\item $\ppr{i\gets [2n]}{\paren{Q^\tuple_{\hl,i} \geq \ceps\ell/\sqrt{n}} \:  \land  \: \paren{q^\tuple_{\hl,i} \leq \ceps\ell/2\sqrt{n}}}< 1/\sqrt{n}$.
	\end{enumerate}
}
\begin{claim}\label{clm:CondensingSV:Approx}
	\defCondApprox
\end{claim}
we prove \cref{clm:CondensingSV:Approx} next, but first we use it in order to prove \cref{clm:CondensingSV:Approx_clear}.
\begin{proof}[Proof of \cref{clm:CondensingSV:Approx_clear}.]
Let $\tuple, \hl$ and $\hv$ as in \cref{clm:CondensingSV:Approx_clear}. Observe that since $\hv \leq \ceps\ell/2\sqrt{n}$, it holds that
\begin{align*}
&\ppr{ i\gets [2n]}{\paren{Q^\tuple_{\hl,i}\ge \hv \land q^\tuple_{\hl,i}< \hv-\alpha} \lor  \paren{Q^\tuple_{\hl,i}< \hv \land q^\tuple_{\hl,i}\ge \hv+\alpha}}\\
&\leq \ppr{i\gets [2n]}{\paren{\size{Q^\tuple_{\hl,i} - q^\tuple_{\hl,i}} > \alpha} \land \paren{Q^\tuple_{\hl,i} \leq \hv  \lor  q^\tuple_{\hl,i} \leq\hv}}\\
&\leq \ppr{i\gets [2n]}{\paren{\size{Q^\tuple_{\hl,i} - q^\tuple_{\hl,i}} > \alpha} \land \paren{Q^\tuple_{\hl,i} \leq \ceps\ell/2\sqrt{n}  \lor  q^\tuple_{\hl,i} \leq \ceps\ell/2\sqrt{n}}}\\
&\leq \ppr{i\gets [2n]}{\paren{Q^\tuple_{\hl,i} \leq \ceps\ell/\sqrt{n} \land \size{Q^\tuple_{\hl,i} - q^\tuple_{\hl,i}} > \alpha} \lor \paren{Q^\tuple_{\hl,i} \geq \ceps\ell/\sqrt{n}\land q^\tuple_{\hl,i} \leq \ceps\ell/2\sqrt{n}}}\\
&\leq 2/\sqrt{n}
\end{align*}
where the last inequality follows by \cref{clm:CondensingSV:Approx} and the union bound.
\end{proof}
\paragraph{Proving \cref{clm:CondensingSV:Approx}.}

\begin{proof}[Proof of \cref{clm:CondensingSV:Approx}.]
	Fix $\tuple=(x,y,t)\in \Supp(D)$ and $\hl\in \iseg{\ell+1, \ell+m\log n}$. First, by definition of $\EveDP(i, (x,y)_i,t)$, the expectation of $Q^{\tuple}_{\hl,i}$ (i.e., of value of $q$ in a random execution) is 
	\begin{align*}
	p_{i}\eqdef \ppr{r\gets \mo^n|_{r_{j(i)=b(i)}}}{\size{f(\transF)-\ip{x\cdot y, r}-b(i)\cdot x_{j(i)}\cdot y_{j(i)}} \leq \hl }
	\end{align*}	
	and thus, by applying the Hoffeding bound it holds that,
	\begin{align}\label{eq:good_v_and_ell:hoffeding}
	&\pr{\size{p_i- Q^{\tuple}_{\hl,i}}\geq 1/n^2}\leq 1/n.
	\end{align}

	Next, recall that,	 
	\begin{align}
	&q^\tuple_{\hl,i} \eqdef \ppr{r\gets \mo^n}{\size{f(\transF)-\ip{x\cdot y, r}-b(i)\cdot x_{j(i)}\cdot y_{j(i)}} \leq \hl }
	\end{align}
	We next want to use \cref{lemma:boundMultDist} in order to show that $p_i$ and $q^\tuple_{\hl,i}$ are close for most values of $i$. However, the event that $\size{f(\transF)-\ip{x\cdot y, r}-b(i)\cdot x_{j(i)}\cdot y_{j(i)}} \leq \hl$ is by definition dependent in $i$. To overcame this, we observe that the only dependency of $q^\tuple_{\hl,i}$ on $i$ is in the term $b(i)\cdot x_{j(i)}\cdot y_{j(i)}$ which can be only $-1$ or $1$. For every $\sigma\in \mo$, we define
	\begin{align}
	&q_\sigma \eqdef \ppr{r\gets \mo^n}{\size{f(\transF)-\ip{x\cdot y, r}-\sigma} \leq \hl }
	\end{align}
	It is easy to see that $q^\tuple_{\hl,i} \in \set{q_1,q_{-1}}$ for every $i \in[2n]$. Below  we fix $\sigma \in \mo$ and denote by $I_\sigma$ the distribution $i\gets[2n]|_{b(i)\cdot x_{j(i)}\cdot y_{j(i)}=\sigma}$. In words, choose $i$ uniformly from $[2n]$ under the condition that $b(i)\cdot x_{j(i)} \cdot y_{j(i)}=\sigma$. It is not hard to see that the distribution of $j(i)$ in this process is uniform over $[n]$. 
	
 Recall that we want to show that $Q^\tuple_{\hl,i}$ is not too far from $q^\tuple_{\hl,i}$ (with high probability over $i$). By \cref{eq:good_v_and_ell:hoffeding}, $Q^\tuple_{\hl,i}$ is close to $p_i$ and thus the heart of the proof is showing that $p_i$ is close to $q_\sigma$ for most $i$'s. Indeed, by applying \cref{lemma:boundMultDist} we get that,  if $q_\sigma \geq 1/n$ then,
	\begin{align}\label{eq:good_v_and_ell:mult_dist}
	\ppr{i\gets I_\sigma}{p_{i} \in (1\pm 4n^{-1/4}\cdot \sqrt{\log n})\cdot q_\sigma}\geq 1-1/(2\sqrt{n}).
	\end{align}	
	We next show that 
	$\ppr{i\gets I_\sigma}{\paren{Q^\tuple_{\hl,i} \leq  \ceps\ell/\sqrt{n}} \: \land \: \paren{\size{Q^\tuple_{\hl,i} - q_\sigma } > \frac{\ell}{2\sqrt{n}\cdot \log^3 n}}}< 1/\sqrt{n}$, and,
	
	$\ppr{i\gets I_\sigma}{\paren{Q^\tuple_{\hl,i} \geq  \ceps\ell/\sqrt{n}} \land \paren{q_\sigma \leq \ceps\ell/2\sqrt{n}}}< 1/\sqrt{n}$.
	The claim will follow since the uniform distribution $i \gets [2n]$ is a convex combination of $I_{1}$ and $I_{-1}$.

	The proof is by splitting into cases:
	\begin{enumerate}
		\item $q_\sigma \leq 1/n$\label{case:good_v_and_ell:small_p}
		\item$q_\sigma \in [1/n, 5\ceps\ell/\sqrt{n}]$\label{case:good_v_and_ell:med_p}
		\item$q_\sigma \ge 5\ceps\ell/\sqrt{n}$\label{case:good_v_and_ell:large_p}
	\end{enumerate}
	
	\paragraph{The case $q_\sigma \leq 1/n$.}
	Assume the first case holds. Observe that, since for every $j$ and $b$ $\ppr{r\gets \mon}{r_j=b}=1/2$, it is true that $p_i\in [0, 2\cdot q_{\sigma}]$. Thus, by \cref{eq:good_v_and_ell:hoffeding},
	\begin{align*}
	&\ppr{i\gets I_\sigma}{Q^\tuple_{\hl,i} \leq  \ceps\ell/\sqrt{n} \land \size{Q^\tuple_{\hl,i} - q_\sigma } > \frac{\ell}{2\sqrt{n}\cdot \log^3 n}}\\
	&\leq \ppr{i\gets I_\sigma}{ \size{Q^\tuple_{\hl,i} - q_\sigma } > \frac{\ell}{2\sqrt{n}\cdot \log^3 n}}\\
	&\leq \ppr{i\gets I_\sigma}{ \size{Q^\tuple_{\hl,i} - p_i } > 1/n^2}\\
	&\leq 1/n.
	\end{align*}
	Similarly,
	\begin{align*}
	&\ppr{i\gets I_\sigma}{Q^\tuple_{\hl,i} \geq  \ceps\ell/\sqrt{n} \land q_\sigma \leq c\ell/2\sqrt{n}}\\
	&\leq \ppr{i\gets I_\sigma}{ \size{Q^\tuple_{\hl,i} - q_\sigma } > \frac{\ell}{2\sqrt{n}\cdot \log^3 n}}\\
	&\leq 1/n.
	\end{align*}
	
	\paragraph{The case $q_\sigma \in [1/n, 5\ceps\ell/\sqrt{n}]$.}
	For the second case, notice that by \cref{eq:good_v_and_ell:mult_dist} and the fact that  $c_\eps = c\cdot e^{4\eps} \leq n^{1/4}/\log^4 n$,
	\begin{align}
	\ppr{i\gets I_\sigma}{\size{q_\sigma-p_i}\leq \frac{\ell}{4\sqrt{n}\cdot \log^3 n}}& = \ppr{i\gets I_\sigma}{\size{q_\sigma-p_i}\leq  (1/(20\ceps\cdot\log^3 n ))( 5\ceps\ell/\sqrt{n}) }\\ \nonumber
	&\geq \ppr{i\gets I_\sigma}{\size{q_\sigma-p_i}\leq  (4n^{-1/4}\cdot \sqrt{\log n})( 5\ceps\ell/\sqrt{n}) }\\
	&\geq \ppr{i\gets I_\sigma}{p_{i} \in (1\pm 4n^{-1/4}\cdot \sqrt{\log n})\cdot q_\sigma} \nonumber \\
	&\geq 1-1/(2\sqrt{n}).\nonumber
	\end{align}
	Thus, we get,
	\begin{align*}
	&\ppr{i\gets I_\sigma}{\paren{Q^\tuple_{\hl,i} \leq  \ceps\ell/\sqrt{n}} \: \land \: \paren{\size{Q^\tuple_{\hl,i} - q_\sigma}  > \frac{\ell}{2\sqrt{n}\cdot \log^3 n}}}\\
	&\leq \ppr{i\gets I_\sigma}{ \size{Q^\tuple_{\hl,i} - q_\sigma } > \frac{\ell}{2\sqrt{n}\cdot \log^3 n}}\\
	&\leq \ppr{i\gets I_\sigma}{ \size{q_\sigma - p_i } > \frac{\ell}{4\sqrt{n}\cdot \log^3 n}}+\ppr{i\gets I_\sigma}{ \size{Q^\tuple_{\hl,i} - p_i } > 1/n^2}\\
	&\leq 1/\sqrt{n}.
	\end{align*}
	And similarly,
	\begin{align*}
	&\ppr{i\gets I_\sigma}{\paren{Q^\tuple_{\hl,i} \geq  \ceps\ell/\sqrt{n}} \: \land  \: \paren{q_\sigma \leq \ceps\ell/2\sqrt{n}}}\\
	&\leq \ppr{i\gets I_\sigma}{ \size{Q^\tuple_{\hl,i} - q_\sigma } > \frac{\ell}{2\sqrt{n}\cdot \log^3 n}}\\
	&\leq 1/\sqrt{n}.
	\end{align*}
	
	\paragraph{The case $q_\sigma \ge 5\ceps\ell/\sqrt{n}$.}
	Lastly, assume the third case. Notice that,
	\begin{align*}
	&\ppr{i\gets I_\sigma}{Q^\tuple_{\hl,i} \leq  \ceps\ell/\sqrt{n} \land \size{Q^\tuple_{\hl,i} - q_\sigma } > \frac{\ell}{2\sqrt{n}\cdot \log^3 n}}\\
	&\leq \ppr{i\gets I_\sigma}{Q^\tuple_{\hl,i} \leq  \ceps\ell/\sqrt{n}}\\
	&\leq \ppr{i\gets I_\sigma}{p_i \leq  2\ceps\ell/\sqrt{n}}+\ppr{i\gets I_\sigma}{ \size{Q^\tuple_{\hl,i} - p_i } > 1/n^2}\\
	&\leq \ppr{i\gets I_\sigma}{p_i \leq  1/2 \cdot q_\sigma}+1/n\\
	&\leq 1/\sqrt{n}
	\end{align*}	
	Where the last inequality holds by \cref{eq:good_v_and_ell:mult_dist}.	
\end{proof}

%% file: HadamardRec.tex
\newcommand{\AlgEstimateBitX}{\MathAlgX{{\AlgEstimateBit}_X}}
\newcommand{\AlgEstimateBitY}{\MathAlgX{{\AlgEstimateBit}_Y}}
\newcommand{\hAlgEstimateBitX}{\MathAlgX{\widehat{\AlgEstimateBit}_X}}
\newcommand{\hAlgEstimateBitY}{\MathAlgX{\widehat{\AlgEstimateBit}_Y}}
\newcommand{\AlgDistinguish}{\MathAlgX{Dist}}
\newcommand{\AlgEstZ}{\MathAlgX{EstZ}}

\def\EInd{i}

\section{Reconstruction from  Non-Boolean Hadamard Code}\label{sec:reconstruction}
In this section, we prove \cref{thm:reconstruction}, restated below.
\begin{definition}[Inner-product  estimator, restatement of \cref{def:estimator}]
	\defEstimator
\end{definition}
\begin{theorem}[Restatement of \cref{thm:reconstruction}]\label{thm:reconstructionRes}
	\theoremEliadi
\end{theorem}
That is,   \cref{thm:reconstructionRes} guarantees the existence of an efficient  predictor \Pc, whose output given    $(i,\zz_{-i})$ and a single sample from an $(\lambda,\ell)$-estimator of $\ip{\zz,\cdot}$, is positively correlated with $z_i$ (for most   $i$'s).   We remark that \cref{thm:reconstructionRes} is tight up to a constant factor:  for  small enough constant $\lambda$  (\eg  $\lambda = 1/4$)  and not too large $\ell$ (\ie $\ell \ll \sqrt{n}$ ), the function $f(\rr) \eqdef 0$ is an $(\lambda,\ell)$-estimator of $\ip{\zz,\cdot}$ (Holds by standard properties of the binomial distribution). Clearly,  it is impossible to predict any information about $z_i$ from $\zz_{-i}$ and such $f$. 

Note that \cref{thm:reconstruction:intro} from the introduction is an immediate corollary of \cref{thm:reconstructionRes}.

\begin{theorem}[Restatement of \cref{thm:reconstruction:intro}]
	There exists a \pptm $\Rec$ that for every database $z \in \oo^n$, given an $\paren{\lambda=300, \ell}$-estimator $f$ of $\ip{z,\cdot}$, for at least $0.9$ fraction of the $i \in [n]$ it holds that $\Rec^{f}(i,z_{-i},\ell) = z_i$ with probability $0.99$. $\Rec$ uses $O(n^3)$ queries to $f$.
\end{theorem}
\begin{proof}
	\cref{thm:reconstructionRes} implies that for at least $0.9$ of the $i \in [n]$ it holds that $\mu_i \eqdef \eex{\rr \la \oo^n}{\Pc(\EInd ,\zz_{-\EInd },\rr,f(\rr),\ell)}$ has $\size{\mu_i} \geq \frac{30}{n^{1.5}}$, and its correlated with $z_i$ (i.e., $z_i \cdot \mu_i > 0$). Therefore, we define algorithm $\Rec$, given an oracle access to $f$ and inputs $i,z_{-i},\ell$, to estimate  $\mu_i$ using $O(n^3)$ uniformly random samples $r \la \oo^n$, and output its sign. Since $\Pc$ outputs a value in $\set{-1,0,1}$, by Hoeffding's inequality it holds that additive error of the estimation is smaller than $\size{\mu_i}$ with probability $0.99$, which yields that the sign is correct. 
\end{proof}

The proof of \cref{thm:reconstructionRes} is an easy corollary of the following lemma. 
\begin{definition}\label{def:est-k}
	For $k \in \Z$,  let 
	\begin{align*}
		\est_k(i,\zz_{-i},\rr, a) \eqdef \begin{cases} (a - \ip{\zz_{-i},\rr_{-i}} - k)\cdot r_i  & a - \ip{\zz_{-i},\rr_{-i}} \in \set{k-1,k+1} \\ 0. & \text{otherwise}\end{cases},
	\end{align*}
For $f\colon \oo^n \mapsto \Z$, let 
	\begin{align*}
		\est^{f}_k(i,\zz_{-i},\rr) \eqdef\est_k(i,\zz_{-i},\rr, f(\rr)).
	\end{align*}
\end{definition}  

\begin{lemma}\label{lem:tight-learner}
	 There exists an efficiently samplable distribution ensemble $\cK =\set{\cK_{n,\ell}}_{ n,\ell \in \N}$ such that the following holds for every  $\lambda \geq 64$,  $\ell \in \N$ and sufficiently large  $n \in \N$:  let $f$ be a   $(\lambda,\ell)$-estimator  of $\ip{\zz,\cdot}$ and let $\set{\est^{f}_k}_{k \in \Z}$ be according to \cref{def:est-k}. Then for every $\zz \in \oo^n$
	\begin{align*}
		\ppr{i \gets [n]}{
		z_i \cdot \eex{k \la \cK_{n,\ell}, \rr \gets \oo^n}{\est^{f}_k(i,\zz_{-i},\rr)} \geq \frac{\lambda }{8 n^{1.5}}} \ge 1 - {4096}/{\lambda^2}.
	\end{align*}

\end{lemma}

We prove \cref{lem:tight-learner} below, but first use it for proving \cref{thm:reconstructionRes}.

\paragraph{Proving \cref{thm:reconstructionRes}.}
We prove \cref{thm:reconstruction} by applying \cref{lem:tight-learner} \wrt  the following \pptm $\Pc$. Let $\est_k$ be according to \cref{def:est-k}, and let $\set{\cK_{n,\ell}}$ is the distribution ensemble guaranteed by  \cref{lem:tight-learner}.

\begin{algorithm}[$\Pc$]\label{alg:P}
	\item {Inputs:} $i \in [n]$, $\zz_{-i}  \in \oo^{n-1}$, $\rr \in \oo^n$,  $a \in \Z$ and $\ell \in [\floor{\sqrt{n}}]$.
	\item {Operation:}
	\begin{enumerate}	
		
		\item Sample $k \la \cK_{n,\ell}$.
		
		\item Output $\est_k(i, \zz_{-i} ,\rr, a)$.
		
	\end{enumerate}	
\end{algorithm}
\begin{proof}[Proof of  \cref{thm:reconstructionRes}]
Immediate by applying  \cref{lem:tight-learner} on \cref{alg:P} (note that  for every $r$, $\ex{\Pc(i,\zz_{-i},\rr,f(\rr),\ell)} = \eex{k \la \cK_{n,\ell}}{\est_k^f(i,z_{-i},r)}$).
\end{proof}

\subsection{Proving \cref{lem:tight-learner}}
The proof of the lemma is an easy corollary of the following claims, which we prove in  \cref{sec:hadamard:proving-claims}. Let $\lambda, n,\ell$ be  as in the lemma statement, let $f\colon \oo^n \mapsto \Z$, and let $\set{\est_k^f}_{k \in \bbZ}$ be according to \cref{def:est-k}, and fix  $\zz \in \oo^n$.
We use  the following notation: for $k \in \Z$, let $\cG_k \eqdef \set{r \in \oo^n \colon f(\rr) = \ip{\zz,\rr} + k}$, \ie those $r$ on which $f(r)$ is off by $k$. Let   $p_k \eqdef \ppr{\rr \gets \oo^n}{\rr \in \cG_k}$ and let $q_{k} \eqdef  p_k + p_{-k}$. For $i \in [n]$, let $\cB_{k}^{i} \eqdef \set{r \in \oo^n \colon f(\rr) = \ip{\zz_{-i},\rr_{-i}} - z_i r_i + k}$, \ie those $r$'s that are not in  $\cG_k$, but to refute that one needs to know $z_i$. By definition, for every  $i$ and $k$, it holds that
\begin{align}\label{eq:reconstruction:cases}
\est^{f}_k(i,\zz_{-i},\rr) = \begin{cases} z_i & \rr \in \cG_k \\ -z_i & \rr \in \cB^{k}_{i} \\ 0 & \text{o.w.}\end{cases}
\end{align}
Fix an (arbitrary) set of indices $\cI \subseteq [n]$, 
and let $\mu_k \eqdef \frac1{\size{\cI}} \cdot p_k \cdot \eex{\rr \la \cG_k}{\ip{\zz_{\cI},\rr_{\cI}}}= p_k \cdot \eex{i \la \cI, r \la \cG_k}{z_i \cdot r_i}$, \ie the (normalized)  correlation between $z_\I$ and $\cG_k$.   The first claim expresses, for every  fixed $k$,  the  accuracy of  $\est_k^{f}$ over $i\gets \cI$, in terms of  $p_j$'s and $\mu_j$'s.
\begin{claim}\label{lem:tight-learner:1}
	For every $k \in \Z$, it holds that $$\eex{i \la \cI, \rr \gets \oo^n}{z_i \cdot \est_k^{f}(i,\zz_{-i},\rr)}
	= \frac12 \cdot (\underbrace{2p_{k} - p_{k+2} - p_{k-2} + \mu_{k+2} - \mu_{k-2}}_{\alpha_k}).$$
\end{claim}
The following claims gradually proves  the existence of  a distribution $\cK_{n,\ell}$ over the values of  $k$,  such that    the expected value of $\alpha_k$ is large.  Towards this end, the next claim expresses the expected value of $\alpha_k$ for $k \gets \iseg{-(m+1),m+1}$, as a function of the $q_j$'s and $\mu_j$'s.
\begin{claim}\label{lem:tight-learner:2}
	For every  $ m \in \N \cup \set{0}$, it holds that  
	\begin{align*}
		&\eex{k \gets \iseg{-(m+1),m+1}}{\alpha_k}=  \underbrace{\frac1{(2m+3)} \paren{q_{m} + q_{m+1} - q_{m+2} - q_{m+3} +  \sum_{j=0}^{3} (\mu_{m+j} - \mu_{-(m+j)})}}_{\beta_m}.
	\end{align*}
\end{claim}

The next claim lower-bounds the expected value of $\beta_m$ (defined in \cref{lem:tight-learner:2}) \wrt  the following distribution.

\begin{definition}[The distribution $\cM_{s,t}$]\label{def:dist-cM}
	For $s,t \in \N$ with $s < t$,  let $\cM_{s,t}$ be the distribution over $\iseg{s,t-1}$ defined by  $\cM_{s,t}(m) \eqdef  \frac{2m+3}{(t-s)(t+s+2)}$. (I.e., $\cM_{s,t} \propto 2m+3$.)
\end{definition} 
\begin{claim}\label{lem:tight-learner:3}
	Assume  the size of $\cI$ is larger than a universal constant, 
	then for  every  $s,t \in \Z$ with $0 \leq s \leq t-3$, and $t\leq \sqrt{n}$, it holds that
	\begin{align*}
		&\eex{m \gets \cM_{s,t}}{\beta_m}
		\ge \underbrace{\frac1{(t-s)(t+s+2)} \cdot \paren{q_s - (q_{t} + 2 q_{t+1} + q_{t+2}) - \frac{32}{\sqrt{\size{\cI}}}}}_{\gamma_{s,t}}.
	\end{align*}
\end{claim}

Finally, assume $f$ is a good estimator,  the next claim lower-bounds the expected value of $\gamma_{s,t}$ (defined in \cref{lem:tight-learner:3}) \wrt  the following distribution.
\begin{definition}[The distribution $\cP_{\cS,\cT}$]\label{def:dist-cP}
	For finite $\cS,\cT \subseteq \N \cup \set{0}$ with $\max(\cS) < \min(\cT)$, let $\cP_{\cS,\cT}$ be the distribution over $\cS \times \cT$ defined by $\cP_{\cS,\cT}(s,t) \eqdef \frac{(t-s)(t+s+2)}{\sum_{(s',t') \in \cS \times \cT} (t'-s')(t'+s'+2)}$ (i.e., $\cP_{\cS,\cT}(s,t)\propto (t-s)(t+s+2)$).
\end{definition}
\begin{claim}\label{lem:tight-learner:4}
	Assume $f$ is an $(\lambda,\ell)$-estimator of $\ip{\zz,\cdot}$  and that $\size{\cI} \geq {4096 n}/{\lambda^2}$.	Let $\cS \eqdef  \iseg{0,\ell-1}$, $\cT \eqdef \iseg{\ell+2,\sqrt{n}}$, and let $\cP_{\cS,\cT}$ be according to \cref{def:dist-cP}. Then 
	$$\eex{(s,t) \la \cP_{\cS,\cT}}{\gamma_{s,t}}\geq {\lambda}/{4 n^{1.5}}.$$
\end{claim}

Given the above claims, we are now ready to prove  \cref{lem:tight-learner}
\begin{proof}[Proof of \cref{lem:tight-learner}.]
Let $\lambda, n,\ell,z,f$ and $\sset{\est_k^f}_{k \in \bbZ}$ be as in the lemma statement. Since  $f$ is an $(\lambda,\ell)$-estimator of $\ip{\zz,\cdot}$, it holds that  $\ell \leq \sqrt{n}/\lambda< \sqrt{n} - 2$.  Let $\cS \eqdef  \iseg{0,\ell-1}$, let $\cT \eqdef \iseg{\ell+2, \sqrt{n}}$,  and let  $\cK_{n,\ell}$ be the output distribution of the following random process:

\begin{algorithm}[The distribution $\cK_{n,\ell}$]~
\begin{enumerate}
	\item Sample  $(s,t) \la \cP_{\cS,\cT}$.  (See \cref{def:dist-cP} for the definition of  $ \cP_{\cS,\cT}$). 
	
	\item  Sample $m \la \cM_{s,t}$. (See \cref{def:dist-cM} for the definition of  $\cM_{s,t}$). 
	
	\item Output $k \la \iseg{-(m+1),m+1}$.
\end{enumerate}
\end{algorithm}
It is clear that  $\cK_{n,\ell}$ is samplable in polynomial time (by an algorithm getting $(1^n,1^\ell)$ as input). By \cref{lem:tight-learner:1,lem:tight-learner:2,lem:tight-learner:3,lem:tight-learner:4}, for any set  $\cI \subseteq [n]$ of size  $4096 n/\lambda^2$, it holds that 
	\begin{align}\label{eq:tight-learner}
	\eex{k \la \cK_{n,\ell},\text{ }i \la \cI,\text{ }\rr \gets \oo^n}{z_i \cdot \est_k^f(i,\zz_{-i},\rr)}&= \frac12 \cdot \eex{(s,t)\gets \cP_{\cS,\cT}, m \gets \cM_{s,t}, k\gets [[-(m+1),m+1]]}{\alpha_k}\\
		&= \frac12 \cdot \eex{(s,t)\gets \cP_{\cS,\cT}, m \gets \cM_{s,t}}{\beta_m}\nonumber\\
		&\ge \frac12 \cdot \eex{(s,t)\gets \cP_{\cS,\cT}}{\gamma_{s,t}}\nonumber\\
		&\ge  \frac{\lambda}{8 n^{1.5}}.\nonumber
	\end{align}
	The first equality  holds by \cref{lem:tight-learner:1}, the second equality by \cref{lem:tight-learner:2},  the first inequality  by \cref{lem:tight-learner:3}, and the last inequality by \cref{lem:tight-learner:4}. To conclude the proof, consider  the   set of ``bad'' indices:
	\begin{align*}
		\cI \eqdef \set{i \in [n] \colon \eex{k \la \cK_{n,\ell},\, \rr \gets \oo^n}{z_i \cdot \est^{f}_k(i,\zz_{-i},\rr)} < \frac{\lambda }{8 n^{1.5}}}
	\end{align*}	
Assume towards a contradiction that the lemma does not hold, and therefore $\size{\cI} \geq 4096 n/\lambda^2$.  \cref{eq:tight-learner} yields that 
	$\eex{k \la \cK_{n,\ell},\text{ }i \la \cI,\text{ }\rr \gets \oo^n}{z_i \cdot \est_k^f(i,\zz_{-i},\rr)} \geq \frac{\lambda }{8 n^{1.5}}$, in a contradiction to the definition of $\cI$.	
\end{proof}

\subsubsection{Proving   \cref{lem:tight-learner:1,lem:tight-learner:2,lem:tight-learner:3,lem:tight-learner:4}}\label{sec:hadamard:proving-claims}

\paragraph{Proving \cref{lem:tight-learner:1}.}
\begin{proof}[Proof of \cref{lem:tight-learner:1}]
	 \cref{eq:reconstruction:cases} yields that for every $i\in [n]$ and $k\in \Z$,
	\begin{align}\label{eq:exp-for-all-i-k}
	 &\eex{\rr \gets \oo^n}{z_i \cdot \est_k^{f}(i,\zz_{-i},\rr)}\\
		&= \ppr{\rr \gets \oo^n}{\rr \in \cG_k} \cdot \eex{\rr \la \cG_k}{z_i \cdot z_i } + \ppr{\rr \gets \oo^n}{\rr \in \cB_{k}^i} \cdot \eex{\rr \la \cB_{k}^i}{z_i \cdot (-z_i)}\nonumber\\
		&= p_k - \ppr{\rr \gets \oo^n}{\rr \in \cB_{k}^i}.\nonumber
	\end{align}
	Note that $\rr \in \cB^{k}_i$ if and only if: (1) $\rr \in \cG_{k+2}$ and $z_i \cdot r_i = -1$, or (2) $\rr \in \cG_{k-2}$ and $z_i \cdot r_i = 1$.
	Therefore, for every $k\in \Z$
	\begin{align}\label{eq:prob-of-B_k_i}
		\lefteqn{\eex{i \la \cI}{\ppr{\rr \gets \oo^n}{\rr \in \cB_{k}^i}}}\\
		&= p_{k+2} \cdot \ppr{i \la \cI, \rr \la \cG_{k+2}}{z_i \cdot r_i = -1} + p_{k-2}\cdot  \ppr{i \la \cI, \rr \la G^{k-2}}{z_i \cdot r_i = 1}\nonumber\\
		&= p_{k+2} \cdot \frac1{\size{\cI}} \cdot \eex{\rr \la \cG_{k+2}}{\size{\set{i \in \cI \colon z_i \cdot r_i =-1}}} + p_{k-2} \cdot \frac1{\size{\cI}} \cdot \eex{\rr \la \cG_{k-2}}{\size{\set{i \in \cI \colon z_i \cdot r_i = 1}}}\nonumber\\
		&= p_{k+2} \cdot \frac{1 - \frac1{\size{\cI}}\cdot \eex{\rr \la \cG_{k+2}}{\ip{\zz_{\cI},\rr_{\cI}}}}{2} + p_{k-2} \cdot  \frac{1 + \frac1{\size{\cI}}\cdot \eex{\rr \la \cG_{k-2}}{\ip{\zz_{\cI},\rr_{\cI}}}}{2}\nonumber\\
		&= \frac12\cdot (p_{k+2} + p_{k-2}  - \mu_{k+2} + \mu_{k-2}).\nonumber
	\end{align}
		The penultimate equality holds since  since $\size{\set{i \in \cI \colon z_i \cdot r_i = 1}} = \frac{\size{\cI} + \ip{\zz_{\cI},\rr_{\cI}}}{2}$. The proof now follows  by \cref{eq:exp-for-all-i-k,eq:prob-of-B_k_i}.
\end{proof}

\paragraph{Proving \cref{lem:tight-learner:2}.}
\begin{proof}[Proof of  \cref{lem:tight-learner:2}]
Note that
	\begin{align*}
		\eex{k \gets \iseg{-(m+1),m+1}}{\alpha_k}
		&= \eex{k \gets \iseg{-(m+1),m+1}}{2p_{k} - p_{k+2} - p_{k-2} + \mu_{k+2} - \mu_{k-2}}\\
		&= \frac1{2m+3} \sum_{k=-(m+1)}^{m+1} {(2p_{k} - p_{k+2} - p_{k-2} + \mu_{k+2} - \mu_{k-2})}.\nonumber
	\end{align*}
	The proof of the claim now follows since\Enote{A note for myself: the following (indeed) holds also for $m=0$.}
	\begin{align*}
		\sum_{k=-(m+1)}^{m+1} {(2p_{k} - p_{k+2} - p_{k-2})}
		&= p_{-(m+1)} + p_{-m} + p_{m} + p_{m+1} - p_{-(m+3)} - p_{-(m+2)} - p_{m+2} - p_{m+3}\\
		&= q_{m} + q_{m+1} - q_{m+2} - q_{m+3},
	\end{align*}
	and since
	\begin{align*}
			\sum_{k=-(m+1)}^{m+1} \paren{\mu_{k+2} - \mu_{k-2}} = \sum_{j=0}^{3} (\mu_{m+j} - \mu_{-(m+j)}).
	\end{align*} 
\end{proof}

\paragraph{Proving \cref{lem:tight-learner:3}.}
\begin{proof}[Proof of \cref{lem:tight-learner:3}]	
	 Let  $\widetilde{\mu}_k \eqdef \frac1{\size{\cI}} \cdot p_k\cdot \eex{\rr \la \cG_k}{\size{\ip{\zz_{\cI},\rr_{\cI}}}}$. Compute
	\begin{align}\label{eq:exp-for-all-s-t}
		&\eex{m \gets \cM_{s,t}}{\beta_m}\\
		&= \eex{m \gets \cM_{s,t}}{\frac1{2m+3} \paren{q_{m} + q_{m+1} - q_{m+2} - q_{m+3} +  \sum_{j=0}^{3} (\mu_{m+j} - \mu_{-(m+j)})}}\nonumber\\
		&= \frac1{(t-s)(t+s+2)}\cdot \sum_{m=s}^{t-1} \paren{q_{m} + q_{m+1} - q_{m+2} - q_{m+3} +  \sum_{j=0}^{3} (\mu_{m+j} - \mu_{-(m+j)})}\nonumber\\
		&= \frac1{(t-s)(t+s+2)}\cdot \paren{q_s + 2 q_{s+1} + q_{s+2} - q_{t} - 2 q_{t+1} - q_{t+2} + \sum_{m=s}^{t-1} \sum_{j=0}^{3}(\mu_{m+j} - \mu_{-(m+j)})}\nonumber\\
		&\geq \frac1{(t-s)(t+s+2)} \cdot \paren{q_s - (q_{t} + 2 q_{t+1} + q_{t+2}) - 8\cdot \sum_{m=-(s+3)}^{t+2} \widetilde{\mu}_m}.\nonumber
	\end{align}
The inequality holds since  $q_j \geq 0$, for every $j$, and since $\size{\mu_j} \leq \widetilde{\mu}_j$. The third equality holds since
	\begin{align*}
		\sum_{m=s}^{t-1} \paren{q_{m} + q_{m+1} - q_{m+2} - q_{m+3}}
		&= \sum_{m=s}^{t-1} (q_{m} - q_{m+2}) + \sum_{m=s}^{t-1} (q_{m+1} - q_{m+3})\\
		&= (q_s + q_{s+1} -q_{t}-q_{t+1}) + (q_{s+1} + q_{s+2} -q_{t+1}-q_{t+2})\\
		&= q_s + 2 q_{s+1} + q_{s+2} - q_{t} - 2 q_{t+1} - q_{t+2}.
	\end{align*}
Let $\cG' \eqdef  \bigcup_{k=-(s+3)}^{t+2} \cG_k$, and observe that
	\begin{align}\label{eq:exp-for-all-s-t-sec}
		 \sum_{m=-(s+3)}^{t+2} \widetilde{\mu}_m
		 = \frac1{\size{\cI}}\cdot \ppr{\rr \gets \oo^n}{\rr \in \cG'} \cdot \eex{\rr \la \cG'}{\size{\ip{\zz_{\cI},\rr_{\cI}}}}
		 \leq \frac{4}{\sqrt{\size{\cI}}}
	\end{align}
	The inequality holds by applying \cref{proposition:exp-of-abs} over the event $\cG'$, noting that  $\ip{\zz_{\cI},\rr_{\cI}}$, for $\rr \gets \oo^n$, is a sum of $\size{\cI}$ uniform and independent random variables over $\oo$.\footnote{The event $\cG'$ is defined over a larger probability space that include also $\rr_{-\cI}$. Yet, for every fixing of $\rr_{-\cI}$, we can apply \cref{proposition:exp-of-abs} over the event $\set{\rr_{\cI} \colon \rr \in \cG'}$.}
	\end{proof}

\paragraph{Proving \cref{lem:tight-learner:4}.}
\begin{proof}[Proof of \cref{lem:tight-learner:4}]
Recall that $f$ is a $(\lambda,\ell)$-estimator for $\lambda \geq 64$, and that $\cS = \iseg{0,\ell-1}$ and $\cT = \iseg{\ell+2, \sqrt{n}}$. Since $q_k = \ppr{\ptr \gets \oo^n}{f(\ptr) - \ip{\zz,\rr} = k} + \ppr{\ptr \gets \oo^n}{f(\ptr) - \ip{\zz,\rr} = -k}$, it holds that
\begin{align}\label{eq:tight-learner:4:1}
	\eex{s \la \cS}{q_s}
	= \frac1{\size{\cS}} \cdot \sum_{s=0}^{\ell-1} q_s
	\geq \frac1{\ell} \cdot \ppr{\ptr \gets \oo^n}{\size{f(\ptr) - \ip{\zz,\rr}} < \ell}
	\geq {\lambda}/{\sqrt{n}}.
\end{align}
On the other hand, since $q_k \leq 2$ for every $k$, it holds that
\begin{align}\label{eq:tight-learner:4:2}
	\eex{t \la \cT}{q_t + 2 q_{t+1} + q_{t+2}}
	&= \frac1{\size{\cT}} \cdot \sum_{t\in \cT} (q_t + 2 q_{t+1} + q_{t+2})
	\leq \frac1{\sqrt{n} - \ell - 1} \cdot 8
	\leq \frac{16}{\sqrt{n}}
	\leq \frac{\lambda}{4 \sqrt{n}}
\end{align}
The first inequality holds since, by assumption, $\ell \leq \sqrt{n}/\lambda < \sqrt{n}/2-1$.	Compute
	\begin{align*}
		\eex{(s,t) \la \cP_{\cS,\cT}}{\gamma_{s,t}}&= \eex{(s,t) \la \cP_{\cS,\cT}}{\frac1{(t-s)(t+s+2)} \cdot \paren{q_s - (q_{t} + 2 q_{t+1} + q_{t+2}) - \frac{32}{\sqrt{\size{\cI}}}}}\\
		&=  \frac1{ \sum_{(s,t) \in \cS \times \cT} (t-s)(t+s+2)} \cdot \sum_{(s,t) \in \cS \times \cT}  \paren{q_s - (q_{t} + 2 q_{t+1} + q_{t+2}) - \frac{32}{\sqrt{\size{\cI}}}}\nonumber\\
		&= \frac{\size{\cS} \size{\cT}}{ \sum_{(s,t) \in \cS \times \cT} (t-s)(t+s+2)} \cdot \paren{\eex{s \gets \cS}{q_s} - \eex{t \gets \cT}{q_t + 2 q_{t+1} + q_{t+2}} - \frac{32}{\sqrt{\size{\cI}}}}\nonumber\\
		&\geq \frac{1}{\eex{(s,t) \gets \cS \times \cT}{(t-s)(t+s+2)}} \cdot  \paren{\frac{\lambda}{\sqrt{n}} - \frac{\lambda}{4 \sqrt{n}} - \frac{\lambda}{2 \sqrt{n}}}.\nonumber
	\end{align*}
	The last inequality holds follows by \cref{eq:tight-learner:4:1,eq:tight-learner:4:2}, since, by assumption,  $\size{\cI} \geq {4096 n}/{\lambda^2}$. This  concludes the proof since $	\eex{(s,t) \gets \cS \times \cT}{(t-s)(t+s+2)}  \leq \eex{t \gets \cT}{t(t+2)} \leq n$.
\end{proof}

%% file: Key_agreement_ampl.tex
\newcommand{\BX}{X}
\newcommand{\BY}{Y}
\newcommand{\Egl}{\Dec}
\newcommand{\Eka}{\Ec_{KA}}
\newcommand{\Dts}{{\Dc_{\hT^*}}}
\newcommand{\Dt}{{\Dc_{\hT}}}

\section{Key-Agreement Amplification over Large Alphabet}\label{sec:KAApmlification}
In this section we prove our amplification result for key-agreement  protocol over large alphabet that we used in \cref{sec:KAProtocol}, restated below.

\begin{theorem}[Key-agreement amplification over large alphabet, \cref{thm:key-agreement-amp} restated]\label{thm:key-agreement-ampRes}
	\thmKAmplificationN

\end{theorem}
That is, given an $n$-size  channel whose  agreement is better than its equality-leakage, we construct a (single-bit) key-agreement channel.  Our amplification protocol is stated below.

\begin{protocol}[$\Pi^{C}_\cH = (\Ac,\Bc)$]\label{protocol:KE}

	 \item Parameter: ensemble of function families  $\cH= \set{\cH_{n,m} = \set{h \colon \zo^n \mapsto \zo^m}}_{n,m \in \N}$.
	\item Inputs: $n,m  \in \bbN$.  
	
	\item Oracle: an $n$-size channel $\CXYT$.

	\item Operation:
	\begin{enumerate}
		\item The parties (jointly) call  $\CXYT$, where $\Ac$ gets $x$, $\Bc$ gets $y$ and $t$ is the common output. 
		
		\item \Ac samples $h \gets \cH_{n,m}$, $r\in\zo^n$  and sends $(h,h(x),r)$ to  \Bc.
		
		\item  $\Bc$ informs $\Ac$ whether $h(y)=h(x)$. 
		
		If positive, $\Ac$ outputs $o_\Ac= \ip{r,x} \bmod 2$, and  $\Bc$ outputs $o_\Bc= \ip{r,y} \bmod 2$.
		
		Otherwise, both parties aborts.
	\end{enumerate}
\end{protocol} 

It is clear that if the function family ensemble $\cH$ is efficient, \ie sampling and evaluation time is polynomial in $n$ and $m$, then so is $\Pi^C_\cH$.  For the security part, we prove that if $\cH_{n,m}$ is pairwise independent, then the protocol is a single-bit (weakly) secure key agreement. 
\begin{definition}[Pairwise independent hash functions]\label{def:PiHash}
	A function family $\cH=\set{h:\zo^n \to\zo^m}$  is {\sf pairwise independent} if for every $x_1\ne x_2\in\zo^n$ and $y_1,y_2\in\zo^m$, it holds that  $\ppr{h\gets \cH}{h(x_1)=y_1  \wedge h(x_2)=y_2}=2^{-2m}$.
\end{definition}
It is well-known. \cf \cite{vadhan2012pseudorandomness}, that efficient ensemble of pairwise independent hash functions exits.

The crux of out proof for \cref{thm:key-agreement-ampRes}  is in the next lemma.

\begin{lemma}[alphabet reduction]\label{lemma:key-agreement amplification}
	Let $\alpha \in (0,1]$, let $m\eqdef\lceil\log(1/\alpha)\rceil+8$, let  $C$ be an  $n$-size channel, and let $\hC$ denote the channel induced by a random execution of  $\Pi^C_\cH(n,m)$ conditioned on non abort. If $C$ is a $(\alpha,\alpha/2^{15})$-key agreement with equality-leakage, and $\cH_{n,m}$ is pairwise independent, then  $\hC$  is $(0.9,0.8)$-key-agreement-with-equality-leakage.
	
		Furthermore, the security proof is black-box: there exists an oracle-aided  \Ec such that  for every $n$-size channel $C$ with $\alpha$-agreement, and   an algorithm \tE violating  the equality-leakage of $\hC$, algorithm  $\Ec^{C,\tE}(n,m,\alpha,\cdot)$ violates the equality-leakage of  $C$ and runs in time $\poly(n,m,1/\alpha)$.
	
\end{lemma}

We prove \cref{lemma:key-agreement amplification} below, but first use it for proving \cref{thm:key-agreement-ampRes}.

\paragraph{Proving \cref{thm:key-agreement-ampRes}.}
\begin{proof}[Proof of \cref{thm:key-agreement-ampRes}]
	
	Let $\cH$ be an efficient ensemble of pairwise independent hash families, and let $\Pi_{\cH}$ be the oracle-aided protocol from \cref{protocol:KE}. Let $\hPi$ be the protocol that given oracle access to an $n$-size channel $C$, and inputs $\kappa,\alpha$, sets $m=\lceil\log(\ceil{1/\alpha})\rceil+8$  and does the following: the parties repetitively interact in $\Pi^C_\cH(n,m)$ until not abort, up to  $5/\alpha$ fail attempts. The parties output their output in the no aborting execution of $\Pi^C_{\cH}$, or $0$ if all executions have aborted.

	Let $C'$ be the channel induced by a random execution of $\hPi$.  By \cref{lemma:key-agreement amplification}, if $C$ has $\alpha$-agreement, $C'$ has agreement at least $0.9$.  Let $B$ be the event that  all attempts made by the parties have failed. Then by construction of $\hPi$ and the agreement of $C$, it holds that 
	$$\pr{B} \leq (1-\alpha)^{5/\alpha} \leq e^{-5}$$
	Assuming that $C$ is  also $\alpha/2^{15}$-secure  with equality-leakage, then by \cref{lemma:key-agreement amplification} and the above bound on $\pr{B}$, $C'$  is $0.8+e^{-5}< 0.81$-secure  with equality-leakage. Hence,   by applying known   amplification  for single-bit channels,  in particular,   applying   \cref{thm:KaAmp:Hol} on  $C'$  with parameters $\alpha=0.9$ and $\delta=0.81$, and input $1^\kappa$, we get the required  key-agreement protocol $\Phi$.
	
	Finally, we note that since both \cref{lemma:key-agreement amplification,thm:KaAmp:Hol} have black-box security reductions, then so is the security of $\Phi$.  
\end{proof}

\subsection{Proving  \cref{lemma:key-agreement amplification}}
In this section we prove \cref{lemma:key-agreement amplification}. We  make use of a weak version of the Goldreich-Levin theorem \cite{GoldreichL89}.
\begin{theorem}[Goldreich-Levin, \cite{GoldreichL89}]\label{Thm:GL}
	There exists an oracle-aided \ppt algorithm $\Egl$ such that the following holds: for every $n\in\N$, algorithm $\Dc:\zo^n\to \zo$, and  $x\in \zn$ that satisfy
	$$\ppr{r\gets \zo^n}{\Dc(r)=(\ip{x,r}\mod 2)}\ge3/4+0.01,$$
	it holds that $\pr{\Egl^\Dc=x}\ge 0.99$. 
\end{theorem}

In the rest of this section we prove \cref{lemma:key-agreement amplification}. Let $\alpha,s,m,\cH, \Pi_{\cH}^{C}$ be as in \cref{lemma:key-agreement amplification}. 
We associate the following random variables with a random execution of $\Pi_{\cH}^{C}(n,m)$.  Let $(X,Y,T)$ be the  output of the call to $\CXYT$ done by the parties,  let $R$ and $H$ be the value of $r$ and $h$ sent in the execution, and let $\OA, \OB$  be the local outputs of \Ac and \Bc, respectively. Let  $\NoAbort$ be the event that the parties did not abort during the execution, and   $T'\eqdef(T,H,H(\BX),R)$. Note that conditioned on  $\NoAbort$, $T'$ fully describes the transcript of the protocol. Finally,  let $\hT=(T,H,H(\BX))$ denote the prefix of $T'$ (without the randomness $R$) such that $T'=(\hT,R)$.

We will make use of the following claims:

The first claim bounds the agreement probability under the event $\NoAbort$.
\begin{claim}\label{Claim:Agreemet}
	$\pr{\OA=\OB\mid \NoAbort}\ge \alpha/(\alpha +2^{-m}).$ 
\end{claim}

The second claim  essentially  bounds the leakage of the protocol. This is done with a reduction to the security of $C$.
\begin{claim}\label{claim:ka_security}
	There exists an oracle-aided algorithm $\Eka$ such that the following holds. For every algorithm $\Dc\colon \Supp(T') \to \zo$ such that $\pr{\Dc(T')=\OA|\OA=\OB,\NoAbort}> 0.8$, it holds that $\pr{\Eka^\Dc(T)=\BX|\BX=\BY}>\alpha\cdot 2^{-15}$.	
\end{claim}

We prove \cref{Claim:Agreemet,claim:ka_security} below, but first we use \cref{Claim:Agreemet,claim:ka_security} in order to prove \cref{lemma:key-agreement amplification}, which is now  follows immediately.
\begin{proof}[Proof of \cref{lemma:key-agreement amplification}]
	Recall that 
	\begin{align}\label{eq:value of k}
		m=\lceil\log(1/\alpha)\rceil+8
	\end{align}
	Thus, by \cref{Claim:Agreemet} it follows that,
	\begin{align*}
		\pr{\OA=\OB\mid \NoAbort}&\ge\alpha/(\alpha +2^{-m})\ge 1/(1 +2^{-8})>0.9.
	\end{align*}
	By \cref{claim:ka_security} and the $\alpha\cdot 2^{-15}$-secrecy with equality-leakage of $C$, we get that $\hC$ is $0.8$-secure with equality-leakage.
	Since \cref{claim:ka_security} is a reduction, the lemma holds.
\end{proof}

We now prove \cref{Claim:Agreemet,claim:ka_security}.

\paragraph{Proving \cref{Claim:Agreemet}.} 
We start with the proof of \cref{Claim:Agreemet}.

\begin{proof}[Proof of \cref{Claim:Agreemet}]
	By construction we get that,
	\begin{align}\label{eq:agreement:1}
		\pr{\OA=\OB|\NoAbort}
		&=\pr{\OA=\OB|H(\BX)=H(\BY)} 
		\ge\pr{\BX=\BY|H(\BX)=H(\BY)}.
	\end{align}
	Let $\beta \eqdef \pr{\BX=\BY} \geq \alpha$, and note that since $\cH$ is pairwise independent, it holds that for $\pr{H(\BX)=H(\BY)\mid \BX \ne \BY}=2^{-m}$. Thus, it holds that,
	\begin{align}\label{eq:agreement:2}
		\pr{\BX=\BY|H(\BX)=H(Y)}&=\frac{\pr{\BX=\BY}}{\pr{H(\BX)=H(\BY)}}\\ \nonumber
		&= \frac{\pr{\BX=\BY}}{\pr{\BX=\BY}+\pr{H(\BX)=H(\BY),\BX\ne \BY}}\\ \nonumber
		&= \frac{\pr{\BX=\BY}}{\pr{\BX=\BY}+\pr{\BX\ne \BY}\cdot 2^{-m}}\\ \nonumber
		&= \frac{\beta}{\beta+(1-\beta)\cdot 2^{-m}}\\ \nonumber
		&\geq \frac{\beta}{\beta+ 2^{-m}}\\ \nonumber
		&\geq \alpha/(\alpha+2^{-m}),
	\end{align}
	where the last inequality holds since $x/(x+2^{-m})$ is a monotonic increasing function for $x\ge 0$. We conclude the claim by combining \cref{eq:agreement:1,eq:agreement:2}.
\end{proof}

\paragraph{Proving \cref{claim:ka_security}.}
To prove \cref{claim:ka_security}, we will use the next two claims. The first claim will be useful in order to bound the probability of an adversary to guess $X$, after seeing (part of) the transcript of  \cref{protocol:KE}.

\begin{claim}\label{Claim:Predicting BA}
	There exists an oracle-aided \pptm $\Ec$ such that the following holds. For every algorithm $\hEc:\Supp(\hT)\to\zo^n$ such that $\pr{\hEc(\hT)=\BX|\BX=\BY}\geq 2^m\cdot\alpha\cdot  2^{-15}$, it holds that $\pr{\Ec^{\hEc}(1^n,1^m,T)=\BX|\BX=\BY}\geq \alpha\cdot 2^{-15}$
\end{claim}

The second claim bounds the probability that $X\neq Y$ under the event that the parties agreed on the output.
\begin{claim}\label{claim: BA=BB| OA=OB}
	$ \pr{\BX\ne\BY|\OA=\OB,\NoAbort}<2^{-m}/\alpha.$
\end{claim}
\begin{proof}[Proof of \cref{claim: BA=BB| OA=OB}]
	\begin{align*}
		\pr{\BX\ne\BY|\OA=\OB,\NoAbort}
		&=\pr{\BX\ne\BY|\OA=\OB,H(\BX)=H(\BY)}\\
		&\le\frac{\pr{\OA=\OB,H(\BX)=H(\BY)| \BX\ne\BY}}{\pr{\OA=\OB,H(\BX)=H(\BY)}}\\
		&
		\le \frac{2^{-m}}{\alpha},
	\end{align*}  
	where the last inequality holds since $\cH$ is pairwise independent hash function, and $\pr{\OA=\OB,H(\BX)=H(\BY)} \geq \pr{\BX=\BY}=\alpha$.
\end{proof}	

We prove \cref{Claim:Predicting BA} below, but first we use it in order to prove \cref{claim:ka_security}.


\begin{proof}[Proof of  \cref{claim:ka_security}]
	Let \Dc be as in  \cref{claim:ka_security}. That is,
	\begin{align}\label{eq:assum}
		\pr{\Dc(T')=\OA|\NoAbort, \OA=\OB}> 0.8
	\end{align} 
	Since the event $\set{\BX=\BY}$ implies the event $\set{\NoAbort,\OA=\OB}$, it holds that,
		\begin{align*}
			\pr{\Dc(T')=\OA|\NoAbort,\OA=\OB}&\le \pr{\Dc(T')=\OA|\BX=\BY} +\pr{\BX\ne\BY|\NoAbort,\OA=\OB}\\
			&\leq \pr{\Dc(T')=\OA|\BX=\BY} + 2^{-m}/\alpha,
		\end{align*}
		where the second inequality holds by \cref{claim: BA=BB| OA=OB}. Thus, by our choice of $m$, we get that
	\begin{align}\label{eq:assum >79/100}
		\pr{\Dc(T')=\OA|\BX=\BY}>0.79
	\end{align} 
	Recall that $T'=(\hT,R)$. It follows by construction that conditioned on the event $\{\BX=\BY\}$, the randomness $R$ is uniform and independent of $\hT$. We now define the set of ``good transcripts" $\cG$  for the algorithm $\Dc$: 
	$$\cG=\set{\htt\in\Supp(\hT)\colon \pr{\Dc(\htt,R)=\OA|\BX=\BY,\hT=\htt}\ge 3/4+0.01}.$$
	
	We next show by an averaging argument over \cref{eq:assum >79/100} that
	$\pr{\hT\in \cG\mid \BX=\BY}\ge1/8$.
	
	Indeed, assume for contradiction that  $\pr{\hT\in \cG\mid \BX=\BY}<1/8$. Then
	\begin{align*}
		\pr{\Dc(T')=\OA|\BX=\BY}&=\pr{\Dc(T')=\OA|\BX=\BY,\hT\in\cG}\cdot\pr{\hT\in\cG|\BX=\BY}\\
		&+\pr{\Dc(T')=\OA|\BX=\BY,\hT\notin\cG}\cdot\pr{\hT\notin\cG|\BX=\BY}\\
		&<1\cdot 1/8 +(3/4+0.01)\cdot(1-1/8)=0.79.
	\end{align*}
	Which is a contradiction to \cref{eq:assum >79/100}.
	
	Let $\Egl$ be the algorithm promised by \cref{Thm:GL}, and for $\htt\in \Supp(\hT)$, let $\Dc_{\htt}(r)\eqdef \Dc(\htt,r)$. It follows by   \cref{Thm:GL} and the definition of $\cG$   that,
	
	$$\pr{\Egl^{\Dts}=\BX|\BX=\BY,\hT^*\in\cG}\ge 0.99.$$   
	Combining the above, we get that,
	\begin{align}\label{eq:gl}
		\pr{\Egl^\Dt=\BX|\BX=\BY}&\ge\pr{\Egl^\Dt=\BX|\BX=\BY, \hT\in\cG}\cdot\pr{\hT\in\cG|\BX=\BY} \\
		&\ge 0.99\cdot 1/8\nonumber\\
		&\ge 0.1\nonumber\\
		&> 2^m\cdot \alpha \cdot 2^{-15},\nonumber
	\end{align}
	where the last inequality holds by the choise of $m$. Finally, let $\hEc^{\Dc}(\htt)\eqdef \Egl^{\Dc_{\htt}}$ and let $\Ec$ be the algorithm promised by \cref{Claim:Predicting BA}. By \cref{eq:gl} and \cref{Claim:Predicting BA} we get that $ \pr{\Ec^{\hEc^{\Dc}}(1^n,1^m,T)=\BX|\BX=\BY}>\alpha\cdot 2^{-15}$. Thus the claim holds with respect to $\Eka^{(\cdot)}\eqdef\Ec^{\hEc^{(\cdot)}}$.

\end{proof}

\paragraph{Proving  \cref{Claim:Predicting BA}.}
In order to prove \cref{Claim:Predicting BA}, consider the following algorithm.
\begin{algorithm}[$\Ec^{\hEc}$] \label{Alg:with hash}
	\item Input: $1^n$, $1^m$, $t\in \Supp(T)$.
	\item Oracle: An algorithm, $\hEc:\Supp(\hT)\rightarrow \zo^n$.
	\item Operation:
	\begin{enumerate}
		\item Sample $h\gets \cH_{n,m}$, $v\from \zo^{m}$ 
		\item Output $\hEc(t,h,v)$.   
	\end{enumerate}
	
\end{algorithm}

\begin{proof}[Proof of \cref{Claim:Predicting BA}]
	Let $\hEc$ be as in \cref{Claim:Predicting BA}. That is, 
	\begin{align}\label{eq:hE(hT)=BA}
		\pr{\hEc(\hT)=\BX|\BX=\BY}>2^{m}\cdot\alpha/2^{15} .
	\end{align}
	We show that the above inequality implies that the algorithm $\Ec_{n,m}\eqdef\Ec^{\hEc}(1^n,1^m,\cdot)$ defined in $\cref{Alg:with hash}$ fulfills the requirement of the claim.
	That is, we want to show that:
	\begin{align}\label{eq:E(T)=BA}
		\pr{\Ec(T)=\BX|\BX=\BY}>\alpha/2^{15}.
	\end{align}
	First note that by the construction, $\Ec_{n,m}(T)=\hEc(T,H,V)$ (where $H$ and $V$ are sampled independently at random) and  recall that by definition $\hT=(T,H,H(\BX))$. Consequently, 
	$\hEc(T,H,V)|_{V=H(\BX)}\equiv \hEc(\hT)$ and  $\pr{H(\BX)=V|\BX=\BY}=2^{-m}$, and it follows that:
	\begin{align*}
		\pr{\Ec_{n,m}(T)=\BX|\BX=\BY}
		&=\pr{\hEc(T,H,V)=\BX|\BX=\BY}\\ 
		&\ge\pr{\hEc(T,H,V)=\BX|\BX=\BY,H(\BX)=V}\cdot\pr{H(\BX)=V|\BX=\BY}\\ 
		&\ge \pr{\hEc(\hT)=\BX|\BX=\BY}\cdot 2^{-m},\\
		&> 2^m\cdot\alpha/2^{15}\cdot 2^{-m}\\
		&=\alpha/2^{15}.
	\end{align*}
	Where the inequality follows by \cref{eq:hE(hT)=BA}.
	Thus \cref{eq:E(T)=BA} holds.
\end{proof}

\remove{
\subsection{The Computational  Case}\label{sec:KAApmlification:Comp}

\Inote{state thorem}

\begin{theorem}
	There exists an oracle-aided protocol $\tPi$ such that the following holds for every $\alpha\in (0,1]$. Let  $\Gamma$ be a $s$-size, $(\alpha,\alpha/2^{15})$-key-agreement-with-equality-leakage protocol. Then  $\tPi^\Gamma(1^\kappa,\alpha)$  is a single-bit,  $(1-2^\kappa,1/2+2^\kappa)$-key agreement protocol. The running time of  $\tPi^\Gamma$ is $\poly(s,\kappa,1/\alpha)$. 
\end{theorem}
\begin{proof}
	Similar to the proof of  \cref{thm:key-agreement-amp} from  \cref{lemma:key-agreement amplification}, using \cref{Thm:keyAgreementAmplificationHol:Com} instead \cref{Thm:keyAgreementAmplificationHol:IT}.
\end{proof}

\begin{lemma}[key-agreement amplification]\label{lemma:key-agreement amplification:comp}
	Let $0<\alpha\le 1$ and $s, m\in\N$ and be parameters such that $m=\lceil\log(1/\alpha)\rceil+8$. Let $\cH=\set{h:\zo^s \to\zo^m}$ be an efficient pairwise independent hash function
	and let  $\Gamma$ be a protocol of size $s$\Nnote{to myself: check}. 
	Let $\tPi_{Amp}^{\Gamma,\cH,n,m}$ denote the protocol of size $1$ induced by \cref{protocol:KE} and let  $\NoAbort$ be the event that both parties did not abort on a random execution of $\tPi_{Amp}^{\Gamma,\cH,n,m}$. If $\Gamma$ is a $(\alpha,\alpha/2^{15})$-key agreement-with-equality leakage, then the channel $\widetilde{C}$ induced by $\tPi_{Amp}^{\Gamma,H,n,m}|_{\NoAbort}$  of size $1$ is a $(9/10,8/10)$-key agreement-with-equality leakage.
\end{lemma} 

\begin{proof}
	Let $n,k,\cH$ and $\Gamma$ be as in \cref{lemma:key-agreement amplification:comp}. Let $C=\set{C_\kappa\eqdef X_\kappa Y_\kappa T_\kappa}_{\kappa\in N}$ be the channel ensemble associated with the protocol $\Gamma$.	For $\kappa\in\N$, protocol $\tPi(1^\kappa)$ does the following: It interacts in protocol $\Gamma(1^\kappa)$ and creates the associated channel $C_\kappa$. For every $C_\kappa$, $\tPi(1^n)$ interacts in protocol $\Pi^{C_\kappa,\cH,n,k}(1^n)$. 
	
	The proof now goes with the same lines of \cref{lemma:key-agreement amplification}. Fix $\kappa\in \N$ \Nnote{To mysels: large enough? check definitions.} and   let $(X,Y,T)$ be the  output of the call to $C_\kappa$ done by the parties,  let $R$ and $H$ be the value of $r$ and $h$ sent in the execution, let $\OA, \OB$ and $T'\eqdef(T,S,H(\BX),R))$,  be the local outputs of \Ac and \Bc, respectively and protocols transcript. 
	
	By \cref{Claim:Agreemet} it follows that,
	\begin{align*}
		\pr{\OA=\OB\mid \NoAbort}&\ge\alpha/(\alpha +(\half)^k)\ge 1/(1 +2^{-8})>9/10.
	\end{align*}
	Therefore it remains to show that for every algorithm $\Ec':\Supp(T')\to \zo$, $\pr{\Ec'(T')=\OA|\NoAbort, \OA=\OB}\le 8/10.$
	
	Indeed, since the event $\set{\BX=\BY}$ implies the event $\set{\NoAbort,\OA=\OB}$, it holds that,
	\begin{align*}
		\pr{\Ec'(T')=\OA|\NoAbort,\OA=\OB}&\le \pr{\Ec'(T')=\OA|\BX=\BY} +\pr{\BX\ne\BY|\NoAbort,\OA=\OB}\\
		&\leq 79/100 + 2^{-k}/\alpha\\
		&< 8/10,
	\end{align*}
	where the second inequality holds by \cref{claim: BA=BB| OA=OB,claim:ka_security}, together with the assumption on $C$. The last by the choice of $k$.
\end{proof}

}

%% file: MissingProofs.tex
\section{Missing Proofs}\label{sec:appendix}

\subsection{Missing Proofs from \cref{sec:Preliminaries}}
\subsubsection{Proving \cref{proposition:exp-of-abs}}\label{sec:missing-proofs:proving-exp-of-abs}

In this section, we prove \cref{proposition:exp-of-abs}, restated below.

\begin{proposition}
	\propExpOfAbs
\end{proposition}

Throughout this section, we let $\cN(0,1)$ be the standard normal distribution with probability density function $\phi(z) = \frac1{\sqrt{2\pi}} \exp\paren{-\frac{z^2}2}$, and we let $S_n$ be the sum of $n$ i.i.d.\ random variables, each takes $1$ w.p. $1/2$ and $-1$ otherwise.
We use the following facts: 

\begin{fact}[\cite{barr1999mean}]\label{fact:truncated-norm-dist}
	Let $Z \gets \cN(0,1)$. Then for every $t \in \bbR:$
	\begin{align*}
		\pr{Z > t}\cdot \ex{Z \mid Z > t} = \phi(t).
	\end{align*}
\end{fact}

\begin{fact}[Nonuniform Berry-Esseen bound \cite{nagaev1965some}]\label{fact:Berry-Esseen}
	Let $X \gets S_n$, and let $Z \gets \cN(0,1)$. Then for every $t \in \bbR:$
	\begin{align*}
		\pr{X > t \sqrt{n}} - \pr{Z > t} = O\paren{\frac{1}{(1 + \size{t}^3) \sqrt{n}}}
	\end{align*}
\end{fact}

\begin{fact}\label{fact:exp-to-int}
	For any random variable $X$ over $\bbR^{+}$, it holds that
	\begin{align*}
		\ex{X} = \int_{0}^{\infty} \pr{X > x} dx
	\end{align*}
\end{fact}


The proof of \cref{proposition:exp-of-abs} immediately follows by the following proposition.

\begin{proposition}
	Let $n \in \N$ be larger than some universal constant, and let $X \la S_n$.   
	Then for every $t  > 0:$
	\begin{align*}
		\pr{X > t}\cdot \ex{X \mid X > t} \leq  2\sqrt{n}
	\end{align*}
\end{proposition}
\begin{proof}
	Compute
	\begin{align*}
		\pr{X > t}\cdot \ex{X \mid X > t} 
		&= \pr{X > t} \cdot   \int_{0}^{\infty} \pr{X > x \mid X > t} dx\\
		&= \int_{0}^{\infty} \pr{X > \max\set{x,t}} dx\\
		&= t \cdot \pr{X > t} + \int_{t}^{\infty} \pr{X > x} dx\\
		&\leq t \cdot \exp\paren{-\frac{t^2}{2n}}+ \sqrt{n} \cdot \int_{t/\sqrt{n}}^{\infty} \pr{X > z \sqrt{n}} dz\\
		&\leq e^{-1/2} \cdot \sqrt{n} + \sqrt{n} \cdot \int_{t/\sqrt{n}}^{\infty} \pr{Z > z} dz + O\paren{1}\\
		&\leq e^{-1/2} \cdot \sqrt{n} + \frac12 \cdot \sqrt{n} \cdot \int_{0}^{\infty} \pr{Z > z \mid Z > 0} dz + O\paren{1}\\
		&= e^{-1/2} \cdot \sqrt{n} + \frac12 \sqrt{n} \cdot \ex{Z \mid Z > 0} + O\paren{1}\\
		&\leq 2 \sqrt{n}.
	\end{align*}
	The first equality holds by \cref{fact:exp-to-int}. The first inequality holds by Hoeffding's inequality (\cref{fact:Hoeff}) along with the variable substitution $x = z \sqrt{n}$ in the integral. The second inequality holds by \cref{fact:Berry-Esseen} along with the fact that $t \cdot \exp\paren{-\frac{t^2}{2n}} \leq e^{-1/2} \cdot \sqrt{n}$ for every $t$. The third inequality holds since $Z$ is symmetric around $0$.  The last equality holds by \cref{fact:exp-to-int}, and the last inequality holds by \cref{fact:truncated-norm-dist} which implies that $\ex{Z \mid Z > 0} = 2 \phi(0) = \sqrt{\frac2{\pi}}$.
\end{proof}

\subsubsection{Proving \cref{lemma:boundMultDist}}\label{sec:missing-proofs:proving-boundMultDist}
In this section we prove \cref{lemma:boundMultDist}, restated below.
\begin{proposition}
	\propBoundMultDist
\end{proposition}

In the following, let $H$ be the Entropy function. That is, for a random variable $X$, $H(X)=-\sum_{x\in \Supp(X)}\log (\pr{X=x})$. We will use the following facts about $H$: 
\begin{fact}
	\label{lemma:prem:entropy_close_to_one}[Entropy upper bound, \cite{calabro2009exponential}]
	Let $X$ be a random variable supported on $\zo$, and let $q \in [0,1]$. Assume $H(x)\geq 1-q^2$. Then $\pr{X=1} \in [1/2-q, 1/2+q]$.
\end{fact}

\begin{fact}\label{fact:entropy:sum}
Let $X=(X_1,\dots ,X_n)$ be a random variable. Then $H(X)\leq \sum_i H(X_i)$.
\end{fact}

\begin{fact}\label{fact:entropy:cond}
	Let $X$ be a random variable and  let $E$ be an event. Then $H(X|E)\ge H(X)+\log(\pr{E})$.
\end{fact}

\begin{proof}[Proof of  \cref{lemma:boundMultDist}.]
	We first show that for every $q\in [0,1]$, it holds that
	\begin{align}\label{eq:boundMultDist:1}
	\ppr{i \gets [n]}{H(R_i|E)\ge 1-q^2}\ge 1- (\log n)/(n\cdot q^2).
	\end{align}
To see this, let $\cB =\set{i: H(R_i|E) < 1-q^2}$. We want to show that $|\cB| \leq (\log n)/q^2$.  Indeed, assume toward contradiction this is not the case. Then 
\begin{align}\label{eq:boundMultDist:2}
\sum_i H(R_i|E) < |\cB|\cdot (1-q^2) + (n-|B|)\cdot 1  = n -q^2 \cdot |\cB| \le n-\log n.
\end{align}
On the other hand, using \cref{fact:entropy:sum,fact:entropy:cond} we get that
\begin{align}
\sum_i H(R_i|E) \ge H(R|E) \geq H(R)-\log (1/\ppr{R}{E})\geq n-\log n
\end{align}
which contradicts \cref{eq:boundMultDist:2}, and thus \cref{eq:boundMultDist:1} holds.

Next, we show that for every $i\in [n]$ with $H(R_i|E)\ge 1-q^2$ it holds that $\pr{E\mid R_i =b} \in (1\pm 2q)\pr{E}$, which concludes the proof. Indeed, by \cref{lemma:prem:entropy_close_to_one}, it holds that for every $b\in \zo$, $\pr{R_i=b\mid E} \in 1/2 \pm q$.  Applying Bayes rule, we get that 
\begin{align*}
\pr{E\mid R_i =b} = \frac{\pr{E}\cdot\pr{R_i=b\mid E}}{\pr{R_i=b}}=2\pr{E}\cdot\pr{R_i=b\mid E} \in  (1\pm 2q)\pr{E}
\end{align*}
as we wanted to show.	
\end{proof}

\subsubsection{Proving \cref{prop:Raz}}\label{sec:missing-proofs:proving-Raz}
In this section we prove \cref{prop:Raz}, restated below.

\begin{proposition}
	\propRaz
\end{proposition}

To prove \cref{prop:Raz}, we will use the following simple lemma:
\begin{lemma}\label{prel:stan_dev_to_var}
	Let $X$ ba a random variable. Then $\ex{\size{X-\ex{X}}} \leq \sqrt{\Var(X)}.$
\end{lemma}
\begin{proof}
	Recall that $\Var(X)= \ex{(X-\ex{X})^2}$. So, by taking $Y \eqdef \size{X-\ex{X}}$, it is enough to show that $\ex{Y} \leq \sqrt{\ex{Y^2}}$.
	Since $Y$ is positive, the above is equivalent to   $\ex{Y}^2 \leq \ex{Y^2}$.
	Recall that 
	\begin{align*}
		0 \leq \Var(Y) = \ex{Y^2}-\ex{Y}^2
	\end{align*}
	which ends the proof.
\end{proof}

We are now ready to prove \cref{prop:Raz}.

\begin{proof}[Proof of \cref{prop:Raz}.]
	For a vector $r\in \zn$, let $1(r)\eqdef \sum_{i\in \cI} r_i$ be the number of $1$'s that $r$ has in $\cI$, and $0(r) \eqdef \size{\cI} - 1(r)$ be the number of $0$'s that $r$ has in $\cI$. By definition of statistical distance,
	\begin{align*}
		SD(R|_{R_I=1},R|_{R_I=0})&= 1/2 \cdot \sum_{ r\in \zn} \size{\pr{R=r \mid R_I=1} - \pr{R=r \mid R_I=0}}\\
		& = 1/2 \cdot \sum_{ r\in \zn} \size{\frac{\pr{R=r, R_I=1}}{\pr{R_I=1}} - \frac{\pr{R=r, R_I=0}}{\pr{R_I=0}}}\\
		& = 1/2 \cdot \sum_{ r\in \zn} 2^{-n}\cdot \size{\frac{\pr{ R_I=1 \mid R=r}}{1/2} - \frac{\pr{R_I=0 \mid R=r}}{1/2}}\\
		& = \sum_{ r\in \zn} 2^{-n}\cdot \size{\pr{ R_I=1 \mid R=r} - \pr{R_I=0 \mid R=r}}\\
		& = \sum_{ r\in \zn} \frac{2^{-n}}{\size{\cI}}\cdot \size{1(r) - 0(r)}\\
		& = \frac{2}{\size{\cI}} \cdot \eex{r \gets R}{\size{1(r)-\size{\cI}/2}},\\
	\end{align*}
	where the last equality holds since  $\size{1(r) - 0(r)} = 2\cdot \size{1(r)-\size{\cI}/2}$.
	
	Notice that $\ex{1(r)} = \size{\cI}/2$ and $\Var(1(r))= \size{\cI}/4$. Thus, by \cref{prel:stan_dev_to_var} we conclude that
	\begin{align*}
		SD(R|_{R_I=1},R|_{R_I=0})& \leq 2/\size{\cI} \cdot \sqrt{\size{\cI}/4}\\
		& = 1/\sqrt{\size{\cI}}.\\
	\end{align*}
\end{proof}

\subsection{Missing Proofs from \cref{sec:CondensingSV}}\label{appendix:cond}

\subsubsection{Proving \cref{cor:CondensingSV}}\label{sec:missing-proofs:CondensingSV}
In this section we prove \cref{cor:CondensingSV}, restated below.
\begin{corollary}
	\defCondCor
\end{corollary}
\begin{proof}
		Let$(X,Y)$, $R$, $\eps$ and $\delta$ be as in \cref{cor:CondensingSV}. Let  $c_1$ and  $c_2$  be as in \cref{thm:CondensingSVRes}$, \ell \eqdef \log n$, $t=\bot$, $D \eqdef (X,Y,\bot)$ and let $f$ be the function defined by 
		\begin{align*}
		f(r,x_{r^+}, y_{r^-},t) \eqdef \argmax_{c \in [n]} \set{\ppr{(x,y,r)\gets (X,Y,R)}{\Cc(X,Y,R)=c \mid (R,X_{R^+},Y_{R^-})=(r,x_{r^+},y_{r^-})}}
		\end{align*}
		Let $Z=(R,X,Y)$ and $Z_r = (R,X_{R^+},Y_{R^-})$. By \cref{thm:CondensingSV,prop:breaking the dp}, it holds that $\ppr{(x,y,t)\gets D, r\gets\mon}{f(\transF)=\Cc(x,y,r)}< e^{c_1\eps}\cdot c_2\cdot \ell/\sqrt{n}$. Thus,
		\begin{align*}
		\eex{z_r\gets Z_r}{\ppr{z\gets Z|_{Z_r=z_r}}{f(\transF)=\Cc(x,y,r) }} < e^{c_1\eps}\cdot c_2\cdot \ell/\sqrt{n}
		\end{align*}
		By the Markov inequality, we get that 
		\begin{align*}
		\ppr{z_r\gets Z_r}{\ppr{z\gets Z|_{Z_r=z_r}}{f(\transF)=\Cc(x,y,r) } \geq 1/\delta \cdot e^{c_1\eps}\cdot c_2\cdot \ell/\sqrt{n} }\le \delta 
		\end{align*}
		and by the definition of $f$, it holds that
		\begin{align*}
		\ppr{(r,x_{r^+},y_{r^-}) \gets Z_r}{\max_{c \in [n]} \set{\ppr{(x,y,t)\gets D, r\gets\mon}{\Cc(x,y,r)=c \mid r, x_{r^+} ,y_{r^-}}} \geq 1/\delta \cdot e^{c_1\eps}\cdot c_2\cdot \ell/\sqrt{n} }\le  \delta 
		\end{align*}
		The last implies by the definition of min-entropy that 
		\begin{align*}
		\ppr{z_r\gets Z_r}{\Hmin(C(X,Y,R) |_{Z_r=z_r}) \le \log\paren{\frac{\delta \sqrt{n}}{c_2 \cdot e^{c_1\eps} \log n}}}\le \delta 
		\end{align*}
		which ends the proof.

\end{proof}